\documentclass[11pt]{article}%

\usepackage{amsfonts}%
\usepackage{amsmath}%
\usepackage{amssymb}%
\usepackage{amsthm}%
\usepackage[english]{babel}%
\usepackage{bbm}%
\usepackage{bm}%

\usepackage{booktabs}%
\usepackage{caption}%
\usepackage{color}%
\usepackage{comment}%
\usepackage{enumerate}%
\usepackage{float}%
\usepackage[T1]{fontenc}%
\usepackage[margin=1in]{geometry}%
\usepackage{graphicx}%
\usepackage[breaklinks=true,bookmarksopen=true,colorlinks=true,citecolor=blue,urlcolor=blue,linkcolor = blue]{hyperref}%
\usepackage{listings}% 
\usepackage{mathtools}%
\usepackage[longnamesfirst]{natbib}%
\usepackage{paralist}%
\usepackage{pgfplotstable}%
\pgfplotsset{compat=1.12}%
\usepackage[nodisplayskipstretch]{setspace}%
\usepackage[subrefformat=parens]{subcaption}%
\usepackage{tabularx}%
\usepackage{verbatim}%
\usepackage{color}%
\newtheorem{theorem}{Theorem}[section]%
\newtheorem{assumption}{Assumption}%
\newtheorem{example}{Example}[section]%

\newtheorem{lemma}{Lemma}%

\newtheorem{proposition}{Proposition}[section]
\newtheorem{remark}{Remark}[section]%
\newtheorem{corollary}[proposition]{Corollary}%

\newcolumntype{C}{>{\centering\arraybackslash}X}%

\newcommand{\E}{\mathbb{E}}

\usetikzlibrary{patterns}

\makeatletter
\newcommand*{\indep}{%
  \mathbin{%
    \mathpalette{\@indep}{}%
  }%
}
\newcommand*{\nindep}{%
  \mathbin{%                   % The final symbol is a binary math operator
    \mathpalette{\@indep}{\not}% \mathpalette helps for the adaptation
  }%
}
\newcommand*{\@indep}[2]{%
  \sbox0{$#1\perp\m@th$}%        box 0 contains \perp symbol
  \sbox2{$#1=$}%                 box 2 for the height of =
  \sbox4{$#1\vcenter{}$}%        box 4 for the height of the math axis
  \rlap{\copy0}%                 first \perp
  \dimen@=\dimexpr\ht2-\ht4-.2pt\relax
  \kern\dimen@
  {#2}%
  \kern\dimen@
  \copy0 %                       second \perp
} 
\makeatother
\newcommand{\bmi}{\boldsymbol{\mathrm{i}}}

\usetikzlibrary{arrows}
\usetikzlibrary{decorations.markings}
\usetikzlibrary{intersections, pgfplots.fillbetween, patterns}

\pgfmathdeclarefunction{gauss}{2}{%
	\pgfmathparse{min(8,3/(#2*sqrt(2*pi))*exp(-((x-#1)^2)/(2*#2^2)))}%
}

\pgfmathdeclarefunction{gaussplus2}{2}{%
	\pgfmathparse{min(8,3/(#2*sqrt(2*pi))*exp(-(((x+2)-#1)^2)/(2*#2^2)))}%
}

\pgfmathdeclarefunction{clamp}{3}{% clamp(x,a,b)
  \pgfmathparse{min(#3,max(#2,#1))}
}
\pgfmathdeclarefunction{smoothstep}{1}{% 3t^2-2t^3
  \pgfmathparse{3*(#1)^2 - 2*(#1)^3}
}

\begin{document}
\begin{singlespace}
%+Title
\title{\textbf{Identification of Causal Effects with a Bunching Design}\thanks{We thank Stéphane Bonhomme, Guido Imbens and Elie Tamer for pushing us towards developing the ideas of an earlier paper in a new direction, which led to this paper. We also thank Luis Alvarez, David Card, Alfonso Flores-Lagunes, Stefan Hoderlein, Hugo Jales, Matthew Masten, Eric Mbakop, Whitney Newey, Joris Pinkse, Alexandre Poirier, Demian Pouzo, Karl Schurter, Chris Taber, Ot\'avio Techio, as well as seminar participants at several institutions and conferences for valuable help and feedback. The analysis and conclusions set forth here are those of the authors and do not indicate concurrence by other members of the research staff, the Board of Governors, or the Federal Reserve System.}}

\author{Carolina Caetano\\\textit{\normalsize University of Georgia}
\and Gregorio Caetano\\\textit{\normalsize University of Georgia}
\and Leonard Goff\\\textit{\normalsize University of Calgary}
\and Eric Nielsen\\\textit{\normalsize Federal Reserve Board}}

\date{}
\end{singlespace}

\maketitle

%-Title
%+Abstract

\vspace{-.5cm}
\begin{abstract}

\begin{singlespace} We show that causal effects can be identified when there is bunching in the distribution of a continuous treatment variable, without imposing any parametric assumptions. This yields a new nonparametric method for overcoming selection bias in the absence of instrumental variables, panel data, or other popular research designs for causal inference. The method leverages the change of variables theorem from integration theory, relating the selection bias to the ratio of the density of the treatment and the density of the part of the outcome that varies with confounders. At the bunching point, the treatment level is constant, so the variation in the outcomes is due entirely to unobservables, allowing us to identify the denominator. Our main result identifies the average causal response to the treatment among individuals who marginally select into the bunching point. We further show that under additional smoothness assumptions on the selection bias, treatment effects away from the bunching point may also be identified. We propose estimators based on standard software packages and apply the method to estimate the effect of maternal smoking during pregnancy on birth weight.
 
\end{singlespace}
\end{abstract}
%-Abstract
%\doublespacing
\onehalfspacing

\section{Introduction}

In this paper, we show that bunching in the distribution of a treatment variable can be used to identify causal effects. In particular, although the treatment may be endogenous, we do not rely on instrumental variables, regression discontinuity designs, panel data, functional form, or distributional assumptions. 

The setting is a standard causal model in which both the treatment and the outcome variables are observed. The treatment variable has a mass point and is continuously distributed near this bunching point. The outcome variable is continuously distributed near and at the bunching point. The example in our application is a useful benchmark:  the treatment is the number of cigarettes a woman smokes during pregnancy (with 81\% of the observations bunching at zero), and the outcome is the baby's birth weight.  \cite{caetano2015} showed strong evidence of endogeneity in this application. To identify treatment effects, we need to eliminate selection bias, the part of the outcome variation that is due to confounders. 

The key insight that makes identification possible is that the change-of-variables theorem from integration theory can be used to write the magnitude of selection bias as the ratio of the probability density of the treatment variable and the probability density of the part of the outcome that is due to confounders. As a consequence, we do not need to observe the values of those confounders; instead, it is sufficient to identify the distribution of the part of the outcome that varies with confounders. This is where bunching is useful: at the bunching point, the distribution of the outcome reflects only the variation of unobservables, since the treatment stays fixed. We use this to identify the selection bias. 

We identify the average per unit effect of marginally increasing treatment among those near the bunching point (equivalently, among the observations at the bunching point that are most similar to those near it). In our application, this quantity represents the expected rate of birth-weight loss if a woman who is currently not a smoker but is very similar to the women who smoke very little were to start smoking. Alternatively, minus this quantity can be interpreted as the expected rate of birth-weight gain if the women who currently smoke little were to quit smoking. 

The approach relies on four conditions. First, the treatment effects must be sufficiently smooth near the bunching point. In our application, this condition means that smoking is not so harmful that a marginal amount could (on average) cause a discrete birth-weight change. Second, selection must also be sufficiently smooth at the bunching point. In our application, this means that the mothers who smoke very little (``marginal smokers'') are comparable to the mothers who do not smoke, but are indifferent between not smoking and smoking a small positive amount (``marginal nonsmokers''). Third, the selection bias must maintain the same sign in a neighborhood near the bunching point. In our application, this means that if selection into smoking is negative among mothers who smoke little (i.e. smoking less is associated with higher untreated birth weights), then selecting into not smoking must be associated with even higher birth weights. 
Finally, the outcomes at the bunching point are determined not only by confounders, but also by idiosyncratic (i.e., unconfounded) variation, which we must then separate from the distribution of outcomes at the bunching point by deconvolution. By construction, the unconfounded variation is mean independent of the confounders, but the deconvolution step requires full independence at the bunching point (or alternatively the weaker subindependence condition in \cite{schennach2019convolution}).

If the function that characterizes selection is real analytic, then our results also allow for the identification of $\text{ATT}$s near the bunching point. If some bounds on the selection bias derivatives hold, then the $\text{ATT}$s can be identified further away from the bunching point. Specifically, we can identify the effect among those who take a given treatment value as compared to a counterfactual in which they take the bunching point value of treatment. In our application, it is thus possible to identify the birth-weight gains (or losses) if mothers who smoke a given amount were to quit smoking.  

We propose estimators that use well-known building blocks. We estimate expectations and derivatives near the bunching point with local linear estimators \citep{fangijbels1992} and boundary densities using the estimator of \cite{pinkse2023estimates}, both of which achieve interior rates of convergence at the boundary. The deconvolution step follows standard nonparametric methods, equivalent to a standard kernel density estimator using a special kernel. All building blocks can be implemented straightforwardly using off-the-shelf packaged software. 

In a supplementary appendix, we also explore how control variables may be used to study heterogeneous treatment effects as well as to weaken the identification assumptions, which are then required to hold only conditional on controls. In particular, this allows the sign of selection bias near the bunching point to differ across different groups of observations. We discuss estimation when controls are discrete, continuous, or a mixed vector of both. 

We apply our approach to data on smoking and birth weight from \cite{almond2005costs}. After correcting for selection bias, we find that the effect of the first daily cigarette is a small loss of about 8 grams in birth weight (less than one-third of an ounce). Smoking five cigarettes per day causes a loss of 1.4 ounces (compare this to the average weight of a full-term newborn, which is 122 ounces in our analysis sample). These estimates confirm and strengthen the qualitative point in \cite{almond2005costs} that smoking is not an important determinant of birth weight, but under different and weaker identification assumptions.

Bunching is a common phenomenon. It is often found at zero in variables naturally constrained to be non-negative, such as consumption goods,\footnote{E.g., number of tobacco products, alcoholic beverages, caffeinated drinks, sugary drinks, fast food meals, dining out meals, subscription services, supplements and vitamins, public transportation rides, books read, gym visits, doctor visits, trips, fuel usage amounts, expenditures on health, fitness, travel, vacations, education, childcare.} financial variables,\footnote{E.g., credit access, bequests, savings, emergency fund levels, retirement account contributions, mortgage balance, credit card debt, student loan debt, income from investments, expenditures on ads, charitable donations, HSA and FSA balances, life insurance coverage, number of trades.} time use,\footnote{Bunching is found for almost all time uses. Some examples include exercising, working, watching TV, using digital devices, doing homework, doing chores, volunteering, and commuting.} and neighborhood characteristics.\footnote{E.g., number of public transportation options or stops, retail stores, coffee shops, rental units, affordable housing units, vacant units, electric vehicle charging stations; length of biking lanes or walking paths; areas of green space, commercial districts, sports fields, parking lots.} Artificial constraints also can generate bunching, such as regulatory minimums\footnote {E.g., schooling time, wages, 401(k) contributions, coverage for auto insurance, nutritional standards for school meals, bank capital, bank deposit insurance, age started working, age started withdrawing from retirement accounts, age retired.} and maximums.\footnote{E.g., contribution size in 401(k), Roth IRA, HSA, FSA accounts, untaxed gifts, FHA loans, FDIC insurance, carbon emissions, liquor licenses, lot coverage, contributions to political campaigns, data usage, grades, absences from school, class size, commissions on sales.} Bunching also occurs at interior points, often due to kinks or notches in budget sets, social norms, and other restrictions.\footnote{E.g. income at tax brackets, hours worked at overtime rules, multiples of 5 or 10, 40 hours per week, car speeds at ticket thresholds, financial reporting around profit targets, energy consumption around utility billing tiers, pricing below psychological points (\$0.99), doctor visits at medical protocol numbers, hospital stay length at insurance payment thresholds.} Of course, small samples, coarse measurements, or attrition could make it impossible to implement this method in some of the examples above. 

The use of bunching phenomena for identification began with \cite{saez10}, followed by a large applied literature interested in identifying the effects of policies that lead to bunching in a manipulable variable at a policy threshold. Theoretical treatments of these approaches may be found in \cite{blomquist2021bunching}, \cite{bertanha2023better}, \cite{goff2020treatment} and \cite{lu2024identifying}. Our approach is more related to the literature initiated by \cite{caetano2015}, where bunching for any reason on the treatment variable of a reduced-form causal model allows the testing of the model's identification conditions (see \cite{caetano2016discontinuity}, \cite{caetano2018identifying}, \cite{CCFN_Dummy}, and \cite{khalil2022test}). \cite{ccn_metrics} developed the first strategy for identification of treatment effects under endogeneity in this setting, followed by \cite{CCNT}, \cite{CCT} and \cite{CMS_Placebo}.  Surveys of the bunching literature include \cite{kleven2016bunching,jalesyu,blomquist2023econometrics}, and \cite{bertanha2024bunching}. In this paper, we show for the first time that nonparametric identification with bunching is possible. 

Section \ref{sec:model} introduces our setting and presents a novel approach to identification of the average marginal treatment effect using the change-of-variables theorem from integration theory. We detail how the average marginal treatment effect at the bunching point can be identified in Section \ref{sec:idfrombunching}, and how treatment effects can then be identified away from the bunching point in Section \ref{sec:global}. We turn to estimation in Section \ref{sec:estimation}, and in Section \ref{sec:application} we present the application to the effects of smoking on birth weight. We conclude in Section \ref{sec:conclusion}. Appendices contain proofs, examples and generalizations referred to in the text. A supplementary appendix contains extensions for the use of control variables as well as further plots and proofs.

\section{Identification near the bunching point}\label{sec:model}

In this section, we set up the identification problem in a neighborhood of the bunching point. Then, we relate the selection bias to the ratio of the density of the treatment and the density of the part of the outcome that varies with confounders.

Our setting is the standard potential outcomes framework, where observation $i$'s outcome depends on the value $x$ of a multivalued scalar treatment variable through the potential outcome function $Y_i(x).$ We observe the treatment value $X_i$ and the outcome $Y_i=Y_i(X_i)$.  The support of $X_i$ includes a nondegenerate interval whose left boundary $\bar{x}$ exhibits bunching. In many relevant applications, $X_i$ is continuously distributed on the positive real line with the bunching point $\bar{x}=0$. The right panel of Figure \ref{fig:bunchingtypes} depicts this case, while the left panel depicts a case in which $\bar{x}$ is in the interior of the support of $X_i$. Bunching on the right boundary of the support of $X_i$ can be accommodated by redefining $X_i$ as $2\bar{x} - X_i$. 

\begin{figure}
\caption{Two examples of distributions with bunching} \label{fig:bunchingtypes}
\begin{center}
    \begin{tikzpicture}[scale = .95]
			\begin{axis}[
			axis x line=bottom,
			axis y line=left,
                xmin=0, xmax=10, 
			ymin=0, ymax=10,
			title={\textbf{Interior bunching}},
			ylabel={\scriptsize{$f_X(x)$}},
                ylabel style={
                at={(axis description cs:0,1)},
                rotate=-90,
                anchor=south,
                xshift=-4pt,
                },
                ytick=\empty,
			xtick=\empty,
			extra x ticks={0,5},
			extra x tick labels={0,$\bar{x}$},
			extra y ticks={\empty},
			extra y tick labels={\empty},
			]
			
		\addplot[blue,only marks,mark=*,mark size=2pt]
      coordinates {(5,8)};
      	\addplot [blue!30,fill=blue!30,draw=none,fill opacity=0.4,domain=0:5] {5*gauss(6.5,2)} \closedcycle;
		\addplot [blue!30,fill=blue!30,draw=none,fill opacity=0.4,domain=5:10] {5*gaussplus2(6.5,2)} \closedcycle;
            \addplot [blue, thick, domain=0:5] {5*gauss(6.5,2)};
            \addplot [blue, thick, domain=5:10] {5*gaussplus2(6.5,2)};
			\end{axis}		
    \end{tikzpicture}
    \quad \quad \quad \quad
    \begin{tikzpicture}[scale = .95]
			\begin{axis}[
			axis x line=bottom,
			axis y line=left,
			xmin=0, xmax=10, 
			ymin=0, ymax=10,
			title={\textbf{Boundary bunching}},
			ylabel={\scriptsize{$f_X(x)$}},
                ylabel style={
                at={(axis description cs:0,1)},
                rotate=-90,
                anchor=south,
                xshift=-4pt,
                },
			ytick=\empty,
			xtick=\empty,
			extra x ticks={0},
			extra x tick labels={$\bar{x}$},
			extra y ticks={\empty},
			extra y tick labels={\empty},
			]
			
        \addplot[blue,only marks,mark=*,mark size=2pt]
      coordinates {(0,8)};
	\addplot [blue, thick, domain=0:10] {5*gaussplus2(4.5,2)};
        \addplot[blue!30,fill=blue!30,draw=none,fill opacity=0.4,domain=0:10] {5*gaussplus2(4.5,2)} \closedcycle;
			\end{axis}		
    \end{tikzpicture}
\end{center}
\end{figure}
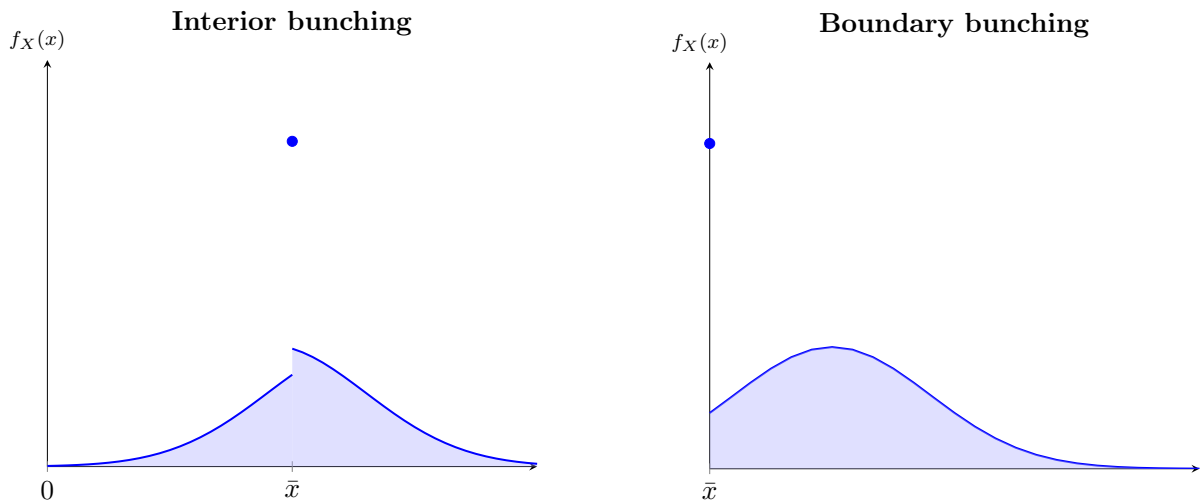

We adopt the following notational conventions. For an arbitrary function $v\mapsto g(v),$ we denote the $k$-th derivative at $\tilde{v}$ as $g^{(k)}(\tilde{v})$. For the first derivative, we also use the notation $g'(\tilde{v}):=g^{(1)}(\tilde{v}).$ For a function $g(x)$, let $g(\bar{x}^+):= \lim_{x \downarrow \bar{x}} g(x)$. We let $g'(\bar{x}^+):=\lim_{x\downarrow \bar{x}}(g(x)-g(\bar{x}^+))/(x-\bar{x})$ denote a right-derivative at $\bar{x}$ defined with respect to the limit $g(\bar{x}^+)$. We show that this definition is equivalent to $g'(\bar{x}^+)=\lim_{x\downarrow \bar{x}} g'(x)$, provided that $g(\bar{x}^+)$ exists and $g$ is differentiable in a neighborhood above $\bar{x}$. We define higher order limit derivatives $g^{(k)}(\bar{x}^+)$ analogously.  For other composite functions $(g\circ h)(x) := g(h(x))$, we let $g(h(\bar{x}^+)):= (g\circ h)(\bar{x}^+)$. For an arbitrary random variable $V_i,$ let $F_{V}(v)=\mathbb{P}(V_i\leq v)$ and $f_{V}(v)=F'_{V}(v)$ if the derivative exists. Analogously, for a set $S \subseteq \mathbbm{R},$ $F_{V|S}(v)=\mathbb{P}(V_i\leq v|S)$ and  $f_{V|S}=F'_{V|S}(v)$ if and where the derivative exists. For example, $f_{V|S}(x)$ denotes $\frac{f_V(v)}{P(S)}$ for any $v \in S$. Define the sign of $v$ as $\text{sgn}(v)=\bm{1}(v\geq 0)-\bm{1}(v\leq 0)$. For a set $S \subseteq \mathbbm{R}$, we denote the image of the function $g$ over $S$ as $g(S)$. For simplicity, we refer to the ``support of $V_i$'' when we mean the support of the distribution of $V_i$.\\  

\noindent \textit{Remark:} (Friction around $\bar{x}$) The examples in Figure \ref{fig:bunchingtypes} depict distributions with ``perfect'' bunching, in the sense that bunching is at a point mass in the distribution of $X_i$. By contrast, interior bunching is often somewhat diffuse around the bunching point because of optimization frictions, or because $X_i$ is measured with error. We abstract from this issue and assume that the researcher has a means of identifying the ``bunched'' observations $X_i=\bar{x}$. This is generally not problematic in settings with boundary bunching, and for interior bunching, measurement error can sometimes be eliminated by using administrative data (see \citealt{goff2020treatment} for an example). For general discussions of optimization frictions in bunching settings, see \citet{kleven2016bunching} and \citet{bertanha2024bunching}.

\subsection{Parameters of interest and the identification problem}
To define treatment effect parameters, we let the outcome $Y_i(\bar{x})$ that would occur if treatment were equal to $\bar{x}$ play the role of the ``untreated'' state. When bunching is at zero, this yields the familiar notation $Y_i(0)$. Let $\text{ATT}(x)$ denote the average effect of changing treatment from the bunching point to $x$, among those treated $X_i=x$:
$$\text{ATT}(x):=\mathbbm{E}[Y_i(x)-Y_i(\bar{x})|X_i=x] = \mathbbm{E}[Y_i|X_i=x]-\mathbbm{E}[Y_i(\bar{x})|X_i=x],$$
where we have used $Y_i = Y_i(X_i)$ in the rightmost expression. This expression highlights the challenge of identifying the causal quantity $\text{ATT}(x)$, since $\E[Y_i(\bar{x})|X_i=x]$ is counterfactual and not directly observed for $x \ne \bar{x}$.

In our empirical application, $x$ measures cigarettes per day, and the bunching point is $\bar{x}=0.$ So, $\text{ATT}(x)$ measures the average birth-weight loss (or gain) mothers who smoke $x$ cigarettes incur for smoking that amount in comparison to a counterfactual where they do not smoke at all. About 80\% of mothers in this application smoke zero cigarettes, leading to a case of boundary bunching as in the right panel of Figure \ref{fig:bunchingtypes}. Here, $\E[Y_i(\bar{x})|X_i=x]$ is the expected birth weight if mothers who smoke $x$ cigarettes per day were to quit. Due to confounders (e.g. drinking during pregnancy), this quantity can be expected to vary with $x$.

Local effects of the treatment around a value $x$ can be obtained by inspecting the derivative of the function $\text{ATT}(x)$. We define the \textit{average marginal effect near the bunching point,} $\text{AME}_{\bar{x}}^{+},$ as the right derivative of $\text{ATT}(x)$ as $x$ approaches the bunching point from above:
$$\text{AME}_{\bar{x}}^{+}:=\lim_{x\downarrow \bar{x}}\frac{\text{ATT}(x)}{x-\bar{x}}=\lim_{x\downarrow \bar{x}}\E\left[\frac{Y_i(x)-Y_i(\bar{x})}{x-\bar{x}}\middle| X_i=x\right].$$
When the bunching point is interior as in the left panel of Figure \ref{fig:bunchingtypes}, or the bunching point is on the right boundary of the support of $X_i$, one could define a similar $\text{AME}_{\bar{x}}^-$ parameter describing the left limit at $\bar{x}.$  We use the right derivative to define the treatment effects of interest because it fits our application and many others. In our application, $\text{AME}_{\bar{x}}^{+}$ is the birth weight loss (or gain) incurred by those who smoked just a little versus the counterfactual where they would not have smoked at all, expressed as a per-cigarette rate.

If the individual potential outcome functions $Y_i(x)$ are differentiable, and regularity conditions permitting the exchange of limits and expectations hold, one can interpret $\text{AME}_{\bar{x}}^{+}$ in terms of an average of the derivatives $Y_i'(x)$ of these dose-response functions, i.e. $$\text{AME}_{\bar{x}}^{+} = \lim_{x\downarrow \bar{x}} \text{AME}(x),$$ with $\text{AME}(x):=\E[Y_i'(x)|X_i=x]$. It is for this reason that we use the term \textit{marginal effect} to refer to $Y_i'(x)$, in line with e.g., \citet{hoderleinmammen,imbensnewey,chiangsasaki}. Other authors use the term \textit{partial effect} (see e.g., \citealt{sasaki2015, katosasaki2017}), emphasizing the interpretation of $Y_i(x)$ as $g(x,U_i)$ for an underlying structural function $g$ over heterogeneity $U_i$ in potential outcomes, in which case $Y_i'(x)$ denotes the partial derivative of $g(x,U_i)$ with respect to $x$. A related quantity is the average causal response on the treated (ACRT) function, studied in \cite{callaway2024difference}.  Despite our use of the term \textit{marginal effect}, our results do not require $Y'_i(x)$ to exist. Rather, we maintain weaker differentiability assumptions at the level of expectations of outcomes, which are sufficient to ensure that $\text{AME}_{\bar{x}}^{+}$ is well-defined.

Like with $\text{ATT}(x)$, identifying average marginal effects is challenging because the regression derivative of $Y_i$ on $X_i$ generally confounds the causal effect of treatment with a bias due to endogeneity:

\vspace{-1cm}
\begin{equation*} \label{eq:avgendogeneity}
 \frac{d}{dx} \mathbbm{E}[Y_i|X_i=x] = \frac{d}{dx} \mathbbm{E}[Y_i(x)|X_i=x] = \overbrace{\left.\frac{d}{dx'} \mathbbm{E}[Y_i(x')|X_i=x]\right|_{x'=x}}^{\text{causal effect}} + \overbrace{\left.\frac{d}{dx'}\mathbbm{E}[Y_i(x)|X_i=x']\right|_{x'=x}}^{\text{selection bias}},
\end{equation*}
where the first term above is, under regularity conditions, equal to the average marginal effect at $x,$ $AME(x)$. This equation depicts the fundamental problem of causal inference: the observed differences in mean outcomes among those with different treatment values combines both the causal effect of $x$ and selection bias. An analogous decomposition for $\text{AME}_{\bar{x}}^{+}$ implies that:
\begin{align} 
\text{AME}_{\bar{x}}^{+} 
=  \lim_{x\downarrow \bar{x}}\frac{d}{dx} \E[Y_i|X_i=x]- \lim_{x\downarrow \bar{x}}\frac{d}{dx} \E[Y_i(\bar{x})|X_i=x].
\label{eq:ameendogeneity}
\end{align}
The first term above is the right limit of the derivative of the regression function, and the second term reflects endogeneity: those with different values of $X_i$ may have different mean values of $Y_i(\bar{x})$. 

The following definitions will be useful in simplifying exposition throughout our analysis. First, we define $m(x)$ as the mean difference in observed outcomes at $x$ versus the boundary as $x \downarrow \bar{x}$:
\begin{equation*} 
    m(x):=\E[Y_i|X_i=x] -\E[Y_i|X_i=\bar{x}^+].
\end{equation*}
Similarly, define:
\begin{equation*} 
    s(x):=\E[Y_i(\bar{x})|X_i=x] - \E[Y_i(\bar{x})|X_i=\bar{x}^+],
\end{equation*}
which denotes the comparison of the counterfactual outcomes $Y_i(\bar{x})$ relative to the boundary as $x \downarrow \bar{x}$. Then, we write
\begin{equation} \label{eq:paramsDelta}
    \text{ATT}(x)=m(x)-s(x) \hspace{1cm} \textrm{ and } \hspace{1cm} \text{AME}_{\bar{x}}^{+}=m'(\bar{x}^+)-s'(\bar{x}^+),
\end{equation}
where the equalities follow from Theorem \ref{prop:AMEDelta} below. Both $m$ and $s$ are defined relative to the limit at the bunching point, so that $m(\bar{x}^+)=s(\bar{x}^+)=0$. Note as well that since only the first term in the definitions of $m$ and $s$ depends on $x$, we have $m'(x)=\frac{d}{dx} \E[Y_i|X_i=x]$ and $s'(x)=\frac{d}{dx} \E[Y_i(\bar{x})|X_i=x],$ whenever these derivatives exist.

Note that the function $m$ and the quantity $m'(\bar{x}^+)$ are identified directly from observables. Thus, the identification challenge for the causal parameters in \eqref{eq:paramsDelta} comes from the $s(\cdot)$ terms, which depend on unobserved counterfactuals.

\subsection{Identification of the average marginal effect near the bunching point} \label{sec:approach}
The main insight of this paper is established in this section. Specifically, it is the observation that to identify the derivative of the counterfactual function $\E[Y_i(\bar{x})|X_i=x]$ for values of $x$ near the bunching point $\bar{x}$, it is sufficient to identify the \textit{distribution} of the expected counterfactuals $E[Y_i(\bar{x})|X_i]$ for $X_i$ near $\bar{x}$. We then show in Section \ref{sec:idfrombunching} that bunching in $X_i$ makes it possible to identify this distribution, even though the counterfactual outcome $Y_i(\bar{x})$ is never observed for those with $X_i \ne \bar{x}$. The bias term $\lim_{x\downarrow \bar{x}}\frac{d}{dx} \E[Y_i(\bar{x})|X_i=x]$ in \eqref{eq:ameendogeneity} can then be identified.

To establish our first results, we assume the following:
\begin{assumption}[Continuity conditions] \label{as:continuoussupport}\label{as:smoothcounterfactual} The following hold:
\begin{enumerate}[(i)]
    \item $f_{X}(x)$ exists and is continuous on an open interval $(\bar{x},\bar{x}+\varepsilon_1)$ for some $\varepsilon_1>0$, and $f_X(\bar{x}^+)$ exists and is strictly positive.
    \item The function $x \mapsto \E[Y_i|X_i=x]$ is differentiable on an interval $(\bar{x},\bar{x}+\varepsilon_2)$ for some $\varepsilon_2>0$, and $\E[Y_i|X_i=\bar{x}^+]$ as well as $\lim_{x \downarrow \bar{x}}\frac{d}{dx}\E[Y_i|X_i=x]$ exist.
    \item The function $x\mapsto \E[Y_i(\bar{x})|X_i=x]$ is differentiable on an interval $(\bar{x},\bar{x}+\varepsilon_3)$ for some $\varepsilon_3>0$, where $\lim_{x\downarrow \bar{x}} \frac{d}{dx}\E[Y_i(\bar{x})|X_i=x]$ exists and is different from zero.
    \item $\text{ATT}(\bar{x}^+)=0,$ and $\text{ATT}'(\bar{x}^+)$ exists.
\end{enumerate}
\end{assumption}

Part (i) of Assumption \ref{as:continuoussupport} states that $X_i$ is continuously distributed with a density on an interval to the right of $\bar{x}$, though this density need not exist everywhere (e.g. there can be multiple bunching points provided that they are well-separated). Part (ii) of Assumption \ref{as:continuoussupport} says that a regression derivative exists, and the regression function and its derivative have a right limit at the bunching point. Both parts (i) and (ii) of Assumption \ref{as:continuoussupport} are restrictions on the observable data that can in principle be verified empirically.

 Part (iii) of Assumption \ref{as:smoothcounterfactual} requires that the selection function $s(x)$ be differentiable in a neighborhood above the bunching point. It also rules out the case in which there is no endogeneity when one approaches $\bar{x}$ from above, although, in this case, no correction for endogeneity is necessary. Note that \cite{caetano2015}'s test can be used in this setting to diagnose the presence of endogeneity needing a correction. 

The first statement in part (iv) of Assumption \ref{as:smoothcounterfactual} requires that $\E[Y_i|X_i=x^+]=\E[Y_i(\bar{x})|X_i=\bar{x}^+]$. The second statement in part (iv) of Assumption \ref{as:smoothcounterfactual} states that on average, the treatment effects are sufficiently smooth near the bunching point. Among those with treatment levels near the bunching point, marginally small doses of the treatment should have, on average, only marginally small effects. In the smoking example, this condition states that among mothers who smoke very little, smoking is not so harmful that a small amount can cause a stark decline in the baby's health. 

Part (iv) of Assumption \ref{as:smoothcounterfactual} is sufficient to guarantee that the $\text{AME}_{\bar{x}}^+$ is well defined, since $\text{ATT}(\bar{x}^+)=0$ implies that $\text{AME}_{\bar{x}}^+=\text{ATT}'(\bar{x}^+)$. A sufficient condition (though stronger than necessary) for part (iv) is that  the $Y_i(x)$ are differentiable and uniformly bounded with probability one near the bunching point, which furthermore implies that $\text{AME}_{\bar{x}}^+=\E[Y_i'(\bar{x})|X_i=\bar{x}^+]$. 
The following proposition establishes the connection between the parameters of interest and the functions $m$ and $s,$ as presented in Equation \eqref{eq:paramsDelta}. All proofs are found in Appendix \ref{ap:proofs}.
\begin{theorem} \label{prop:AMEDelta}
Under Assumption \ref{as:continuoussupport}, $\text{ATT}(x)=m(x)-s(x)$ and $\text{AME}_{\bar{x}}^{+}=m'(\bar{x}^+)-s'(\bar{x}^+)$.
\end{theorem}

We now move to the first main result of this paper. Part (iii) of Assumption \ref{as:smoothcounterfactual} implies that the counterfactual function $s(x)$ is differentiable and locally monotonic for $x$ in some neighborhood above the bunching point (without loss, we can take $I=(\bar{x},\bar{x}+\varepsilon),$ for some $\varepsilon < \min\{\varepsilon_1,\varepsilon_2,\varepsilon_3,\varepsilon_4\},$ and we note that $\varepsilon$ need never be known for our identification strategy).
The local monotonicity and differentiability in $I$ allow us to apply the well-known change-of-variables formula from integration theory to derive the following result. 
\begin{theorem}\label{thm:changeofvariables} If Assumption \ref{as:continuoussupport} holds, then there exists $I=(\bar{x},\bar{x}+\varepsilon)$ for some $\varepsilon>0$ such that, for all $x\in I,$ the density $f_{s(X)|I}(s(x))$ exists and is non-zero, and
\begin{equation}\label{eq:changeofvariables}
\left|s'(x)\right|=\frac{f_{X|I}(x)}{f_{s(X)|I}(s(x))}.
\end{equation}   
\end{theorem}
\noindent Note the absolute value $\left|s'(x)\right|$ in \eqref{eq:changeofvariables}: the right-hand side is always positive, but $s'(x)$ will be negative if there is negative selection. The textbook change-of-variables formula (see e.g. \citealt{fremlin2011measure} for a general formulation) states that $f_{u(X)}(t) = f_{X}(u^{-1}(t))/|u'(u^{-1}(t))|$ for any $t \in u^{-1}(I)$, given any function $u(x)$ that is differentiable and strictly increasing on an interval $I$ (and analogously for a strictly decreasing $u$). Then, the claim follows with $u(X)=s(X)$ and $t=s(x)$. Though the change-of-variables formula is a standard tool, a proof of Theorem \ref{thm:changeofvariables} is provided in Appendix \ref{ap:proofs}, which also establishes the local monotonicity of $s$ required for the change-of-variables result. Figure \ref{fig: intuition change of variables} provides a visual illustration of the change-of-variables theorem for scenarios with high and low selection.

\begin{figure}[H]
\begin{center}
\caption{Using the Change-of-Variables Theorem to Identify $|s'(\bar{x}^+)|$}
\label{fig: intuition change of variables}
\vspace{.1in}
%—— Panel 1: High Selection ——
\begin{tikzpicture}[scale=1.2, transform shape]
  \begin{axis}[
      width=0.45\textwidth, height=5cm,
      domain=-4:4, samples=200,
      axis lines=left,
      xlabel={$\bar{x}$}, ylabel={\scriptsize{Density}},
                xtick=\empty,  ytick=\empty,
      xmin=-4, xmax=4, ymin=0, ymax=1.2,
      title={\footnotesize{High Selection: $s(X_i)=2(X_i-\bar{x})$}},
      title style={
        at={(axis description cs:0.5,0)},  % middle of bottom edge
        anchor=north,                       % title’s north edge sits on that line
        yshift=-2em                         % tweak vertical gap
      },
      legend pos=north east,
      legend style={font=\footnotesize},
      legend image code/.code={},   % no legend swatches
    ]
    % Blue curve: f_{X^*}(x) = φ(x;μ=-1,σ=1)
    \addplot[blue,thick,domain=-4:0]
      {exp(-((x+1)^2)/2)/sqrt(2*pi)};
    \addplot[blue,thick,domain=0:4]
      {exp(-((x+1)^2)/2)/sqrt(2*pi)};
    % Shade NEGATIVE side for blue
    \addplot[blue!30,fill=blue!30,draw=none,fill opacity=0.4,domain=-4:0]
      {exp(-((x+1)^2)/2)/sqrt(2*pi)} \closedcycle;

    % Red curve: f_{S}(x) = ½·φ((x)/2;μ=-1,σ=1)
    \addplot[red!80!black,thick,domain=-4:0]
      {0.5*exp(-((x/2+1)^2)/2)/sqrt(2*pi)};
    \addplot[red!80!black,thick,domain=0:4]
      {0.5*exp(-((x/2+1)^2)/2)/sqrt(2*pi)};
    % Shade NEGATIVE side for red
    \addplot[red!30,fill=red!30,draw=none,fill opacity=0.4,domain=-4:0]
      {0.5*exp(-((x/2+1)^2)/2)/sqrt(2*pi)} \closedcycle;

    % Divider & bracket at x=0
    \draw[dotted] (axis cs:0,0) -- (axis cs:0,0.4);

      \node at (0.3,0.8) {\tiny{${\color{blue}f_{X}(\bar{x}^+)}=2{\color{red!80!black}f_{s(X)}(s(\bar{x}^+))}$}};
    \draw[->, blue, thick]
    (-0.06,0.3) -- (-.8,0.7);
    \draw[->, red!80!black, thick]
    (0.1,0.16) -- (1.1,0.7);

    \addplot[blue,only marks,mark=*,mark size=2pt]
      coordinates {(0,{exp(-(((0)+1)^2)/2)/sqrt(2*pi)})};
    \addplot[red!80!black,only marks,mark=*,mark size=2pt]
      coordinates {(0,{0.5*exp(-(((0)/2+1)^2)/2)/sqrt(2*pi)})};

  \end{axis}
\end{tikzpicture}
\quad
%—— Panel 2: Low Selection ——
\begin{tikzpicture}[scale=1.2, transform shape]
  \begin{axis}[
      width=0.45\textwidth, height=5cm,
      domain=-4:4, samples=200,
      axis lines=left,
                xtick=\empty,  ytick=\empty,  
      xlabel={$\bar{x}$}, ylabel={\scriptsize{Density}},
      xmin=-4, xmax=4, ymin=0, ymax=1.2,        % increased ymax
      title={\footnotesize{Low Selection: $s(X_i)=\tfrac13(X_i-\bar{x})$}},
      title style={
        at={(axis description cs:0.5,0)},  % middle of bottom edge
        anchor=north,                       % title’s north edge sits on that line
        yshift=-2em                         % tweak vertical gap
      },      
      legend pos=north east,
      legend style={font=\footnotesize},
      legend image code/.code={},  
    ]
    % Blue curve
    \addplot[blue,thick,domain=-4:0]
      {exp(-((x+1)^2)/2)/sqrt(2*pi)};
    \addplot[blue,thick,domain=0:4]
      {exp(-((x+1)^2)/2)/sqrt(2*pi)};
    \addplot[blue!30,fill=blue!30,draw=none,fill opacity=0.4,domain=-4:0]
      {exp(-((x+1)^2)/2)/sqrt(2*pi)} \closedcycle;

    % Red curve: f_s(x) = 3·φ(3x;μ=-1,σ=1)
    \addplot[red!80!black,thick,domain=-4:0]
      {3*exp(-((3*x+1)^2)/2)/sqrt(2*pi)};
    \addplot[red!80!black,thick,domain=0:4]
      {3*exp(-((3*x+1)^2)/2)/sqrt(2*pi)};
    \addplot[red!30,fill=red!30,draw=none,fill opacity=0.4,domain=-4:0]
      {3*exp(-((3*x+1)^2)/2)/sqrt(2*pi)} \closedcycle;

    % Divider & bracket at x=0
    \draw[dotted] (axis cs:0,0) -- (axis cs:0,1.2);

    \draw[->, blue, thick]
    (-0.06,0.27) -- (-.55,0.5);
    \draw[->, red!80!black, thick]
    (0.07,0.7) to[out=120,in=90] (1.1,0.63);
      \node at (0.55,0.55) {\tiny{${\color{blue}f_{X}(\bar{x}^+)}=\tfrac13{\color{red!80!black}f_{s(X)}(s(\bar{x}^+))}$}};

    \addplot[blue,only marks,mark=*,mark size=2pt]
      coordinates {(0,{exp(-(((0)+1)^2)/2)/sqrt(2*pi)})};
    \addplot[red!80!black,only marks,mark=*,mark size=2pt]
      coordinates {(0,{3*exp(-(((3*(0)+1)^2))/2)/sqrt(2*pi)})};

 \node[font=\footnotesize, yshift=-10pt] at (axis cs:0,0) {\(\bar{x}\)};
    % Legend
    %\addlegendentry{\(\color{blue}f_{X^*}(x)\)}
    %\addlegendentry{\(\color{red}f_{s(X^*)}(s(x))\)}
  \end{axis}
\end{tikzpicture}

\end{center}
Note: The blue curve is the density of $X_i\sim N(\bar{x}-1,\bar{x}+1)$ in both panels. The red curve is the density of $s(X_i)$ given that the selection takes the linear form $\E[Y_i(\bar{x})|X_i=x]= a+b \cdot (x-\bar{x})$, so that $s(x) = b \cdot (x-\bar{x})$ and $s'(x)=b$. The equations relating the heights of the solid dots on each panel highlight the proportionality of the densities at $x=\bar{x}^+$, matching $s'(\bar{x}^+)=b$. Note that for the case with less selection ($b=\frac{1}{3}$), the density of $s(X_i)$ at $\bar{x}^+$ is larger, increasing the denominator in the formula for  $s'(\bar{x}^+).$ To simplify, this illustration assumes $s(\bar{x}^+)=\bar{x},$ which is true when $\bar{x}=0$, as in our application.
\end{figure}
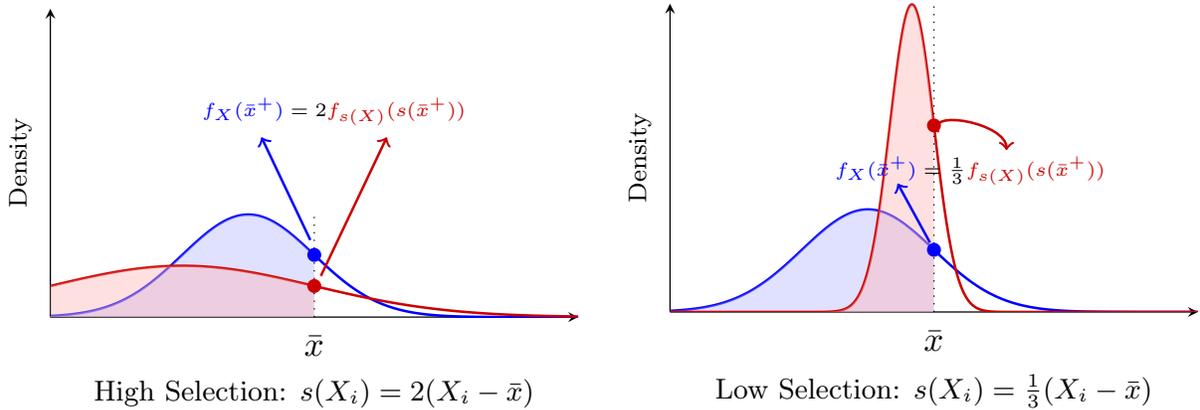

Let $\theta:=\text{sgn}\left(\lim_{x\downarrow \bar{x}}\frac{d}{dx}\E[Y_i(\bar{x})|X_i=x]\right)=\text{sgn}(s'(\bar{x}^+))$ be the sign of the selection bias as $x$ approaches $\bar{x}$ from the right. As an application of Theorem \ref{thm:changeofvariables}, we have the following expression for the $\text{AME}_{\bar{x}}^{+}$:
\begin{theorem}\label{thm:extensivemargin} If Assumption \ref{as:continuoussupport} holds, then
\begin{equation*}
\text{AME}_{\bar{x}}^{+} =m'(\bar{x}^+)-\theta \cdot\frac{f_{X|I}(\bar{x}^+)}{f_{s(X)|I}(s(\bar{x}^+))}.
\end{equation*}
\end{theorem}

\section{Identification using bunching} \label{sec:idfrombunching}
The previous section established that identification of treatment effects may be possible if we can identify the sign of the selection bias at the boundary point,  $\theta=\text{sgn}(s(\bar{x}^+))$, and the limit of the density of the selection bias variable, $f_{s(X)}(s(x)),$ as $x \downarrow \bar{x}$. In this section, we show how information at the bunching point may be used to obtain these quantities. 

We start by noting that the treatment variable plays two distinct roles. First, it describes the dose taken by an individual, i.e. the number of cigarettes smoked in our application. Second, it tells us something about that individual, namely the fact that those with that treatment value are of the ``type'' that selected (or was selected into) that amount. Thus, for example, in the parameter $\text{ATT}(x)=\E[Y_i(x)-Y_i(0)|X_i=x],$ the $x$ in ``$Y_i(x)$'' describes the dose, and the  $x$ in ``$X_i=x$'' describes the group that selected it. Henceforth, we separate the notation of the two concepts: the dose is the treatment variable, denoted $X_i$ as before, and the selection variable is denoted $X_i^*$. 

In most cases, the selection variable $X_i^*$ is identical to $X_i.$ Indeed, this is precisely what gives rise to endogeneity: if $X_i^*$ is correlated with the potential outcomes $Y_i(x),$ then the correlation between $Y_i=Y_i(X_i)$ and $X_i=X_i^*$ reflects both the causal effect (i.e. the part that refers to the variation of the function $Y_i(x)$ with $x$ for a given $i$), and the selection function (i.e. the part that refers to the variation of $Y_i(x)$ across the $i$ with different values of $X_i^*$). This is the case here as well for values away from the bunching point because, when $X_i^*>\bar{x},$ there is no constraint on the treatment value, and so $X_i^*=X_i$. However, the bunching point is special in that multiple values of the selection variable $X_i^*$ occur simultaneously at the same treatment value $X_i=\bar{x}$. 

A separation between ``types'' and their value of a variable that exhibits bunching is a common feature of the bunching literature \citep{saez10, klevenwaseem, blomquist2021bunching, bertanha2024bunching,bertanha2023better,ccn_metrics,goff2020treatment,CCNT,pollinger}. The separation arises from constraints on individuals' choices that cause different types to all choose the common bunching point. The idea is that at the bunching point, selection breaks away from the dose, and observations with diverse selection values receive the exact same dosage amount. Thus, the bunching point affords the possibility of learning about the relationship between the selection variable and other variables without any interference from dose variation.

We write
\begin{equation}\label{eq:xbarmax}
  X_i=\max\{X_i^*,\bar{x}\},
\end{equation}
which fits the case of bunching at the left boundary of the support of $X_i$. The right boundary case also fits this description, by redefining the treatment to be $2\bar{x}-X_i$. Interior bunching resulting from kinks in the budget function can also be adapted to fit Equation \eqref{eq:xbarmax}, as we describe in Example \ref{ex:kink} (Section \ref{sec:xstar}).

In general, one can think of $X_i^*$ as an index summarizing all observable and unobservable individual characteristics that determine the treatment. We return to $X_i^*$ in Section \ref{sec:xstar}, where we examine its possible economic interpretation in terms of structural primitives or reduced form quantities in specific models, and discuss its (non-)uniqueness. To aid comprehension in the following sections, we provide here a brief intuitive discussion of $X_i^*$ in the context of our empirical application to maternal smoking. In that setting, it is natural to think that, although all nonsmoking mothers share a value $X_i=0$ of the treatment, they may differ in the intensity of their preference towards not smoking or other factors that influence their choice. For instance, if $\rho_i$ is a parameter that governs the person's relative preference towards smoking, there may exist a value $\bar{\rho}$ and a smooth function $h$ such that $X_i = h(\rho_i)$ when $\rho_i \ge \bar{\rho}$ (i.e. when the person's preference towards smoking is sufficiently high), and $X_i=0$ otherwise, where $P(\rho_i < \bar{\rho})>0$ (i.e. some mothers strictly prefer not to smoke). 

Intuitively, in our method $X^*_i$ allows us to track observations at $X_i=\bar{x}$ based on how selected they are relative to the observations near the bunching point on the positive side. The key assumptions about $X_i^*$ will be that the smoothness conditions from Section \ref{sec:approach}  can be extended to values of $X_i^*$ around $\bar{x},$ so that those observations away from the bunching point with $X_i$ near $\bar{x}$ are comparable to the observations with $X_i^*=\bar{x}$. This allows us to substitute the limit of the density of the selection above the bunching point in Theorem \ref{thm:extensivemargin} into the density of the selection at the bunching point. The following sections then show that the sign of the selection bias is identified (Section \ref{sec:thetaidentification}), and that the density may be obtained by a deconvolution of the density of the outcome near the bunching point from the density of the outcome at the bunching point (Sections \ref{sec:translating} and \ref{sec:densityidentification}). Finally, in Section \ref{sec:xstar} we discuss the nature of $X_i^*$ when it arises from choice models, how it may be artificially constructed under some conditions, and the invariance of the identification results to monotonic transformations of $X_i^*$. 

\subsection{Identifying the sign of the selection bias, \texorpdfstring{$\theta$}{theta}}\label{sec:thetaidentification}

We begin by extending the definition of $s(x)$ to leverage the notation $X_i^*,$ as $s(x):=\E[Y_i(\bar{x})|X_i^*=x]-\E[Y_i(\bar{x})|X_i^*=\bar{x}^+]$. Since $X_i=X_i^*$ when $X_i > \bar{x}$, this coincides with the function $s$ defined in Section \ref{sec:approach} for all $x > \bar{x}$. For $x\leq \bar{x},$ $X_i=\bar{x}$ but $s(x)$ may vary with the following restriction:

\begin{assumption} \label{as:smono} For any $x < \bar{x},$ $\text{sgn}(s(x))=-\lim_{x \downarrow \bar{x}} \text{sgn}(s(x))$.
\end{assumption}
Assumption \ref{as:smono} states that if $s(x)$ is increasing in a positive neighborhood around the bunching point, then $s(x)<0$ for all values of $x<\bar{x}$. Conversely, if $s(x)$ is decreasing in a positive neighborhood around the bunching point, then $s(x)>0$ for all values of $x< \bar{x}$. Specifically, either $\E[Y_i(\bar{x})|X_i^*=x']<\E[Y_i(\bar{x})|X_i^*=\bar{x}^+]<\E[Y_i(\bar{x})|X_i^*=x'']$ for all $x'< \bar{x}$ and all $x''>\bar{x}$ in a neighborhood right above the bunching point, or the reverse ordering is true. Intuitively, the selection function $\E[Y_i(\bar{x})|X_i^*=x]$ remains uniformly either above or below its value at the bunching point $\E[Y_i(\bar{x})|X_i^*=\bar{x}]$, for all $x$ to the left of the bunching point, and is unrestricted to the right of the bunching point. 

In the smoking example, suppose that among mothers smoking positive amounts, those who smoke more are negatively selected relative to those who smoke less (i.e., smoking more is associated with worse untreated outcomes). Then, Assumption \ref{as:smono} states that the nonsmoking mothers would have even higher birth weights. This needs to hold only on average for every selection value, so  some individual nonsmoking mothers could have lower birth weights than some smoking mothers.

Assumption \ref{as:smono} holds trivially if $x\mapsto\E[Y_i(\bar{x})|X_i^*=x]$ is monotonic, in which case $s(x)$ is also monotonic. However, the assumption is weaker than monotonicity. Figure \ref{fig:squigglyincreasing} illustrates two examples of non-monotonic functions that satisfy Assumption \ref{as:smono}, one case with $s'(\bar{x})>0,$
and another with $s'(\bar{x})<0,$ respectively. Note that $s(\bar{x}^+)=0$ by definition, and Assumption \ref{as:smono} does not constrain the behavior of $s(x)$ for positive $x,$ outside of a neighborhood of the bunching point.

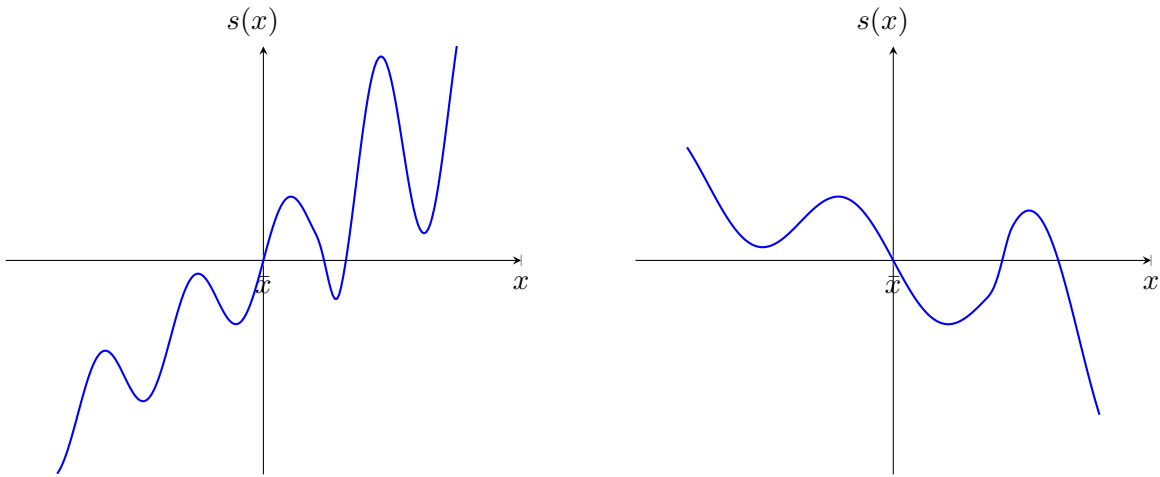
\begin{figure}[htp!]
\begin{center}
    \caption{Examples of functions $s$ that satisfy Assumption \ref{as:smono} but are not monotonic.} \label{fig:squigglyincreasing}
    \begin{tikzpicture}[scale = 1]
			\begin{axis}[
                samples=200,
			axis x line=middle,
			axis y line=middle,
                %y axis line style={opacity=0},
			xmin=-5, xmax=5, 
			ymin=-5, ymax=5,
			ylabel={$s(x)$},
                ylabel style={
                at={(axis description cs:0.5,1)},
                anchor=south,
                xshift=-4pt,
                },            
			ytick=\empty,
			xtick=\empty,
			extra x ticks={0,5},
			extra x tick labels={$\bar{x}$,$x$},
			extra y ticks={\empty},
			extra y tick labels={\empty},
			]
                \pgfmathsetmacro{\deltaval}{0.5}
			\addplot[blue,thick,domain=-4:4]
      expression{(1 + 1.5 * smoothstep(clamp((x-1)/\deltaval,0,1)))* sin(200*x)+ x};      
      %{sin(200*x)*(1+1.5*(x>1))+x};
			\end{axis}		
    \end{tikzpicture} \quad \quad \quad  
    \begin{tikzpicture}[scale = 1]
			\begin{axis}[
                samples=200,
			axis x line=middle,
			axis y line=middle,
                %y axis line style={opacity=0},
			xmin=-5, xmax=5, 
			ymin=-5, ymax=5,
			ylabel={$s(x)$},
                ylabel style={
                at={(axis description cs:0.5,1)},
                anchor=south,
                xshift=-4pt,
                },
			ytick=\empty,
			xtick=\empty,
			extra x ticks={0,5},
			extra x tick labels={$\bar{x}$,$x$},
			extra y ticks={\empty},
			extra y tick labels={\empty},
			]
                \pgfmathsetmacro{\deltaval}{0.5}
			\addplot[blue,thick,domain=-4:4]
            expression{-sin(100*x)* (1 + 1.5 * smoothstep(clamp((x-1.8)/\deltaval,0,1)))- 0.5*x}; 
      %{-sin(100*x)*(1+1.5*(x>1.8))-.5*x};
			\end{axis}		
    \end{tikzpicture}
\end{center}
Note: Each panel shows an example of a selection function $s(x)$ that satisfies Assumption \ref{as:smono} but is not monotonic in $x$. Note that in both cases $s(\bar{x}^+)=0$, which is true by definition.
\end{figure}

The following lemma establishes that the sign of the selection bias is equal to the sign of the discontinuity of the expected outcome at the bunching point. 
\begin{lemma}\label{lem:theta} If Assumptions \ref{as:continuoussupport} and \ref{as:smono} hold, then $\theta$ is identified as
$$\theta=\text{sgn}\left(\E[Y_i|X_i=\bar{x}^+]-\E[Y_i|X_i=\bar{x}]\right).$$
\end{lemma}
\noindent Lemma \ref{lem:theta} relates to \citet{caetano2015}'s test of exogeneity, which is based on the discontinuity of the outcome at the bunching point. If the sign of the discontinuity is not zero, then the test rejects the exogeneity of $X_i$. However, in our setting, we can do more than just test exogeneity, as the same discontinuity also allows us to sign the selection bias within an interval of $\bar{x}$.

The intuition of Lemma \ref{lem:theta} can be obtained from Figure \ref{fig:squigglyincreasing}. Suppose that we observe a positive discontinuity of the outcome at the bunching point. Then,  $\E[Y_i(\bar{x})|X_i^*=x]$ must be below $\E[Y_i(\bar{x})|X_i^*=\bar{x}^+]$ for at least some $x<0$. By Assumption \ref{as:smono}, we then know that all $\E[Y_i(\bar{x})|X_i^*=x]$ must be below $\E[Y_i(\bar{x})|X_i^*=\bar{x}^+]$ for all $x<\bar{x},$ so $s(x)<0$ for all $x<\bar{x}$. We must therefore be in a situation akin to the left plot. It follows that $\E[Y_i(\bar{x})|X_i^*=x]$ for $x$ slightly above $\bar{x}$ are all above $\E[Y_i(\bar{x})|X_i^*=\bar{x}^+],$ and thus $s'(\bar{x}^+)>0$. 

\subsection{Translating the identification problem into the bunching point}\label{sec:translating}

We next concern ourselves with the elimination of the unknown interval $I$ from the formula for $\text{AME}_{\bar{x}}^{+}$ in Theorem \ref{thm:extensivemargin}. Precisely, we need to substitute the density $f_{s(X)|I}$ with the density $f_{s(X^*)|X=\bar{x}}.$ In the following section, we then show that $f_{s(X^*)|X=\bar{x}}$ may be identified.

\begin{assumption}\label{as:bunchingcontinuousnew1} The following hold:
\begin{enumerate}[(i)]
    \item $f_{s(X^*)}(v)$ exists for all $v \in s\left((-\infty,\bar{x}] \cup I\right)$, and is right-continuous in $v$ at $0$. 
    \item The function $x\mapsto \E[Y_i(\bar{x})|X^*_i=x]$ is right-continuous in $x$ at $\bar{x}$.
\end{enumerate}
\end{assumption}

In Theorem \ref{thm:changeofvariables}, we established the existence of $f_{s(X^*)}(v)$ on $I$. Part (i) of Assumption \ref{as:bunchingcontinuousnew1} extends this result, requiring that the density exists also to the left of the bunching point. It is sufficient (but not necessary) for Assumption \ref{as:bunchingcontinuousnew1} (i) that $x\mapsto \E[Y_i(\bar{x})|X_i^*=x]$ is monotonic and $X_i^*$ has a density for $(-\infty,\bar{x}]$. More generally, if $X_i^*$ has a density on $(-\infty,\bar{x}]$, then part (i) holds provided that the set of $x$ such that $s(x)=t$ has Lebesgue measure zero, for any $t \in s^{-1}((-\infty,\bar{x}])$. Thus part (i) could be thought of as a consequence of $X_i^*$ being continuously distributed on the left side of $\bar{x}$, requiring no restrictive assumptions on selection. We state condition (i) above because it is weaker than this.

Item (ii) of Assumption \ref{as:bunchingcontinuousnew1} implies that the observations with $X_i^*=\bar{x}$ are comparable on average to the observations with $X_i^*=\bar{x}^+$. In the smoking example, it is equivalent to saying that if the mothers who smoke very little (i.e. those with $\rho_i$ slightly larger than $\bar{\rho}$) were to stop smoking, their outcomes would be very similar to the outcomes of nonsmoking mothers who are indifferent between not smoking and smoking a bit (i.e., those with $\rho_i=\bar{\rho}$). Note that the claim is not that mothers near the bunching point are comparable to mothers \textit{at} the bunching point, since the mothers at the bunching point can also include those who strictly prefer not to smoke.

Assumption \ref{as:bunchingcontinuousnew1} (ii)   states that $\E[Y_i(\bar{x})|X_i^*=\bar{x}] = \E[Y_i(\bar{x})|X_i^*=\bar{x}^+],$ and thus it implies that $s(\bar{x})=0$. Moreover, it allows us to extend part (iii) of Assumption \ref{as:smoothcounterfactual} to the interval $[\bar{x},\bar{x}+\varepsilon_3),$ and note that, since $\lim_{x \downarrow \bar{x}} \frac{d}{dx}\E[Y_i(\bar{x})|X_i=x] \ne 0,$ we have $s'(\bar{x})\ne 0$. We then obtain the following extension of Theorem \ref{thm:extensivemargin}:
\begin{theorem} \label{cor:AMEbunching} If Assumptions \ref{as:continuoussupport}-\ref{as:bunchingcontinuousnew1} hold, then
\begin{equation} \label{eq:AMEbunching}
\text{AME}_{\bar{x}}^{+} = m'(\bar{x}^+)-\theta \cdot\frac{ f_{X}(\bar{x}^+)/F_X(\bar{x})}{f_{s(X^*)|X=\bar{x}}(0)}.
\end{equation}
\end{theorem}
\noindent Equation \eqref{eq:AMEbunching} shows that, under Assumptions \ref{as:continuoussupport}-\ref{as:bunchingcontinuousnew1}, identifying $\text{AME}_{\bar{x}}^{+}$ reduces to the problem of identifying $\theta$ and $f_{s(X^*)|X= \bar{x}}(0).$ Since $\theta$ is identified by Lemma \ref{lem:theta} in Section \ref{sec:thetaidentification}, the only remaining piece is the identification of the density of $s(X_i^*)$ at the bunching point, which we tackle in the following section.

Figure \ref{fig: intuition change of variables section 3} illustrates how we apply the change-of-variables theorem in Theorem \ref{cor:AMEbunching}. The left panel is identical to the left panel in Figure \ref{fig: intuition change of variables}, except that the densities are with respect to the selection variable $X^*_i,$ and dashed lines depict the parts of these densities that are not identifiable. The densities can \textit{both} be identified only at $x=\bar{x}$. The right panel shows how identification is achieved. The red line is  $f_{s(X^*)|X=\bar{x}},$ which on $(-\infty,\bar{x}]$ corresponds to the solid red line in the left panel divided by the probability of bunching, $F_X(\bar{x})$. Equivalently, we rescale the solid part of the blue density in the left panel by the same amount.

\begin{figure}[H]
\begin{center}
\caption{Using Theorem \ref{cor:AMEbunching} to identify $|s'(\bar{x})|$}
\label{fig: intuition change of variables section 3}
\vspace{.1in}
%—— Panel 1: High Selection ——
\begin{tikzpicture}[scale=1.2, transform shape]
  \begin{axis}[
      width=0.45\textwidth, height=5cm,
      domain=-4:4, samples=200,
      axis lines=left,
      xlabel={$\bar{x}$}, ylabel={\scriptsize{Density}},
      xtick=\empty,  ytick=\empty,
      xmin=-4, xmax=4, ymin=0, ymax=1.2,
      title={\footnotesize{Unconditional Densities}},
      title style={
        at={(axis description cs:0.5,0)},
        anchor=north,
        yshift=-2em
      },
      legend pos=north east,
      legend style={font=\footnotesize},
      legend image code/.code={},
    ]
    % Blue curve: f_{X^*}(x) = φ(x+1)
    \addplot[dashed,blue,thick,domain=-4:0]
      {exp(-((x+1)^2)/2)/sqrt(2*pi)};
    \addplot[blue,thick,domain=0:4]
      {exp(-((x+1)^2)/2)/sqrt(2*pi)};
    \addplot[blue!30,fill=blue!30,draw=none,fill opacity=0.4,domain=-4:0]
      {exp(-((x+1)^2)/2)/sqrt(2*pi)} \closedcycle;

    % Red curve: f_S(x) = 0.5 * φ(x/2 + 1)
    \addplot[red!80!black,thick,domain=-4:0]
      {0.5*exp(-((x/2+1)^2)/2)/sqrt(2*pi)};
    \addplot[dashed,red!80!black,thick,domain=0:4]
      {0.5*exp(-((x/2+1)^2)/2)/sqrt(2*pi)};
    \addplot[red!30,fill=red!30,draw=none,fill opacity=0.4,domain=-4:0]
      {0.5*exp(-((x/2+1)^2)/2)/sqrt(2*pi)} \closedcycle;

    \draw[dotted] (axis cs:0,0) -- (axis cs:0,0.4);
    \node at (0.3,0.8) {\tiny{${\color{blue}f_{X}(\bar{x}^{+})}=2{\color{red!80!black}f_{s(X^*)}(0)}$}};
    \draw[->, blue, thick]
    (-0.06,0.3) -- (-.8,0.7);
    \draw[->, red!80!black, thick]
    (0.1,0.16) -- (1.1,0.7);

    \addplot[blue,only marks,mark=*,mark size=2pt]
      coordinates {(0,{exp(-((0+1)^2)/2)/sqrt(2*pi)})};
    \addplot[red!80!black,only marks,mark=*,mark size=2pt]
      coordinates {(0,{0.5*exp(-((0/2+1)^2)/2)/sqrt(2*pi)})};
  \end{axis}
\end{tikzpicture}
\quad
%—— Panel 2: Rescaled conditional densities ——
\begin{tikzpicture}[scale=1.2, transform shape]
  \begin{axis}[
      width=0.45\textwidth, height=5cm,
      domain=-4:4, samples=200,
      axis lines=left,
      xlabel={$\bar{x}$}, ylabel={\scriptsize{Density}},
      xtick=\empty,  ytick=\empty,
      xmin=-4, xmax=4, ymin=0, ymax=.8,
      title={\footnotesize{Conditional Densities}},
      title style={
        at={(axis description cs:0.5,0)},
        anchor=north,
        yshift=-2em
      },
      legend pos=north east,
      legend style={font=\footnotesize},
      legend image code/.code={},
    ]
    % Blue curve: f_{X^*}(x) = φ(x+1)
    \addplot[blue,thick,domain=0:4]
      {exp(-((x+1)^2)/2)/sqrt(2*pi)};
    % Red curve: f_S(x) = 0.5 * φ(x/2 + 1)
    \addplot[red!80!black,thick,domain=-4:0]
      {0.5*exp(-((x/2+1)^2)/2)/sqrt(2*pi)};
    \draw[dotted] (axis cs:0,0) -- (axis cs:0,0.4);
    \node at (0.3,0.6) {\tiny{${\color{blue}f_{X}(\bar{x}^{+})/F_X(\bar{x})}=2{\color{red!80!black}f_{s(X^*)|X=\bar{x}}(0)}$}};
    \draw[->, blue, thick]
    (-0.085,0.28) -- (-.8,0.5);
    \draw[->, red!80!black, thick]
    (0.1,0.16) -- (1.1,0.5);

    \addplot[blue,only marks,mark=*,mark size=2pt]
      coordinates {(0,{exp(-((0+1)^2)/2)/sqrt(2*pi)})};
    \addplot[red!80!black,only marks,mark=*,mark size=2pt]
      coordinates {(0,{0.5*exp(-((0/2+1)^2)/2)/sqrt(2*pi)})};
  \end{axis}
\end{tikzpicture}
\end{center}
Note: The left panel is identical to the left panel from Figure \ref{fig: intuition change of variables}, except that it now refers to $X_i^*.$ Dashed lines show the unidentifiable parts of the distributions. Note that $\bar{x}$ is the only point at which both densities are identifiable. The right panel rescales the identifiable parts of the left-panel densities by $1/F_X(\bar{x}),$ so the red line shows the density conditional on $X^*_i\leq \bar{x}$ (which is equivalent to conditioning on $X_i=\bar{x}$). In this illustration, $\bar{x}=0,$ as in our application, and therefore $s(\bar{x}^+)=\bar{x}$.
\end{figure}
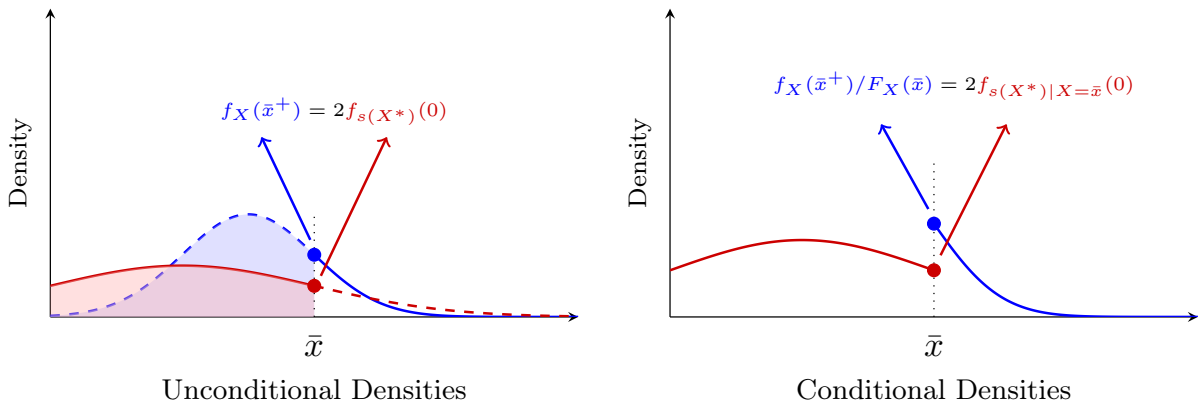

\subsection{Identifying the distribution of \texorpdfstring{$s(X^*)|X= \bar{x}$}{}}\label{sec:densityidentification}
Finally, we turn to the identification of $f_{s(X^*)|X=\bar{x}}(0)$. Define the random variable
$$\epsilon_i = Y_i - \mathbbm{E}[Y_i|X^*_i]=Y_i-\text{ATT}(X_i)+\E[Y_i(\bar{x})|X_i^*],$$
which is the unconfounded variation in the outcome, i.e. the idiosyncratic part of the outcome that remains after we eliminate the mean effect of treatment and the part of the outcome determined by selection. Note that for $X_i>\bar{x},$ $\epsilon_i=Y_i - \mathbbm{E}[Y_i|X_i]$ is identified.

We can write with probability one that $Y_i=y_{\bar{x}}+\text{ATT}(X_i)+s(X_i^*)+\epsilon_i,$ where the constant $y_{\bar{x}}:=\mathbbm{E}[Y_i(\bar{x})|X_i^*=\bar{x}]$ (using that $s(X_i^*)=\mathbbm{E}[Y_i(\bar{x})|X^*_i]-y_{\bar{x}}$ by Assumption \ref{as:bunchingcontinuousnew1}). This shows that when $X_i > \bar{x}$, there is variation in $Y_i$ coming both from $\text{ATT}(X_i)$ and from $s(X_i^*)$. However, exactly at $X_i=\bar{x},$ $\text{ATT}(\bar{x})=0,$ and thus
\begin{equation}\label{eq:equationzero}
    Y_i = y_{\bar{x}}+ s(X_i^*) + \epsilon_i
\end{equation}
with probability one. Thus, the outcome variation at the bunching point reflects the variation of $s(X_i^*)$ unconfounded by the variation of $\text{ATT}(X_i)$. Unfortunately, the distribution of the outcome at the bunching is a convolution of the density we want to identify, $f_{s(X^*)|X=\bar{x}},$ and the distribution of the idiosyncratic remainder, $y_{\bar{x}}+\epsilon_i$ (i.e., $f_{Y|X=\bar{x}}=f_{s(X^*)|X=\bar{x}}\ast f_{y_{\bar{x}}+\epsilon|X=\bar{x}}$). The following conditions allow us to deconvolve these distributions.

\begin{assumption}\label{as:deconvregularity} The following hold:
\begin{enumerate}[(i)]
    \item $f_{Y|X=\bar{x}}$ exists.
    \item $\epsilon_i|X^*_i=x\rightarrow_{d}\epsilon_i|X_i^*=\bar{x}$ as $x\downarrow \bar{x}.$ 
    \item $\epsilon_i \indep X_i^*| X_i=\bar{x}$.
\end{enumerate}
\end{assumption}
Item (i) of Assumption \ref{as:deconvregularity} states that the outcome has a density at the bunching point. Item (ii) says that the idiosyncratic variation in the outcome near the bunching point is distributed similarly to the idiosyncratic variation in the outcome of those at the bunching point with $X_i^*=\bar{x}$. If $Y|X=x$ also has a density for all $x$ near the bunching point, then item (ii) of Assumption \ref{as:deconvregularity} is equivalent to $Y_i(\bar{x})|X^*_i=x\rightarrow_{d}Y_i(\bar{x})|X_i^*=\bar{x}$ as $x\downarrow \bar{x}$, i.e. that the distribution of $Y_i(\bar{x})$ does not change abruptly as $X_i^*$ approaches $\bar{x}$.

Note that $\epsilon_i=Y_i-\E[Y_i|X_i^*]$ is mean independent of $X_i^*$ by construction. Part (iii) of Assumption \ref{as:deconvregularity} extends this into full independence, at least at the bunching point. In our application, this condition says that, among the nonsmoking mothers, after removing the mean of the birth weight that is due to the selection variable $X_i^*$, the remainder is independent of the relative preference for smoking (or whatever else determines $X_i^*$). Part (iii) of Assumption \ref{as:deconvregularity} may be substituted with a weaker but less intuitive condition known as subindependence, which has long been used in the deconvolution literature (see discussion in e.g. \citealt{Hamedanisubindependence}), and was formalized in \citet{schennach2019convolution}. In our context, the relevant subindependence condition translates to: for all $t \in \mathbbm{R},$
$$\E[e^{\textbf{i}t(s(X_i^*)+\epsilon_i)}|X_i=\bar{x}]=\E[e^{\textbf{i}t s(X_i^*)}|X_i=\bar{x}]\cdot \E[e^{\textbf{i}t\epsilon_i}|X_i=\bar{x}],$$ where $\textbf{i}=\sqrt{-1}$. \citet{schennach2019convolution} shows that subindependence is no ``stronger'' than mean independence, in the sense that subindependence imposes the same number of restrictions on the data-generating process as mean independence does. Nevertheless, mean independence does not imply subindependence.

Item (iii) of Assumption \ref{as:deconvregularity} implies that $f_{y_{\bar{x}}+\epsilon|X=\bar{x}}=f_{y_{\bar{x}}+\epsilon|X^*=\bar{x}}.$  Item (ii) of Assumption \ref{as:deconvregularity} implies that 
$f_{y_{\bar{x}}+\epsilon|X^*=\bar{x}}=f_{Y|X=\bar{x}^+}.$ Therefore, we can write 
$f_{Y|X=\bar{x}}=f_{s(X^*)|X=\bar{x}}\ast f_{Y|X=\bar{x}^+},$ which is a standard convolution problem with a well-known closed form solution using the Fourier representation (see e.g. \citealt{schennachmeasurement}),  
\begin{align*} 
    f_{s(X^*)|X=\bar{x}}(v)&=\frac{1}{2\pi}\int \frac{\E[e^{\textbf{i}\xi Y_i}|X_i=\bar{x}]}{\E[e^{\textbf{i}\xi Y_i}|X_i=\bar{x}^+]}e^{-\textbf{i}\xi v} d\xi.
\end{align*}
Evaluating the density at zero yields the desired identification result.
\begin{lemma}\label{lem:u(D*)}
If Assumptions \ref{as:smoothcounterfactual}-\ref{as:deconvregularity} hold, then  $f_{s(X^*)|X=\bar{x}}(0)$ is identified as
\begin{equation*}
    f_{s(X^*)|X=\bar{x}}(0)=\frac{1}{2\pi}\int \frac{\E[e^{\bmi\xi Y_i}|X_i=\bar{x}]}{\E[e^{\bmi\xi Y_i}|X_i=\bar{x}^+]}d\xi,
\end{equation*}
where the integral on the right-hand side converges.
\end{lemma}

Figure \ref{fig:deconvolution intuition} illustrates the components of the deconvolution.  The outcome distributions were specified as a sequence of normal distributions $f_{Y|X=x},$ depicted in solid blue. Since, for $X_i>0,$ $Y_i=\text{ATT}(X_i)+y_{\bar{x}}+\epsilon_i,$ the observable solid blue densities $f_{Y|X=x}$ converge to $f_{y_{\bar{x}}+\epsilon|X=\bar{x}^+}$ as $x\downarrow \bar{x}$. So, the dotted blue density $f_{Y|X=\bar{x}^+}=f_{y_{\bar{x}}+\epsilon|X=\bar{x}^+}.$ The dashed red line is the unobserved density of $s(X_i^*)$ at the bunching point, $f_{s(X^*)|X=\bar{x}}$, which is specified here as normal as well. The solid black line is the observed density of the outcome at the bunching point, which is the convolution of the red dashed density ($f_{s(X^*)|X=\bar{x}}$) and the blue dotted density ($f_{\epsilon|X=\bar{x}^+}$). Since the images were produced using the actual convolution of the depicted densities, the dimensions illustrate the real relationship between these densities. 

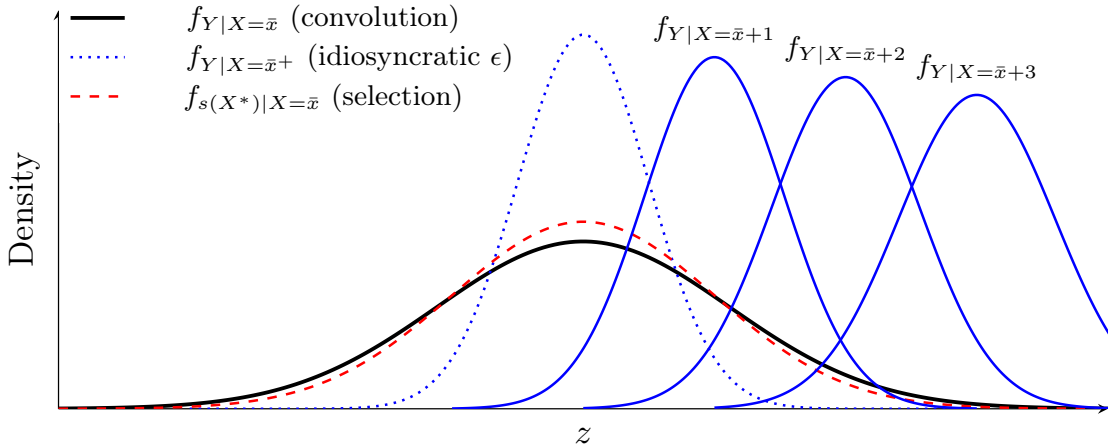
\begin{figure}[htp!]
\begin{center}
\caption{Components of the Deconvolution to Identify $f_{s(X^*)\mid X=\bar{x}}$}
\label{fig:deconvolution intuition}

\begin{tikzpicture}[scale=1.2, transform shape]
  \begin{axis}[
      width=0.8\textwidth, height=6cm,
      domain=-4:4, samples=200,
      axis lines=left,
      xlabel={$z$}, ylabel={Density},
      xmin=-4, xmax=4,
      ymin=0, ymax=0.85,
      xtick=\empty,
      ytick=\empty,
      scaled ticks=false,
      yticklabel style={/pgf/number format/fixed, /pgf/number format/precision=4},
      legend style={
        at={(axis cs:-4,0.75)},
        anchor=west,
        font=\footnotesize,
        draw=none
      },
    ]

    % 1) Convolution: σ ≈ sqrt(0.5^2 + 1^2) = 1.118
    \addplot[black, very thick, domain=-4:4, samples=200] 
      {exp(-x^2/(2*1.118^2)) / (1.118*sqrt(2*pi))};  
    \addlegendentry{$f_{Y\mid X=\bar{x}}$ (convolution)}

    % 2) Right‐limit error: σ = 0.5
    \addplot[blue, thick, dotted, domain=-3:3, samples=200]
      {exp(-x^2/(2*0.5^2)) / (0.5*sqrt(2*pi))};
    \addlegendentry{\hspace{.23in}$f_{Y\mid X=\bar{x}^+}$ (idiosyncratic $\epsilon$)}

    % 3) Selection density: σ = 1
    \addplot[red, thick, dashed, domain=-4:4, samples=200]
      {exp(-x^2/2) / sqrt(2*pi)};
    \addlegendentry{$f_{s(X^*)\mid X=\bar{x}}$ (selection)}

    % 4) x=1: μ = 1, σ = 0.532
    \addplot[blue, thick, domain=-1:3, samples=200]
      {exp(-(x - 1)^2 / (2 * 0.532^2)) / (0.532 * sqrt(2*pi))};

    % 5) x=2: μ = 2, σ = 0.564
    \addplot[blue, thick, domain=0:4, samples=200]
      {exp(-(x - 2)^2 / (2 * 0.564^2)) / (0.564 * sqrt(2*pi))};

    % 6) x=3: μ = 3, σ = 0.596
    \addplot[blue, thick, domain=1:5, samples=200]
      {exp(-(x - 3)^2 / (2 * 0.596^2)) / (0.596 * sqrt(2*pi))};

    % Labels at each mode
    \node[anchor=south,font=\footnotesize] 
      at (axis cs:1,{1/(0.532*sqrt(2*pi))}) {$f_{Y\mid X=\bar{x}+1}$};
    \node[anchor=south,font=\footnotesize] 
      at (axis cs:2,{1/(0.564*sqrt(2*pi))}) {$f_{Y\mid X=\bar{x}+2}$};
    \node[anchor=south,font=\footnotesize] 
      at (axis cs:3,{1/(0.596*sqrt(2*pi))}) {$f_{Y\mid X=\bar{x}+3}$};
  \end{axis}
\end{tikzpicture}
\end{center}
Note: The solid blue lines depict the observable densities of the outcome, which approaches the density of $y_{\bar{x}}+\epsilon|X_i=\bar{x}^+$ as $x\downarrow \bar{x},$ depicted by the dotted blue line. The dashed red line is the unobservable density of $s(X_i^*)|X_i=\bar{x},$ which is the object we want to identify. The solid black line is the observed density of $Y_i=X_i=\bar{x},$ which is the convolution of the dotted blue and the dashed red distributions. Plots were constructed with $\bar{x}=0,$  $Y_i|X_i=x\sim N(x, 0.5+0.032x),$ and $s(X_i^*)|X_i=\bar{x}\sim N(0, 1)$. The density $f_{Y| X=\bar{x}}$ was produced by convoluting the distributions of $Y_i|X_i=\bar{x}^+$ and $s(X^*_i)|X_i=\bar{x}$.
\end{figure}

Our final identification result is derived from the combination of  Equation \eqref{eq:AMEbunching} with Lemmas \ref{lem:theta} and \ref{lem:u(D*)}:
\begin{theorem} \label{thm:AMEid}
    Under Assumptions \ref{as:continuoussupport}-\ref{as:deconvregularity}, $\text{AME}_{\bar{x}}^+$ is identified as:
    \begin{equation}\label{eq:u'}
    \text{AME}_{\bar{x}}^+=\lim_{x\downarrow \bar{x}}\frac{d}{dx}\E[Y_i|X_i=x]-2\pi\theta \left(\int \frac{\E[e^{\textbf{i}\xi Y_i}|X_i=\bar{x}]}{\E[e^{\textbf{i}\xi Y_i}|X_i=\bar{x}^+]}d\xi\right)^{-1}\cdot \frac{f_X(\bar{x}^+)}{F_X(\bar{x})},
    \end{equation}
     where $\theta=\text{sgn}\left(\E[Y_i|X_i=\bar{x}^+]-\E[Y_i|X_i=\bar{x}]\right)$.
\end{theorem}
\noindent The expression is familiar in that the treatment effect is calculated by correcting the outcome variation with a scaled inverse Mills Ratio term, as is usually seen in the censoring and sample selection literatures, where some of the model components are truncated or missing below a certain threshold. In our setting the scaling factor makes use of the full distribution of the outcome observable from the data both at the bunching point and as one approaches it.

Note that when there is no endogeneity at $X_i=\bar{x}$, $s'(\bar{x})=0$. Therefore, $\theta=0$ and Equation \eqref{eq:u'} still holds. This means that the right-hand side of \eqref{eq:u'} can be used to identify $\text{AME}_{\bar{x}}^+$ while remaining agnostic about endogeneity

In the Supplementary Appendix \ref{sec:controls}, we show how identification may be obtained in the presence of controls. Specifically, all the assumptions required for identification may be done conditional on a vector of controls $Z_i,$ so that the requirements may effectively be weaker. In particular, the treatment and selection effects may change direction for different subgroups, and we can identify the $\text{AME}_{\bar{x}}^+(Z_i)$ to study heterogeneous treatment effects. We propose estimators in the case with discrete controls (Section \ref{sec:discretecontrols}), continuous controls (Section \ref{sec:continuouscontrols}), and when the vector of controls is large and with mixed continuous and discrete controls (Section \ref{sec:discretecontinuouscontrols}). 

\subsection{The selection variable \texorpdfstring{$X_i^*$}{}} \label{sec:xstar}

While the existence of a selection variable $X_i^*$ satisfying Equation \eqref{eq:xbarmax} is without loss of generality, the identification result of Theorem \ref{thm:AMEid} relies upon assumptions made about $X_i^*$. To guide researchers in assessing the plausibility of these assumptions, in this section we illustrate how $X^*_i$ relates to well-defined empirical quantities.

We begin with a parametric choice model that leads to a particularly simple expression for $X_i^*$.
\begin{example}[Isoelastic model of constrained choice] \label{ex:cigarettes}
Each individual chooses the number of units of a good or service $x$ to consume at price $p$ along with the quantity of a numeraire good $r$, subject to the budget constraint $W_i = px + r$ and the non-negativity constraints $x\geq 0$ and $r\geq0.$ Note that this model can describe our application, where $x$ refers to the choice of cigarettes smoked per day, $p$ is the price of cigarettes, and $r$ is spending on other consumption categories. Consider the isoelastic family of utility functions: \begin{align*}
V(x,r;\rho_i) =\left\{\begin{tabular}{ll}
$(1+\rho_i)^{\gamma} \left( \frac{(1+x)^{1-\gamma} - 1}{1-\gamma} \right) + r,\quad $ &if $\gamma >0, \medspace \gamma \neq 1$
\\
$(1+\rho_i)\log(1+x) + r,\quad $ &if  $  \gamma = 1$.
\end{tabular}\right.
\end{align*}
\noindent The parameter $\gamma$ modulates the degree of concavity of the utility function. The individual-level parameter $\rho_i$ can be interpreted as the preference for $x$ relative to the numeraire good. For individuals with $\rho_i<-1$, $x$ is a ``bad,'' while for individuals with $\rho_i>-1$, $x$ is a ``good.'' In this case, define 
$$X^*_i = \frac{1+\rho_i}{p^{1/\gamma}}-1,$$
which implies that $X_i^*=X_i$ if $\rho_i\geq p^{1/\gamma}-1.$ Otherwise, the individual chooses $X_i=0.$ This implies that at the bunching point, we have individuals that consider $x$ a ``bad'' and others who consider $x$ a ``good,'' but whose positive preferences for $x$ are not strong enough to overcome the opportunity cost of foregoing $r,$ given its price. 

If the price $p=1$, then $X^*_i = \rho_i$, so the selection variable is exactly the preference parameter $\rho_i$. For instance, this is the case in time use models, where $x$ refers to the number of hours in the day spent on a given activity (e.g. time watching TV, as in \citealt{ccn_metrics}).\footnote{In time use models, $r$ refers to the remaining activities adding up to $W_i=24$ hours per day, so that $x+r=24$. In this setting, the budget $W_i$ is constant across individuals, and $p=1$ because individuals trade-off $x$ and $r$ at the rate 1-to-1.} In this example, individuals with $-1 < \rho_i \leq 0$ are choosing $X_i=0$ but still value watching TV positively, they just do not value it as highly as they value the alternative use of their time.
\end{example}

In Appendix \ref{sec:economic}, we extend this example beyond the isoelastic family. We show that in a class of utility maximization models that feature a scalar preference parameter $\rho_i$, we can write $X_i^*=h(\rho_i)$ where $h$ is a strictly increasing and differentiable function. The intuition is that with scalar heterogeneity, if $\rho_i$ pins down $i$'s marginal rate of substitution between $X_i$ and the numeraire good when $X_i=0$, then it also pins down $X_i$.

In Appendix \ref{sec:generalization AME}, we show that any monotonic differentiable transformation of $X_i^*$ yields precisely the same constructive estimand that identifies the parameter $AME_{\bar{x}}^+$ in Theorem \ref{thm:AMEid}. Specifically, we prove that if $X_i^*=h(\rho_i)$ for any strictly increasing and differentiable $h$, the assumptions that establish Theorem \ref{thm:AMEid} can be made directly  on $\rho_i=h^{-1}(X_i^*)$, rather than on $X_i^*$.  Our result does not require that $\rho_i$ be observable, nor that the function $h$ be known to the econometrician, and holds whether or not one posits a utility maximization model. Thus, to make use of Theorem \ref{thm:AMEid} to identify $AME_{\bar{x}}^+$, one needs only to believe that there exists \textit{some} $\rho_i$ that satisfies the Assumptions \ref{as:smono}-\ref{as:deconvregularity}, and some monotonic and differentiable transformation $h$ such that $X_i = \max\{\bar{x},h(\rho_i)\}$. This result allows the researcher to make assumptions about $\rho_i$ rather than about the more abstract object $X_i^*$ without any loss of generality, while retaining the same identification equations and estimators.

In Appendix \ref{sec:unrestrictedheterigeneity}, we consider a general utility maximization framework in which heterogeneity in individuals' choices are not necessarily explained by a single scalar parameter. We show that $X_i^*$ can be defined for those individuals with $X_i=\bar{x}$ as the marginal utility (after substituting budget constraints and profiling the utility function over any additional choice variables), evaluated at the bunching point. Those individuals who are indifferent between $X_i=\bar{x}$ or an amount slightly larger than the bunching point have $X_i^*=\bar{x}$ exactly, while those who strictly prefer $X_i=\bar{x}$ to a value just above it have $X_i^* < \bar{x}$. One can then translate the key aspects of Assumptions \ref{as:smono}-\ref{as:deconvregularity} in terms of comparisons of observations with different degrees of indifference towards veering away from the bunching point. 

Finally, in Appendix \ref{sec:noxstar}, we show that $X_i^*$ can be constructed without the need for any underlying choice model. As discussed in the introduction, bunching has been observed in some examples where the treatment variable is not a clear function of individual choices (e.g., \cite{caetano2018identifying} study of the effects of neighborhood crime), and this type of construction can be useful in such cases.  Concretely, if $\epsilon_i \indep X_i $ for $X_i>\bar{x},$ and $s(x)$ is sufficiently smooth, then $s$ can be extrapolated to $x\leq \bar{x}$. In this case, an $X_i^*$ satisfying $\epsilon \indep X_i^*|X_i=\bar{x}$ can always be constructed to satisfy the identification restrictions. 

While the above results are motivated by settings in which bunching occurs at the boundary of the support of the treatment variable, the following example discusses how $X_i^*$ emerges naturally in settings with interior bunching at a kink. In such settings, there is no hard constraint that $X_i \ge \bar{x}$. Rather, Equation \eqref{eq:xbarmax} emerges from a discontinuous change in individuals' incentives at $X_i=\bar{x}$.
\begin{example}[Interior bunching at a kink] \label{ex:kink}
 Each individual chooses $x$ (e.g., income) by maximizing a utility function  $u(x,t;A_i)$ that is strictly quasi-concave in $x$ and decreasing in $t$ (e.g., tax liability). Although the vector of individual characteristics $A_i$ influences the individual's choice, it need not be observed by the researcher. Suppose that $t$ as a function of $x$ exhibits a convex kink at $\bar{x}$, so that costs increase faster with $x$ when $x > \bar{x}$ than they do when $x < \bar{x}$.  \citet{goff2020treatment} shows that in this setting optimal choice can be written as
$$\tilde{X}_i = \begin{cases}
X_i(0) & \textrm{ if } X_i(0) < \bar{x}\\
\bar{x} & \textrm{ if } X_i(1) \le \bar{x} \le X_i(0)\;,\\
X_i(1) & \textrm{ if } X_i(1) > \bar{x}\\
\end{cases}$$
where $X_i(0)$ and $X_i(1)$ are, respectively: the counterfactual choices that the individual would make if the budget function to the left of the kink applied globally, or if the budget function to the right of the kink applied globally.  

If we consider only the observations with $X_i\geq \bar{x}$, then the mapping from Equation \eqref{eq:xbarmax}, $X_i=\max\{X^*_i,\bar{x}\}$, holds by defining $X^*_i = X_i(1)$. Thus, $X_i^*$ could be interpreted as the counterfactual choice that would be made if the budget function to the right of the kink applied globally. Conversely, if we consider only the observations with $X_i\leq \bar{x}$, we would have $X_i=\min\{X^*_i,\bar{x}\}$ by defining $X^*_i = X_i(0)$, and $X_i^*$ could be interpreted as the counterfactual choice that would be made if the budget function to the left of the kink applied globally.
\end{example}

\section{Identification of causal effects away from the bunching point} \label{sec:global}
Given identification of $\text{AME}^+_{\bar{x}}$ demonstrated in Section \ref{sec:idfrombunching}, we consider now how global effects $\text{ATT}(x)$ may be identified by extrapolating the information available near the bunching point. The extension requires a sufficient degree of smoothness of the counterfactual function near $\bar{x},$ which is guaranteed by the following assumption. 
\begin{assumption}\label{as:analytic}  The counterfactual function $x\mapsto \E[Y_i(\bar{x})|X_i=x]$ is real analytic on an interval $I'=(\bar{x},\bar{x}+\varepsilon_5)$ for some $\varepsilon_5>0$.
\end{assumption}
\noindent Functions that are real analytic on $I$ are infinitely differentiable functions whose Taylor series around any point $x \in I$ converges pointwise to the function in a neighborhood of $x$. This class includes all functions that locally behave like polynomial, exponential, trigonometric, hyperbolic, logarithmic, or inverse trigonometric functions, as well as compositions, ratios, and roots of these. A sufficient condition for Assumption \ref{as:analytic} is that the observable function $\mathbbm{E}[Y_i|X_i=x]$ and the $\text{ATT}(x)$ function are both analytic on $I$. This may be a more appealing argument than reasoning about the properties of the selection function $s(x),$ since $\E[Y_i|X_i=x]$ is identifiable, and it may be plausible to assume that the dose-response function $\text{ATT}(x)$ is sufficiently smooth in some applications. 

Recall that $\text{ATT}(x)=m(x)-s(x)$. The following theorem establishes the desired local extrapolation. 
\begin{theorem}\label{thm:localexpansion} (Local 
 ATTs) If Assumptions \ref{as:continuoussupport}-\ref{as:analytic} hold, then there exists an $\varepsilon>0$ such that, for all $x\in I_{\varepsilon}:=(\bar{x},\bar{x}+\varepsilon)\subset I',$ 
\begin{equation} \label{eq:taylorexpression}
    \text{ATT}(x)=m(x)-\sum_{k=1}^\infty s^{(k)}(\bar{x}^+)\cdot  \frac{(x-\bar{x})^k}{k!},
\end{equation}
 where all the derivatives and limits in the equation above are well defined. Moreover, for any $K \ge 0,$ there exists a value $\zeta_{\bar{x}}(x)\in (\bar{x},x]$ such that,
\begin{equation*}
    R_{k}(x-\bar{x}):=\!\!\sum_{k=K+1}^{\infty}\!\! s^{(k)}(\bar{x}^+) \cdot \!\frac{(x-\bar{x})^{k}}{k!}=s^{(k)}(\zeta_{\bar{x}}(x)) \cdot \!\frac{(x-\bar{x})^{K+1}}{\hspace{-.4cm}(K+1)!}.
\end{equation*} 
\end{theorem}
The second part of Theorem \ref{thm:localexpansion} offers a practical strategy for finite approximations to Equation \eqref{eq:taylorexpression}. Explicitly, for a suitably large $K,$ $\text{ATT}(x)\approx m(x)-\sum_{k=1}^K s^{(k)}(\bar{x}^+)\cdot  \frac{(x-\bar{x})^k}{k!}
$. The error of the $K$-th degree approximation does not exceed $\sup_{x \in I_{\varepsilon}}|s^{(K+1)}(x)|\cdot \varepsilon^{K+1}/(K+1)!$. Since $(K+1)!$ has supra-exponential growth, even if $\varepsilon$ and the high-order derivatives are very large, the approximation error decays quickly.

 Equation \eqref{eq:taylorexpression} indicates that if all the derivatives $s^{(k)}(\bar{x}^+)$ are identifiable, then the $\text{ATT}(x)$ is identifiable in a neighborhood of $\bar{x}$. This implies that extrapolations near the bunching point are possible, provided $\theta$ is identified and the $s^{(k)}(\bar{x}^+)$ are known for all $k \ge 1$. By differentiating Equation \eqref{eq:changeofvariables} with respect to $x$, we can see that this is indeed possible since the density $f_{s(X)|I}$ (and hence its derivatives) is identified on $s(I_{\epsilon})$.
\begin{corollary} \label{cor:extrapolation}
    If Assumptions \ref{as:continuoussupport}-\ref{as:analytic} hold and $f_{X^*}(x)$ is infinitely differentiable at $x=\bar{x}$, then $\text{ATT}(x)$ is identified for each $x \in I_\epsilon$.
\end{corollary}
Corollary \ref{cor:extrapolation} guarantees the identification of the $\text{ATT}(x)$ in a neighborhood of the bunching point. In practice, finite approximations may be used to approximate the value. For example, a first-degree approximation is simply 
$$
\text{ATT}_{\bar{x}}(x)\approx \E[Y_i|X_i=x]-\E[Y_i|X_i=\bar{x}^+]-2\pi\theta \left(\int \frac{\E[e^{\textbf{i}\xi Y_i}|X_i=\bar{x}]}{\E[e^{\textbf{i}\xi Y_i}|X_i=\bar{x}^+]}d\xi\right)^{-1} \frac{f_X(\bar{x}^+)}{F_X(\bar{x})}(x-\bar{x}),$$ and a second-degree approximation adds the term
\begin{equation*}
\left(s'(0)\frac{f'_{X}(\bar{x}^+)}{f_X(\bar{x}^+)}-s'(0)^2\frac{f'_{s(X^*)|X=\bar{x}}(0)}{f_{s(X^*)|X=\bar{x}}(0)}\right)\frac{(x-\bar{x})^2}{2},
\end{equation*}
where $s'(0)$ is the correction term in Equation \eqref{eq:u'}, and $f'_{s(X^*)|X=\bar{x}}(0)/f_{s(X^*)|X=\bar{x}}(0)$ 
 is equal to $\left(\int \E[e^{\textbf{i}\xi Y_i}|X_i=\bar{x}^+]^{-1}{\E[e^{\textbf{i}\xi Y_i}|X_i=\bar{x}]}d\xi\right)^{-1}\left(\int \textbf{i}\xi\E[e^{\textbf{i}\xi Y_i}|X_i=\bar{x}^+]^{-1}\E[e^{\textbf{i}\xi Y_i}|X_i=\bar{x}]d\xi\right)$. This expression may seem complex but, in practice, one would have already estimated $f_{s(X^*)|X=\bar{x}}(0)$ for a first-degree approximation, and standard deconvolution packages often automatically provide the first derivative $f_{s(X^*)|X=\bar{x}}'(0)$ at the same time.

 One limitation of Theorem \ref{thm:localexpansion} and Corollary \ref{cor:extrapolation} is that the interval $I_\varepsilon$ could in principle be quite small. A sufficient condition for the ATT to be defined far away from the bunching point is that the higher order derivatives decay suitably fast with $k$.

\begin{corollary} \label{cor:extrapolationfar}
    Suppose the assumptions of Corollary \ref{cor:extrapolation} hold and that $\limsup_{k \to \infty} \left| \frac{s^{(k)}(\bar{x})}{k!} \right|^{1/k} < 1/M$, then $\text{ATT}(x)$ is identified for all $x \in [\bar{x},\bar{x}+M]$.
\end{corollary}
\noindent Corollary \ref{cor:extrapolationfar} specifies ``how far'' one can extrapolate from the derivatives of $s^{(k)}(\bar{x})$ to obtain $\text{ATT}(x)$. Specifically, if the $|s^{(k)}(\bar{x})|$ are bounded by $M^{-k}\cdot k!$ uniformly over $k$ for some $M,$ then the $\text{ATT}(x)$ identification can be extrapolated as far as $\bar{x}+M$.

\section{Estimation}\label{sec:estimation}

For a sample $\{(Y_i,X_i)', i=1,\dots, n\},$ the average marginal effect at the bunching point may be estimated following Equation \eqref{eq:u'}. Specifically, we use the following formulas:
 \begin{equation*}\widehat{\text{AME}}_{\bar{x}}^+=\hat{m}'(\bar{x}^+)-\hat{\theta} \cdot \hat{f}_{s(X^*)|X=\bar{x}}(0)^{-1}\cdot \frac{\hat{f}_X(\bar{x}^+)}{\hat{F}_X(\bar{x})},
    \end{equation*}
where $\hat{m}'(\bar{x}^+)$ is an estimator of $\lim_{x\downarrow\bar{x}}\frac{d}{dx}\E[Y_i|X_i=x],$ and $\hat{\theta}=\text{sgn}(\hat{\E}[Y_i|X_i=\bar{x}^+]-\hat{\E}[Y_i|X_i=\bar{x}])$. A first-degree approximation following Corollary \ref{cor:extrapolation} uses the estimator
\begin{equation*}\widehat{\text{ATT}}(x)=\hat{E}[Y_i|X_i=x]-\hat{E}[Y_i|X_i=\bar{x}^+]-\hat{\theta} \cdot \hat{f}_{s(X^*)|X=\bar{x}}(0)^{-1}\cdot \frac{\hat{f}_X(\bar{x}^+)}{\hat{F}_X(\bar{x})}\cdot (x-\bar{x}),
    \end{equation*}
    and analogously for a second-degree approximation. All the components of these formulas are standard objects frequently studied in econometrics. We discuss next how each component can be estimated.

The terms  $\hat{F}_X(\bar{x})$ and $\hat{\E}[Y_i|X_i=\bar{x}]$ may be estimated with simple averages:
    \begin{equation*}
    \hat{F}_{X}(\bar{x})=\frac{1}{n}\sum_{i=1}^n\bm{1}(X_i=\bar{x}), \text{ and }\hat{\E}[Y_i|X_i=\bar{x}]=\hat{F}_{X}(\bar{x})^{-1}\cdot \frac{1}{n}\sum_{i=1}^nY_i\bm{1}(X_i=\bar{x}).
\end{equation*}
    
 The terms $\hat{\E}[Y_i|X_i=\bar{x}^+]$ and   $\hat{m}'(\bar{x}^+)$ are standard non-parametric regression boundary quantities. Estimation of these objects has been extensively researched in the statistics literature on local polynomial estimators, and in the Regression Discontinuity Design and Regression Kink Design literatures in economics. In line with classical methods in this literature and with the vast majority of applications in boundary regression estimation, we propose using a local linear regression of $Y_i$ onto $X_i$ at $X_i=\bar{x},$ using only observations such that $X_i>\bar{x},$ for its superior properties of bias reduction and variance control at the boundary over other methods.\footnote{See \cite{ruppert1994multivariate} and \cite{fan2018local}, and also \cite{cheruiyot2020local}. See \cite{imbens2019optimized} and citations therein for recent proposals which may be superior to local linear estimators.} The intercept coefficient of this regression is $\hat{\E}[Y_i|X_i=\bar{x}^+],$ and the slope coefficient is $\hat{m}'(\bar{x}^+).$ This may be executed using any package for local linear regression available in standard statistical software (R, STATA, etc.). 
    
    Explicitly, for a bandwidth $h_1>0$ and a kernel function $k_1,$\footnote{The triangular kernel $k_1(\nu)=(1-|\nu|),$ where $|\nu|\leq 1$ is recommended for boundary regressions such as this (\citealt{cheng1997automatic}).} solve the problem
    \begin{equation*}
        \hat{b}_0,\hat{b}_1=arg\min_{b_0,b_1}\sum_{i=1}^n(Y_i-b_0-b_1(X_i-\bar{x}))^2\cdot k_1\left(\frac{X_i-\bar{x}}{h_1}\right)\bm{1}(X_i>\bar{x}),
    \end{equation*}
    then $\hat{\E}[Y_i|X_i=\bar{x}^+]=\hat{b}_0,$ and  $\hat{m}'(\bar{x}^+)=\hat{b}_1.$
    This estimator has a closed-form expression, which is commonly found in nonparametric econometrics textbooks, e.g. \cite{li2007nonparametric}. 
    Note that the term $\hat{E}[Y_i|X_i=x]$ is not a boundary quantity, but it may be estimated analogously, with a local linear regression of $Y_i$ onto $X_i$ at $x,$ using only observations such that $X_i>\bar{x}.$
    
     The term $\hat{f}_X(\bar{x}^+)$ is a boundary density. As in the case of nonparametric boundary regression discussed above, the tendency for higher bias in this scenario necessitates the use of corrective methods, such as the use of local polynomial estimators. We recommend the approach recently proposed in \cite{pinkse2023estimates},\footnote{Other estimators of boundary densities include \cite{hjort1996locally}, \cite{loader1996local},  \cite{cheng1997automatic}, \cite{zhang1998kernel}, \cite{bouezmarni2010nonparametric} and \cite{cattaneo2020simple}.}  which has two important properties which are of great value in our case, and which are not found in other estimators currently available. First, this estimator achieves the same rates of bias convergence at the boundary that is normally achieved in interior points. Second, the density estimator is never negative, a situation which would be complicated to address in our case. Additionally, the estimators have simple closed-form expressions, requiring only the choice of a bandwidth tuning parameter, $h_2$.

Following \cite{pinkse2023estimates}, let $L_{X}(x)=\log f_{X}(x).$ We begin by estimating $L'_{X}(\bar{x}^+)$ as\footnote{This estimator is derived from applying Example 1 with $z=0$ to Equation (2) in \cite{pinkse2023estimates}.}
\begin{equation*}
    \hat{L}'_{X}(\bar{x}^+)=-\frac{\sum_{i=1}^n\left(1-{2(X_i-\bar{x})}/{h_2}\right)\bm{1}(\bar{x}< X_i\leq \bar{x} +h_2)}{\sum_{i=1}^n(X_i-\bar{x})\left(1-{(X_i-\bar{x})}/{h_2}\right)\bm{1}(\bar{x}< X_i\leq \bar{x}+ h_2)}.
\end{equation*}
This, then, allows us to estimate the density at the boundary as
\begin{equation*}
    \hat{f}_{X}(\bar{x}^+)=\frac{\frac{1}{nh_2}\sum_{i=1}^nk_2\left(\frac{X_i-\bar{x}}{h_2}\right)}{\int_0^1k_2(\nu)\exp(\hat{L}_{X}'(\bar{x}^+)\nu h_2)d\nu}, 
\end{equation*}
which can be calculated for many standard positive kernel functions $k_2$. For example, as in Example 5 of \cite{pinkse2023estimates}, when $k_2$ is the Epanechnikov kernel $k_2(\nu)=3/4(1-\nu^2)$ (which is the kernel recommended for boundary estimation in that paper), the denominator is equal to 
\begin{equation*}
    \frac{3}{2}\cdot\frac{2+\hat{L}_{X}'(\bar{x}^+)^2h_2^2-e^{\hat{L}_{X}'(\bar{x}^+)h_2}(2-2\hat{L}_{X}'(\bar{x}^+)h_2)}{\hat{L}_{X}'(\bar{x}^+)^3h_2^3}.
\end{equation*}
This estimator is available in packaged form in standard statistics software and can be implemented by simply restricting the sample to observations such that $X_i>\bar{x}$ and then using the package to estimate the density of $X_i$ at $X_i=\bar{x}.$ Incidentally, the same package also provides the estimator of the derivative $f_{X}'(\bar{x}^+),$ that can be used in the second-order approximation of the $\text{ATT}(x)$ (see Section \ref{sec:global}).

 The final term $\hat{f}_{s(D^*)|X=0}(0)$ is a standard deconvolution estimator. We follow the estimator described in \cite{schennachmeasurement}, which is the focus of an extensive literature, although there are many alternative proposals which are also referenced therein. 

We first write $\hat{\E}[e^{\textbf{i}\xi Y_i}|X_i=\bar{x}^+]$ as  a local linear regression of $e^{\textbf{i}\xi Y_i}$ onto $X_i$ at $X_i=\bar{x}$ using only observations such that $X_i>\bar{x}.$ To do this, for a matrix $\textbf{x}$ with rows $(1,(X_i-\bar{x}))'$ and a diagonal matrix $\textbf{k},$ with diagonal elements $k_3((X_i-\bar{x})/h_3)\bm{1}(X_i>\bar{x}),$ where $k_3$ is the triangular kernel, and $\textbf{e}_1=(1,0)',$ define the vector $A(\xi)=(e^{\textbf{i}\xi Y_1},\dots,e^{\textbf{i}\xi Y_n})',$ and program the function
\begin{equation*}
    \hat{\phi}(\xi)=\textbf{e}_1(\textbf{x}'\textbf{k}\textbf{x})^{-1}\textbf{x}'\textbf{k}A(\xi).
\end{equation*}

This is then imputed into a standard convolution estimator such as, for example:
\begin{equation*}
    \hat{f}_{s(X^*)|X=\bar{x}}(0)=\frac{1}{nh_4}\sum_{i=1}^ng(Y_i)\bm{1}(X_i=\bar{x}), 
\end{equation*}
with
\begin{equation*}
    g(Y_i)=\frac{1}{\hat{F}_X(\bar{x})\cdot 2\pi} \int  e^{\textbf{i}\xi Y_i}\frac{\phi_K(h_4\xi)}{A(\xi)} d\xi,
\end{equation*}
where $\phi_{k_4}(h_4\xi)=\int k_4(\nu)e^{\textbf{i}h_4\xi \nu}d\nu$ is the Fourier transform of the kernel $k_4$ evaluated at $h_4\xi.$ 

The nonparametric estimators just described require the choice of the bandwidth tuning parameters: $h_1,h_2, h_3$ and $h_4$, which modulate the bias-variance trade-off. This choice is rather important, and the subject of a great deal of interest in the nonparametrics estimation literature. At this stage, our recommendation is that if an optimal method for bandwidth selection exists for the specific estimator used at a given step, then it should be used.\footnote{For the selection of $h_1$ and $h_3$, \cite{ruppert1995effective} propose an optimal bandwidth estimator for the local linear regression, and this or similar approaches for bandwidth selection are usually offered in standard local linear regression packages. There are many proposals for improvement of bandwidth selection in the Regression Discontinuity Design literature which may be adapted to this context, see, e.g. \cite{imbens2012optimal}, \cite{arai2016optimal}, \cite{arai2018simultaneous} and \cite{calonico2020optimal}. For $h_2,$ the optimal bandwidth is $h = (72/(nf_X(0)^+\beta_2^2))^{1/5}$, which may be calculated following Example 6 in \cite{pinkse2023estimates}. $\beta_2$ can be estimated using a pilot estimate of $\hat{f}_X(0)_+,$ and both these terms are then added to the formula of the optimal bandwidth. Nevertheless, although theoretically sound, this method has not yet been studied. Thus, in practice, we recommend testing several bandwidths around this benchmark and looking for robustness of the results. For $h_4,$ consider the several approaches studied in \cite{delaigle2004practical}.} However, it is possible that the optimal bandwidths for $\widehat{\text{AME}}_{\bar{x}}^+$ are not the optimal bandwidths for each of the separate components. 

Additionally, there is an interest in the use of bias correction techniques for inference in the Regression Discontinuity Design literature which may have relevance in this context as well (e.g. \cite{calonico2014robust}, \cite{noack2024bias}, \cite{he2020wild}, \cite{armstrong2020simple} and citations therein). This is because,  if optimal bandwidths are used,  $\widehat{\text{AME}}_{\bar{x}}^+$ will likely be asymptotically biased.  We leave these questions for future research. 

\begin{remark}\label{rem:efficiency}(Improving efficiency by parameterizing the outcome distributions)\label{sec:parametric}
We can estimate $\hat{f}_{s(X^*)|X=\bar{x}}(0)$  more efficiently  via deconvolution if either of the distributions $f_{Y|X=\bar{x}}$ or $f_{Y|X=\bar{x}^+}$ are assumed to be of a known parametric family. Since the outcome is observed both at and above the bunching point, this assumption may not be overly speculative, and it is directly testable using standard Kolmogorov-Smirnoff tests or \cite{goldman2018comparing}. 

Assuming a parametric distributional class for the outcome is helpful because it allows the relevant characteristic function to be estimated less noisily. For example, in the empirical analysis in Section \ref{sec:application}, we assume that $Y_i|X_i=\bar{x}^+\sim N(\E[Y_i|X_i=\bar{x}^+],\sigma^2)$, while we allow $f_{Y|X=\bar{x}}$ to be fully nonparametric. This normality assumption appears to be a good approximation, as can be seen in Figure \ref{fig:condition_kdensities}. 

We estimate the parameter $\sigma^2$ in two steps. First, we restrict the sample to observations with $X_i>\bar{x},$ and for each $i,$ we predict $\hat{Y}_i$ via local linear regression estimated on the non-bunched sample. Then, we form the squared residual $\hat{\epsilon}_i^2=(Y_i-\hat{Y}_i)^2.$ Second, we fit a local linear model of $\hat{\epsilon}_i^2$ on $X_i$ estimated on the $X_i>\bar{x}$ subsample, and we estimate $\hat{\sigma}^2$ as the predicted value of this regression at $X=\bar{x}.$ We note that the variance of $\epsilon_i|X_i=\bar{x}^+$ is the same as the variance of $Y_i|X_i=\bar{x}^+.$ We then deconvolve the distribution $N(0,\hat{\sigma}^2)$ from $\hat{F}_{Y|X=\bar{x}}$ using standard methods.\footnote{Specifically, we employ the ``decon'' package in R.}
\end{remark}

\begin{remark} (Normal case) If $Y_i|X_i=\bar{x}$ and $Y_i|X_i=\bar{x}^+$ are both normal, then it is straightforward to verify that $s(X_i)|X_i=\bar{x}\sim N(\E[Y_i|X_i=\bar{x}]-\E[Y_i|X_i=\bar{x}^+],\sigma^2_{Y_i|X_i=\bar{x}}-\sigma^2_{Y_i|X_i=\bar{x}^+}),$ where $\sigma^2_V$ is the variance of $V.$ 

Although our identification and estimation results do not require that the distributions of $Y_i|X_i=\bar{x}$ and $Y_i|X_i=\bar{x}^+$ be known a priori, this simple expression may be used as a robustness check of the nonparametric estimation results whenever the distributions appear to be normal, as in the case of our application.
\end{remark}
%%%%%%%%%%%%%%

\section{Application: the effect of maternal smoking on birth weight\label{sec:application}}

We apply our method to estimate the marginal effect of maternal smoking during pregnancy on birth weight. This question is important for both economics and epidemiology, given that maternal smoking during pregnancy is recognized as a critical modifiable risk factor for low birth weight (\citealt{almond2005costs}). Low birth weight not only results in immediate societal costs, but also has significant long-term consequences for children's later-life outcomes (\citealt{black2007cradle}).

Our application uses the data set from \cite{almond2005costs}, which is also used in \cite{caetano2015}. These data, from the U.S. National Center for Health Statistics, include both maternal cigarettes smoked daily during pregnancy (our treatment) and birth weight in grams (our outcome) for over 430,000 mother-child pairs. These data also include many additional covariates/controls which \cite{almond2005costs} use to estimate an effect of maternal smoking on birth weight of around -200 grams, assuming selection on observables. \cite{caetano2015} uses these data to illustrate the discontinuity test, showing that selection-on-observables does not seem to be a valid assumption using \cite{almond2005costs}'s very detailed control specification.

We drop premature births (gestation < 36 weeks) as well as birth weight outliers (a small subset of observations with very high, >6 kg, or very low, <1 kg, full-term birth weights) from the data.\footnote{Our estimates barely change when we extend the sample to include premature babies and outliers.} About 81\% of the mothers in our analysis sample smoke zero cigarettes daily, about 11\% smoke between 1 and 10 cigarettes, and 99.95\% smoke 40 cigarettes or less. Figure \ref{fig:E[Y|X]} shows $\E[Y_i|X_i=x]$, the average birth weight among mothers smoking different amounts in our sample. 
The evidence of a discontinuity in $\E[Y_i|X_i=x]$ at $x=0$ is clear. While the average birth weight among mothers who smoke zero cigarettes is 3,499 grams, the average birth weight for mothers who smoke one cigarette is 3,338 grams, with the analogous quantity for mothers smoking 2--5 cigarettes ranging between 3,278 and 3,330. The rich controls in \cite{almond2005costs} can account for only 55 out of the 161 grams (3,499-3,338) difference in birth weight between the children of mothers who smoke zero versus one cigarette. Thus, there remains a lot of ``room''---106 grams---for both the treatment 
\vspace{.2cm}\begin{figure}[h!]
\protect\caption{Evidence of Bunching and Discontinuity Test\label{fig:E[Y|X]}}
\vspace{-.2in}
\begin{center}
\includegraphics[scale=.85]{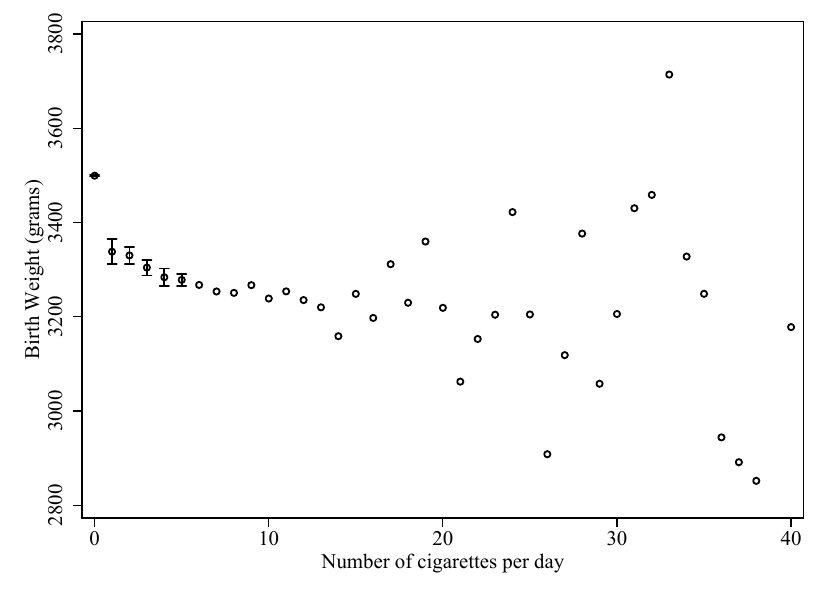}
\end{center}
\vspace{-0.2in}
\singlespace \footnotesize{Note: The figure shows $\hat{\E}[Y_i|X_i=x]$ for different values of $x$ (along with the 95\% confidence interval for $x\leq 5$).}
\end{figure} effect of cigarettes and selection on unobservables to explain this birth weight discontinuity.\footnote{The numbers cited in this paragraph are very similar to those reported in \cite{caetano2015}. That paper removes neither premature nor ``outlier'' births from the analysis sample. Similarly, including premature and outlier births in the causal analysis yields similar estimates of $\text{AME}_{0}^{+}$ and $\text{ATT}(x).$ Indeed, the full-sample results tend to be smaller in magnitude than those reported in Table \ref{tab:beta0} and Figure \ref{fig:ATT_d}, further supporting our qualitative claims.}

We estimate the marginal effect $\text{AME}_{0}^{+}$ of cigarette smoking on birth weight using the procedure outlined in Section \ref{sec:estimation}. In particular, we increase efficiency by assuming that $Y_i|X_i=\bar{x}^+$ is normally distributed (see Remark \ref{sec:parametric}). This assumption is empirically well-supported in our setting; Figure \ref{fig:condition_kdensities} makes clear that the conditional distributions $Y_i|X_i=x$ are all very close to normal with nearly the same variance for $0<x\leq 5$.\footnote{Figure \ref{fig:qq_plots} in the Supplementary Appendix \ref{ap:emp_supp} presents QQ-plots as additional graphical evidence that these conditional distributions are approximately normal.} In addition to simplifying estimation, Figure \ref{fig:condition_kdensities} can also be interpreted as indirect evidence for item (iii) of Assumption \ref{as:deconvregularity}, that $\epsilon_i \indep X_i^*|X_i=0$ (cf. Proposition \ref{prop:additive} in Appendix \ref{sec:noxstar}). To see this, note that the distribution of $\epsilon_i|X_i=x$ is just a horizontal shift of the distribution of $Y_i|X_i=x.$ Therefore, the fact that the $f_{Y|X=x}$ look like simple horizontal shifts for $0<X_i\leq 5$ is direct evidence that $\epsilon_i\indep X_i|0<X_i\leq 5$ as well as indirect evidence that this pattern may continue for $X_i^*\leq 0$ (though, of course, this cannot be directly verified). 
\begin{figure}[h!]
\protect\caption{Birth Weight Distributions Conditional on Maternal Smoking\label{fig:condition_kdensities}}
\vspace{-.2in}
\begin{center}
\includegraphics[scale=0.8]{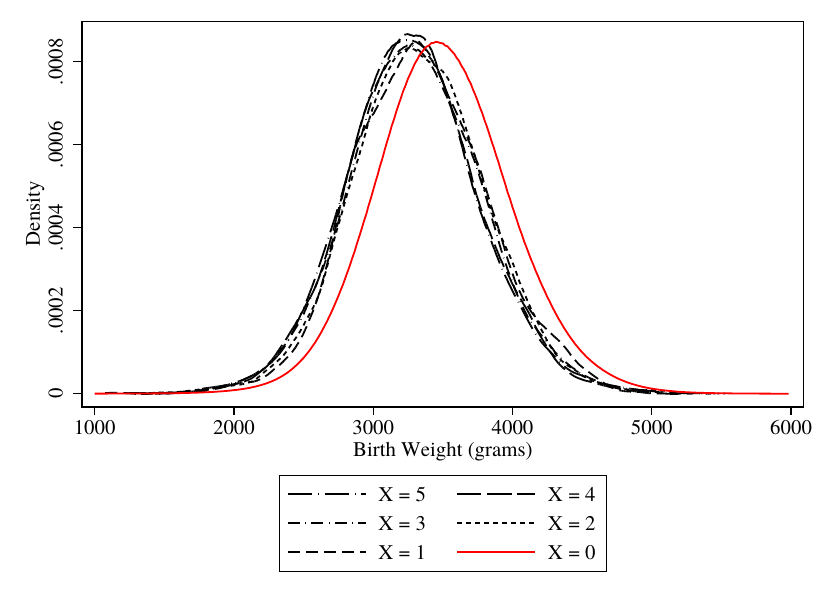}
\end{center}
\vspace{-0.2in}
\singlespace \footnotesize{Note: The figure plots the kernel density estimates (Epanechnikov, bandwidth=100) of birth weight conditional on maternal cigarettes smoked per day ($X_i$).}
\end{figure}

For context, in Figure \ref{fig:condition_kdensities} we also show the conditional distribution $Y|X=0$. 
While we estimate the variance of $Y_i|X_i=\bar{x}^+$ allowing for a trend in the analogous variances of $Y_i|X_i=x$ for $x>0$, we obtain nearly identical results simply using the variance of $Y|X=1$ as our estimator. This is not surprising given the evident homoskedasticity in the figure.

Estimation requires that we select several different bandwidths. We use a bandwidth of 4 cigarettes for both $\hat{\E}[Y_i|X_i=0^+]$ and $\hat{f}_X(0^+).$ However, we find very similar results using alternative, reasonable bandwidth choices for these objects. The selection of a bandwidth for $\hat{m}'(0^+)$ is more consequential for the standard error of our estimates, thus we report $\widehat{\text{AME}}_{0}^{+}$ for several bandwidth choices.

Table \ref{tab:beta0} presents our main results. We estimate the marginal effect of smoking on birth weight at zero to be around -8 grams. We estimate $s'(0^+),$ the selection bias around zero, to also be around -8 grams. The endogeneity term is quite precisely estimated in our application -- most of the sampling variation in $\widehat{\text{AME}}_{0}^{+}$ comes from sampling variation in the estimated slope of $\E[Y_i|X_i=x]$ as $x\downarrow 0$. 

\vspace{.3cm}
\begin{table}[h!]
\begin{singlespace}
\protect\caption{Main Results: The Average Marginal Effect of Smoking Near Zero Cigarettes\label{tab:beta0}}
\vspace{-0.2in}
\begin{center}
{\small \begin{tabularx}{\textwidth}{CCCCCC}
\toprule

& \multicolumn{5}{c}{bandwidth used in $\hat{m}'(0^+)$} \tabularnewline
\midrule\addlinespace[1.5ex]

&\(h =\) 4&\(h =\) 5&\(h =\) 6 &\(h =\) 7 &\(h =\) 8 \tabularnewline \addlinespace[.005in]
\midrule \midrule \(\widehat{\text{AME}}_{0}^{+} \)&-8.42&-10.57&-8.77 & -7.97 & -7.44  \tabularnewline \addlinespace[.005in]
&(8.26)&(5.47)&(3.20)&(2.74)&(2.37) \tabularnewline \addlinespace[.005in]

\bottomrule \addlinespace[1.5ex]

\end{tabularx}
}
\end{center}
\vspace{-.2in}
\footnotesize Note: $X_i$ is measured in cigarettes per day, and $Y_i$ is measured in grams. The bandwidths for $\hat{\E}[Y_i|X_i=0^+]$ and $\hat{f}_X(0^+)$ are both set to 4. The estimate for $s'(0^+)$ used for each of the displayed bandwidths is -8.10. Standard errors based on 2,500 bootstrap iterations. Data taken from \cite{almond2005costs}.
\end{singlespace}
\end{table}
\vspace{.1cm}

As shown in Section \ref{sec:global}, if we can make a local extrapolation, then we can further recover $\text{ATT}(x)$ for positive $x$. We use the first-degree ATT approximation formula in Section \ref{sec:estimation}, and present the estimates in Figure \ref{fig:ATT_d} for $x \leq 5.$  

\begin{figure}[h!]
\protect\caption{ATT Estimates on Birth Weight for Different Levels of Maternal Smoking\label{fig:ATT_d}}
\vspace{-.2in}
\begin{center}
\includegraphics[scale=0.75]{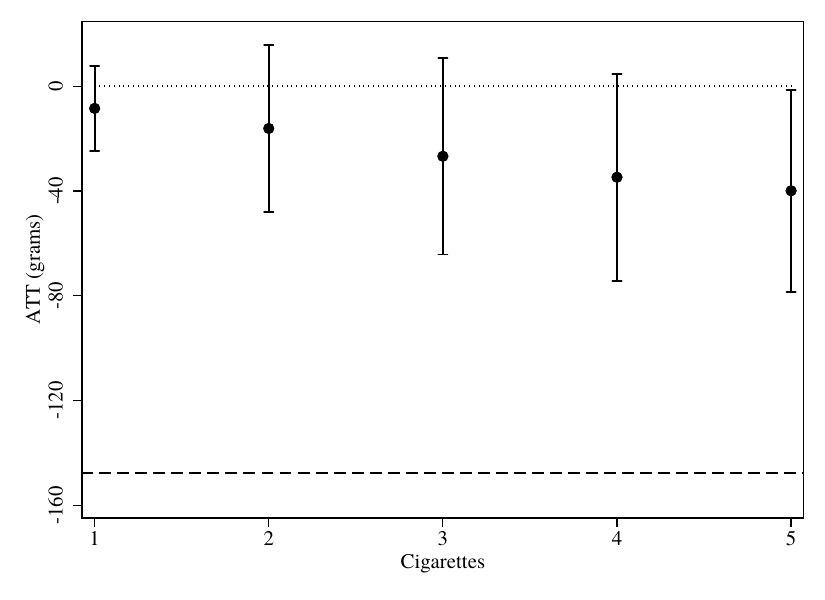}
\end{center}
\vspace{-0.2in}
\singlespace \footnotesize{Note: This figure reports $\widehat{\text{ATT}}(x)=\hat{E}[Y_i|X_i=x]-\hat{E}[Y_i|X_i=\bar{x}^+]-\hat{\theta} \cdot \hat{f}_{s(X^*)|X=\bar{x}}(0)^{-1}\cdot ({\hat{f}_X(\bar{x}^+)}/{\hat{F}_X(\bar{x})})\cdot (x-\bar{x})$  for different values of $x$, represented in the horizontal axis. It uses the same bandwidths in the estimation of $\hat{\theta} \cdot \hat{f}_{s(X^*)|X=\bar{x}}(0)^{-1}\cdot {\hat{f}_X(\bar{x}^+)}/{\hat{F}_X(\bar{x})}$ as in the $h=6$ estimate from Table \ref{tab:beta0}. The $\hat{\E}[Y_i|X_i=x]$ are estimated using using a local linear polynomial with triangular kernel and bandwidth of 3. 95\% confidence intervals based on 2,500 bootstrap iterations shown. The dashed line at -147.6 g is the estimated discontinuity in birth weight at $X=0$ (i.e., $\hat{\E}[Y_i|X_i=0^+]-\hat{\E}[Y_i|X_i=0]$). Data taken from \cite{almond2005costs}. }
\end{figure}

Our estimates show small negative $\text{ATT}$s. For instance, if mothers who currently smoke one cigarette were to quit smoking, their babies would gain about 8 grams at birth. The effects of quitting smoking add up to a gain of only about 40 grams (1.4 ounces) for mothers who smoked 5 cigarettes per day. The estimates for $x\leq 3$ are not significant at standard levels, while the $x=4$ and $x=5$ estimates are marginally significant at 10\% and 5\%, respectively. Figure \ref{fig:ATT_d} thus suggests that the effect of maternal smoking on birth weight is small, and we can rule out effects larger than 85 grams (3 ounces) at standard levels of significance. Indeed, these estimates are much smaller than the effect implied by the discontinuity of the outcome at $X=0$ shown in the dashed line of Figure \ref{fig:ATT_d}. Our estimates support the qualitative point in \cite{almond2005costs} that smoking seems to have only small effects on birth weight, although our findings suggest even smaller effects than in that paper. For reference, the average birth weight for a full-term birth in our sample is 3,458 grams (7 pounds and 10 ounces).

\section{Concluding remarks}\label{sec:conclusion}
When a treatment variable has bunching, this paper presents a new design for identification of the average marginal treatment effects at the bunching point. This is the first identification approach leveraging bunching which does not make assumptions on functional forms or on the shape of the distribution of the unobservables. Since the method does not rely on exclusion restrictions or special data structures, it provides a new avenue for the identification of treatment effects when well-established methods are not applicable.

The approach requires that the treatment be continuously distributed near the bunching point, and it relies on the continuity of the selection function at the bunching point selection value. Intuitively, those who chose the bunching point as an interior solution (i.e., not as a corner solution) are comparable to those right above the bunching point. Besides this and other regularity conditions, the method also requires two other conditions that are hard to explain succinctly, but are implied if the selection equation is monotonic in the selection variable and if the idiosyncratic errors, which are by definition mean independent from the selection variable, are independent from the selection variable at the bunching point.

Identification is achieved by the comparison of the density of the treatment near the bunching point (observed on the positive side) and the density of the selection function at the bunching point (identifiable on the negative side, thanks to a deconvolution of the distribution of the outcome at the bunching point to eliminate the noise from the idiosyncratic error). The ratio of these is exactly the magnitude of the selection bias. 

The approach results in the identification of the average marginal effect as a closed-form expression of identifiable quantities which are fairly standard well-known quantities in the econometrics literature, including the limits as the treatment approaches the bunching point of  (1) the density of the treatment, (2) the expected outcome, and  (3) the derivative of the expected outcome. The final term is the density of the selection variable at the bunching point, which is obtained through a deconvolution of the outcome near the bunching point from the outcome at the bunching point. All the terms in the identification equations can be estimated with off-the-shelf methods readily available in package form in all standard statistical software.

We apply the method to the estimation of the effect of smoking during pregnancy on the baby's birth weight. Our results show that the effects are rather small, strengthening the qualitative results in the previous economics literature.

There is ample opportunity for further technical advancements that would enhance the applicability of this method. Of note, there seems to be a scarce supply of options for the estimation of boundary derivatives in the literature. Even for local polynomial estimators \citep{fangijbels1996}, the optimal degree, kernel, and bandwidths for estimation of boundary derivatives remain unknown. Deconvolution estimators are ubiquitous in other fields, but still rare in economics, with the notable exception of the measurement error literature (see e.g. \citealt{schennachmeasurement}). The application of deconvolution to bunching is more aligned with the classical setting, where the distributions of the ``recorded signal'' and the ``distortion'' are identified, but the behavior of deconvolution estimators with boundary plugins is largely unexplored.

%Biliography%%%%%%%%%%%%%%%%%%%%%%%%%%%%%%%%%%%%%%%%%%%%%%%%%%%%%%%%%%%%%%%%%%%%
%\addtolength{\voffset}{.5cm}
\begin{singlespace}
\bibliographystyle{apalike}
\bibliography{nielsen_caetano_investments_outcomes.bib}
\end{singlespace}

%%%%%%%%%%%%%%%%%%%%%%%%%%%%%%%%%%%%%%%%%%%%%%%%%%%%%%%%%%%%%%%%%%%%%%%%%%%%%%%%
%Appendices

\appendix
\begin{center}
\section*{{\huge Appendix}}
\end{center}

\section{Proofs} \label{ap:proofs}

\subsection{Proof of Proposition \ref{prop:AMEDelta}}
Note that $\mathbbm{E}[Y_i(\bar{x})|X_i=x] = \E[Y_i|X_i=x]-ATT(x)$. The limit $\E[Y_i|X_i=\bar{x}^+]$ exists by Assumption \ref{as:continuoussupport} (ii) and $ATT(\bar{x}^+)$ exists by Assumption \ref{as:continuoussupport} (iii). Therefore,  $\E[Y_i(\bar{x})|X_i=\bar{x}^+]$ exists as well. Then, by $ATT(\bar{x}^+)=0,$ it must be equal to $\E[Y_i|X_i=\bar{x}^+]$. Therefore, 
\begin{align*}
    \text{ATT}(x)&=\E[Y_i|X_i=x]-\E[Y_i(\bar{x})|X_i=x]
    \\
    &=(\E[Y_i|X_i=x]-\E[Y_i|X_i=\bar{x}^+])-(\E[Y_i(\bar{x})|X_i=x])+\E[Y_i(\bar{x})|X_i=\bar{x}^+]):=m(x)-s(x),
\end{align*} 
which completes the result for the $\text{ATT}(x)$.

  Using again that $\mathbbm{E}[Y_i(\bar{x})|X_i=\bar{x}^+] = \E[Y_i|X_i=\bar{x}^+],$
    \begin{align*}
        \text{AME}_{\bar{x}}^+&:=\lim_{x\downarrow \bar{x}}\frac{\E[Y_i|X_i=x]-\E[Y_i(\bar{x})|X_i=x]}{x-\bar{x}}\\
        &=\lim_{x\downarrow \bar{x}} \left\{\frac{\E[Y_i|X_i=x]-\E[Y_i|X_i=\bar{x}^+]}{x-\bar{x}}-\frac{\E[Y_i(\bar{x})|X_i=x]+\E[Y_i(\bar{x})|X_i=\bar{x}^+]}{x-\bar{x}}\right\}\\
        &=\lim_{x\downarrow \bar{x}}\frac{\E[Y_i|X_i=x]-\E[Y_i|X_i=\bar{x}^+]}{x-\bar{x}}-\lim_{x\downarrow \bar{x}}\frac{\E[Y_i(\bar{x})|X_i=x]+\E[Y_i(\bar{x})|X_i=\bar{x}^+]}{x-\bar{x}}\\
        &=\lim_{x\downarrow \bar{x}}\frac{m(x)-m(\bar{x}^+)}{x-\bar{x}}-\lim_{x\downarrow \bar{x}}\frac{s(x)-s(\bar{x}^+)}{x-\bar{x}}:=m'(\bar{x}^+)-s'(\bar{x}^+),
    \end{align*}
    where the third line follows from the two limits existing separately, and the final line by $m(\bar{x}^+)=s(\bar{x}^+)=0$ (which follows immediately from their definitions), and then the definitions of $m'(\bar{x}^+)$ and $s'(\bar{x}^+)$ (see beginning of Section \ref{sec:model} for the introduction of this notation). 
    
    That the two limits above exist follows from parts (ii) and (iii) from Assumption \ref{as:continuoussupport} and the mean value theorem. Specifically, because $\E[Y_i|X_i=\bar{x}^+]$ exists and is finite and $\E[Y_i|X_i=x]$ is differentiable by part (ii) of Assumption \ref{as:continuoussupport}, the mean-value theorem applies, and therefore
    \begin{align} \label{eq:allisoneohm}
        \frac{\E[Y_i|X_i=x]-\E[Y_i|X_i=\bar{x}^+]}{x-\bar{x}}=\frac{d}{dx}\E[Y_i|X_i=x]\Big|_{x=\zeta(x)},
    \end{align}
    for some $\zeta(x)\in (\bar{x},x)$. The limit of the right-hand side exists by part (ii) of Assumption \ref{as:continuoussupport}, thus the limit of the left-hand side also exists. The argument for $s$ is analagous using part (iii) of Assumption \ref{as:smoothcounterfactual}. Note that \eqref{eq:allisoneohm} implies that $m'(\bar{x}) = \lim_{x \downarrow \bar{x}} m'(x)$ and $s'(\bar{x}) = \lim_{x \downarrow \bar{x}} s'(x)$, which will be useful in later proofs.

\subsection{Proof of Theorem \ref{thm:changeofvariables}}
    We begin by establishing local (strict) monotonicity of $s(x)$ on an interval $I = (\bar{x}, \bar{x}+\delta)$ where $\delta > 0$. By part (iii) of Assumption \ref{as:smoothcounterfactual} $\lim_{x \downarrow \bar{x}} s'(x) \ne 0$, or equivalently that for any $\epsilon > 0$, there exists a $\delta > 0$ such that $x \in (\bar{x},\bar{x}+\delta)$ implies $|s'(x)-k| \le \epsilon$, where we let $k:=s'(\bar{x}^+) \ne 0$. Now, consider $\epsilon = |k|/2,$ and define $I_\delta  = (\bar{x},\bar{x}+\delta)$ for the corresponding value of $\delta$. Then, $x \in I_\delta \implies |s'(x) - k| \le k/2$, which implies that $s'(x)$ has the same (non-zero) sign as $k$ does for all $x \in I_\delta$.
    
    Thus, $s$ admits an inverse function $u^{-1}(t)$ for all $t \in s(I_\delta)$. Given part (iii) of Assumption \ref{as:smoothcounterfactual}, $s'(x)$ exists for all $x \in (\bar{x},\bar{x}+\delta')$, where $\delta':=\min\{\delta, \varepsilon_3\}$. This implies that $s^{-1}$ admits the derivative function $\frac{d}{dt} s^{-1}(t) = 1/s'(s^{-1}(t))$ on $s(I)$, where we define $I:=(\bar{x},\bar{x}+\delta')$.
    
    Suppose first that $s$ is strictly increasing. Then, for any $t$,
    \begin{align*}
       P(s(X_i) \le t|X_i \in I) = P(X_i \le s^{-1}(t)|X_i \in I)
    \end{align*}
    Since $f_{X|I}(x) = \frac{d}{dx} P(X_i \le x|X_i \in I) = f_X(x)/P(X_i \in I)$ exists for all $x \in I$, this implies that $f_{s(X)|I}(t)$ exists for any $t \in s(I)$ and is equal to $\frac{d}{dt}P(X_i \le s^{-1}(t)|X_i \in I) = f_{X|I}(s^{-1}(t))/s'(s^{-1}(t))$, using the chain rule. Note that $P(X_i \in I)>0$ by part (i) of Assumption \ref{as:continuoussupport}. The case in which $s$ is decreasing is the same aside from the introduction of a minus sign. Combining both cases, we have that $|s'(x)| = f_{X|I}(x)/f_{s(X)|I}(s(x))$ for any $x \in I$.

\subsection{Proof of Theorem \ref{thm:extensivemargin}}  
    We have by Theorem \ref{thm:changeofvariables} that $s'(x)=\text{sgn}(s'(x)) \cdot f_{X|I}(x)/f_{s(X)|I}(s(x))$ for any $x \in I$. The limit $s'(\bar{x}^+)$ exists (and differs from zero) by part (iii) of Assumption \ref{as:continuoussupport}, and thus 
    $$s'(\bar{x}^+)= \lim_{x \downarrow \bar{x}} \left\{\text{sgn}(s'(x)) \cdot f_{X|I}(x)/f_{s(X)|I}(s(x))\right\}$$
    where $s'(\bar{x}^+) = \lim_{x \downarrow \bar{x}}s'(x)$ by the discussion following Eq. \eqref{eq:allisoneohm}. 
    
    Since $\text{sgn}(s'(x))$ is constant for $x\in I$ (see proof of Theorem \ref{thm:changeofvariables} above), $\theta= \text{sgn}(s'(\bar{x}^+))$ exists. The limit $f_{X|I}(x) = \frac{f_X(\bar{x}^+)}{P(X_i \in I)}$ exists by part (i) of Assumption \ref{as:continuoussupport}. Thus $f_{s(X)|I}(s(\bar{x}^+))$ must also exist, and
    $$s'(\bar{x}^+)= \theta \cdot f_{X|I}(\bar{x}^+)/f_{s(X)|I}(s(\bar{x}^+))$$
    Then by Proposition \ref{prop:AMEDelta}, we obtain the expression in Theorem \ref{thm:extensivemargin}.
    
\subsection{Proof of Lemma \ref{lem:theta}}
Under Assumption \ref{as:continuoussupport}, $\text{ATT}(\bar{x}^+)=0,$ so $\E[Y_i|X_i=\bar{x}^+]=\E[Y_i(\bar{x})|X_i=\bar{x}^+]$. Denote $y_{\bar{x}}:=\E[Y_i(\bar{x})|X_i^*=\bar{x}]$. Now consider first $\E[Y_i|X_i=\bar{x}] = \E[Y_i(\bar{x})|X_i^* \le \bar{x}] = \mathbbm{E}[\E[Y_i(\bar{x})|X_i^*]|X_i^* \le \bar{x}]$. The inner expectation is $\E[Y_i(\bar{x})|X_i^*] = y_{\bar{x}} + s(X^*_i)$, and thus $$\E[Y_i|X_i=\bar{x}^+]-\E[Y_i|X_i=\bar{x}] = -\mathbbm{E}[s(X_i^*)|X_i^* \le \bar{x}]$$
By part (ii) of Assumption \ref{as:smono}, $\text{sgn}(s(x))=-\theta$. Suppose that $\theta = 1$. Then $s(x) < 0$ for all $x \le \bar{x}$, and $\text{sgn}\left(\E[Y_i|X_i=\bar{x}^+]-\E[Y_i|X_i=\bar{x}]\right)=1$. If on the other hand $\theta = -1$, then $s(x) > 0$ for all $x \le \bar{x}$, and $\text{sgn}\left(\E[Y_i|X_i=\bar{x}^+]-\E[Y_i|X_i=\bar{x}]\right)=-1$.

Finally, suppose that in violation of part (iii) of Assumption \ref{as:smoothcounterfactual}, $s'(\bar{x})=0,$ so that $\theta = 0$. In this case, $s(x)=0$ for all $x < \bar{x}$ by Assumption \ref{as:smono}, and the expression of Lemma \ref{lem:theta} still holds. This is useful in establishing that the expression in Theorem \ref{thm:AMEid} holds even when there is no selection at $\bar{x}$.

\subsection{Proof of Theorem \ref{cor:AMEbunching}}
We first give a shorter proof that relies on a strengthening of Assumption \ref{as:bunchingcontinuousnew1}, as it is more intuitive. Specifically, suppose we add to Assumption \ref{as:bunchingcontinuousnew1} that $s'(x)$ and $f_{X^*}(x)$ are both continuous at $x=\bar{x}$. Then, paralleling the proof of Theorem \ref{thm:changeofvariables}, we could establish local monotonicity of $s(x)$ on an interval $I^* = (\bar{x}-\delta, \bar{x}+\delta)$ that now straddles both sides of the bunching point. To see this, let $k:=s'(\bar{x}) \ne 0$. Then by $\lim_{x \rightarrow \bar{x}} s'(x) = s'(\bar{x}) \ne 0$, we know that for $\epsilon =|k|/2$, there exists a $\delta > 0$ such that $|x-\bar{x}| < \delta$ implies $|s'(x)-k| \le |k|/2$, which implies that $s'(x)$ has the same (non-zero) sign as $k$ does for all $x \in I^*$. Then we would have, exactly as in the case of Theorem \ref{thm:changeofvariables}, that for any $x \in I^*$,
\begin{equation}\label{eq:changeofvariablesbunchingclosure}
s'(x)= \theta \cdot \frac{f_{X^*|I^*}(x)}{f_{s(X^*)|I^*}(s(x))},
\end{equation}
and in particular since $\bar{x} \in I^*$, $s'(\bar{x})= \theta \cdot \frac{f_{X^*|I^*}(\bar{x})}{f_{s(X^*)|I^*}(0)}$ using that $s(\bar{x})=0$ by continuity of $s$ at $\bar{x}$. We can now transform this expression so that it no longer depends on $f_{s(X^*)|I^*}(0),$ but it depends instead on $f_{s(X^*)|X^*\le \bar{x}}(0)$. Note that the point $\bar{x}$ belongs to the intersection of $I^*$ and the set $(-\infty,\bar{x}]$, so
$$f_{s(X^*)|I^*}(0) = \frac{P(X_i^* \le \bar{x})}{P(X_i^* \in I^*)} \cdot f_{s(X^*)|X^* \le \bar{x}}(0). $$
Meanwhile, $f_{X^*|I^*}(\bar{x}^+) = f_{X^*}(\bar{x})/P(X_i^* \in I^*)=f_{X^*}(\bar{x}^+)/P(X_i^* \in I^*)$. Thus, all together we have:
\begin{equation}\label{eq:sprimexnoi}
s'(\bar{x})=\theta\cdot \frac{f_{X^*}(\bar{x}^+)/F_X(\bar{x})}{f_{s(X^*)|X=\bar{x}}(0)}, 
\end{equation}
noting that $f_{s(X^*)|X^*\leq \bar{x}}=f_{s(X^*)|X=\bar{x}}$ and $P(X^* \le \bar{x}) = P(X=\bar{x}) = F_{X}(\bar{x}),$ since $X=\bar{x}$ if and only if $X^* \le \bar{x}$. Eq. \eqref{eq:sprimexnoi} is useful because the quantity $P(X_i^* \in I^*)$ for the unknown interval $I^*$ has canceled out. The final result then follows by Proposition \ref{prop:AMEDelta}, and using that $X_i=\bar{x} \iff X_i^* \le \bar{x}$.\\

Now, we show that the result holds under the weaker version of Assumption \ref{as:bunchingcontinuousnew1} made before the statement of Theorem \ref{cor:AMEbunching}. To begin, note that Assumptions \ref{as:continuoussupport}-\ref{as:bunchingcontinuousnew1} do not require $s(x)$ to be differentiable exactly at $\bar{x}$, though they do guarantee that $s'(\bar{x}^+)$ and $f_{X^*}(\bar{x}^+)$ exist. Without $s'(\bar{x})$, we cannot necessarily construct the interval $I^*$ used above that extends on either side of $\bar{x}$ whose width is based upon the value of $s'(\bar{x})$.

Instead, we can use Assumption \ref{as:smono} and the existence of $f_{s(X^*)}(0)$ by Assumption \ref{as:bunchingcontinuousnew1} to allow our logic to ``cross'' the bunching point: existence of $f_{s(X^*)}(0)$ implies that $\lim_{h \downarrow 0} h^{-1}\{P(s(X_i^* \le t+h)-P(s(X_i^* \le t)\}$ and $\lim_{h \uparrow 0} h^{-1}\{P(s(X_i^* \le t+h)-P(s(X_i^* \le t)\}$ both exist and are equal to one another. Let $I^-:=(-\infty, \bar{x}] \cup I = (-\infty,\bar{x}+\delta')$ denote the interval in part (i) of Assumption \ref{as:bunchingcontinuousnew1} upon which $f_{s(X^*)}(s(\cdot))$ is assumed to exist, which includes $\bar{x}$. Note that by part (ii) of Assumption \ref{as:bunchingcontinuousnew1}, $s(\bar{x}) = s(\bar{x}^+)=0$. Note that $0$ belongs to the interior of $s(I^-)$, and the existence of $f_{s(X^*)}(0)$ implies the existence of $f_{s(X^*)|I^-}(0)=f_{s(X^*)}(0)/P(X_i^* \in I^-)$.

Suppose first that $s'(\bar{x}^+) > 0$ so that $s$ is strictly increasing on $I$. In this case, Assumption \ref{as:smono} implies that $s(x) < 0$ for all $x < \bar{x}$. Let $s^{-1}(\cdot)$ continue to denote the inverse of $s$ on the interval $I$, which is well-defined as a function from $s(I)$ to $I$. Therefore, conditional on the event $X_i^* \in I^-$, $s(X^*_i) \le t$ occurs if and only if $X_i^* \le s^{-1}(t)$ for any $t \in s(I)$, even though $s(\cdot)$ need not be invertible on all of $I^{-}$. Thus:
\begin{align*}
    f_{s(X^*)|I^-}(0)=\lim_{t \downarrow \bar{x}} \frac{d}{dv} P(s(X^*_i) \le t|X_i^* \in I^-) = \lim_{t \downarrow \bar{x}} \frac{d}{dv} P(X^*_i \le s^{-1}(t)|X_i^* \in I^-) = \lim_{t \downarrow \bar{x}} \frac{f_{X^*|I^-}(s^{-1}(t))}{s'(s^{-1}(t))}.
\end{align*}
 By differentiability of $s$ on $I$ (part (iii) of Assumption \ref{as:continuoussupport}), $s^{-1}$ is continuous on $I$ and thus we have that $\lim_{t \downarrow \bar{x}} f_{X^*|I^-}(s^{-1}(t)) =\lim_{t \downarrow \bar{x}} f_{X^*}(s^{-1}(t))/P(X_i^* \in I^-) = f_{X^*}(\bar{x}^+)/P(X_i^* \in I^-)$, where the limit $f_{X^*}(\bar{x}^+)$ exists by part (i) of Assumption \ref{as:continuoussupport}. Similarly, $\lim_{t \downarrow \bar{x}} s'(s^{-1}(t)) = \lim_{x \downarrow \bar{x}} s'(x)$, where the limit exists and differs from zero by part (iii) of Assumption \ref{as:continuoussupport}. Thus, using right continuity of $f_{s(X^*)}(t)$ at $t=0$ by part (i) of Assumption \ref{as:bunchingcontinuousnew1}:
\begin{equation} \label{eq:fsx0}
f_{s(X^*)|I^-}(0) = \frac{\lim_{t \downarrow \bar{x}} f_{X^*|I^-}(s^{-1}(t))}{\lim_{t \downarrow \bar{x}} s'(s^{-1}(t))} = \frac{f_{X^*}(\bar{x}^+)/P(X_i^* \in I^-)}{\lim_{x \downarrow \bar{x}} s'(x)}= \frac{f_{X^*}(\bar{x}^+)/P(X_i^* \in I^-)}{s'(\bar{x}^+)},
\end{equation}
using in the last equality that $s'(\bar{x}^+) = \lim_{x \downarrow \bar{x}}s'(x)$ by the discussion following Eq. \eqref{eq:allisoneohm}.

Now, again given the existence of $f_{s(X^*)}(0)$ and hence of $f_{s(X^*)|I^{-}}(0)$, we can also use the derivative from the left to write:
\begin{align*}
    f_{s(X^*)|I^-}(0)&=\lim_{h \uparrow 0} \frac{P(s(X_i^*) \le h|X_i^* \in I^-)-P(s(X_i^*) \le 0|X_i^* \in I^-)}{h}\\
    &= \frac{P(X_i^* \le \bar{x})}{P(X_i^* \in I^-)} \cdot \lim_{h \uparrow 0} \frac{P(s(X_i^*) \le h|X_i^* \le \bar{x})-P(s(X_i^*) \le 0|X_i^* \le \bar{x})}{h}\\
    &= \frac{P(X_i^* \le \bar{x})}{P(X_i^* \in I^-)} \cdot f_{s(X^*)|X^* \le 0}(0).
\end{align*}
Combining with Eq. \eqref{eq:fsx0}, we have that $s'(\bar{x}^+)=\frac{f_{X^*}(\bar{x}^+)}{f_{s(X^*)|X_i^* \le \bar{x}}(0)},$ using that $P(X_i^* \in I^-) \ge P(X_i^* \le \bar{x}) \ne 0$. 

The case in which $s$ is decreasing on $I$ proceeds analogously, aside from the introduction of a minus sign. We then have that in either case, $s'(\bar{x}^+)=\theta \cdot \frac{f_{X^*}(\bar{x}^+)}{f_{s(X^*)|X_i^* \le \bar{x}}(0)}.$ The final result then follows by Proposition \ref{prop:AMEDelta}, and using that $X_i=\bar{x} \iff X_i^* \le \bar{x}$.

\subsection{Proof of Theorem \ref{thm:AMEid}}
Combining Eq. \eqref{eq:AMEbunching} with Lemma \ref{lem:u(D*)}, we have
\begin{equation*}
    \text{AME}_{\bar{x}}^+=m'(\bar{x}^+)-2\pi\theta \left(\int \frac{\E[e^{\textbf{i}\xi Y_i}|X_i=\bar{x}]}{\E[e^{\textbf{i}\xi Y_i}|X_i=\bar{x}^+]}d\xi\right)^{-1}\cdot \frac{f_X(\bar{x}^+)}{F_X(\bar{x})}.
\end{equation*}
To see that $m'(\bar{x}^+)=\lim_{x\downarrow \bar{x}}\frac{d}{dx}\E[Y_i|X_i=x]$, take the limit of Eq. \eqref{eq:allisoneohm} in the proof of Proposition \ref{prop:AMEDelta} as $x \downarrow \bar{x}$. The expression for $\theta$ follows from Lemma \ref{lem:theta}.

\subsection{Proof of Lemma \texorpdfstring{\ref{lem:u(D*)}}{2}}
Conditional on $X_i=\bar{x}$, Equation \eqref{eq:equationzero} holds, so:
\begin{align*}
    F_{Y-y_{\bar{x}}|X=\bar{x}}(y)&=\mathbb{P}(s(X_i^*)+\epsilon_i\leq y|X_i=\bar{x})\\
    &=\int F_{s(X^*)|X=\bar{x},\epsilon=e}(y-e)dF_{\epsilon|X=\bar{x}}(e)=\int F_{s(X^*)|X=\bar{x},\epsilon=e}(y-e)dF_{\epsilon|X^*=\bar{x}}(e)
    \\
    &=\int F_{s(X^*)|X=\bar{x}}(y-e)dF_{\epsilon|X=\bar{x}^+}(e),
\end{align*}
where the second equality follows from part (iii) of  Assumption \ref{as:deconvregularity}. The third equality uses part (ii) and the Helly-Bray Theorem, due to the fact that $F_{s(X^*)|X=\bar{x}}(y-e-y_{\bar{x}})$ is bounded and continuous.

Then, by part (i) of Assumption \ref{as:deconvregularity}, $F_{Y-y_{\bar{x}}|X=\bar{x}}(y)$ is differentiable, and by part (i) of Assumption \ref{as:bunchingcontinuousnew1}, $ F_{s(X^*)|X=\bar{x}}$ is differentiable. By the Dominated Convergence Theorem,  we can write the convolution inverse problem
\begin{equation}\label{eq:convolution}
f_{Y-y_{\bar{x}}|X=\bar{x}}(y)=\int f_{s(X^*)|X=\bar{x}}(y-e)dF_{\epsilon|X=\bar{x}^+}(e).
\end{equation}
Equation \eqref{eq:convolution} has a well-known closed form solution using the Fourier representation (see e.g. \citealt{schennachmeasurement}): 
\begin{align} 
    f_{s(X^*)|X=\bar{x}}(v)&=\frac{1}{2\pi}\int \frac{\E[e^{\textbf{i}\xi (Y_i-y_{\bar{x}})}|X_i=\bar{x}]}{\E[e^{\textbf{i}\xi \epsilon_i}|X_i=\bar{x}^+]}e^{-\textbf{i}\xi v} d\xi=\frac{1}{2\pi}\int \frac{\E[e^{\textbf{i}\xi Y_i}|X_i=\bar{x}]}{\E[e^{\textbf{i}\xi (\epsilon_i+y_{\bar{x}})}|X_i=\bar{x}^+]}e^{-\textbf{i}\xi v} d\xi \nonumber \\
    &=\frac{1}{2\pi}\int \frac{\E[e^{\textbf{i}\xi Y_i}|X_i=\bar{x}]}{\E[e^{\textbf{i}\xi Y_i}|X_i=\bar{x}^+]}e^{-\textbf{i}\xi v} d\xi, \label{eq:deconvgeneralpointvappendix}
\end{align}
where the second equality follows because, for $X_i>\bar{x},$ $\epsilon_i=Y_i-\E[Y_i|X_i],$ and thus $\E[e^{\textbf{i}\xi (\epsilon_i+y_{\bar{x}})}|X_i=x]=\E[e^{\textbf{i}\xi (Y_i-\E[Y_i|X_i]+y_{\bar{x}})}|X_i=x]=e^{-\textbf{i}\xi(\E[Y_i|X_i=x]-y_{\bar{x}})}\E[e^{\textbf{i}\xi Y_i}|X_i=x],$ and then $\E[Y_i|X_i=x]$ converges to $y_{\bar{x}}$ as $x\downarrow 0.$

The result then follows if we evaluate the expression in Equation \eqref{eq:deconvgeneralpointvappendix} above at $v=0$. The integral in Eq. \eqref{eq:deconvgeneralpointvappendix} must converge for $v=0$ by existence of $f_{s(X^*)|X=\bar{x}}(0)$, which is guaranteed by part (i) of Assumption \ref{as:bunchingcontinuousnew1}.

\subsection{Proof of Theorem \ref{thm:localexpansion}}
Assumption \ref{as:analytic} implies that $s$ is also analytic on $I'$. Let $\bar{I}=[\bar{x},\bar{x}+\epsilon_5]$, the closure of $I'$ in $\mathbbm{R}$. Taking a direct analytic continuation from $I'$ to $\bar{I}$, we can take $s$ to be analytic on all of $\bar{I}$ with $s(\bar{x})=0$. Then, there exists an $\varepsilon>0$ such that $s(x)= \sum_{k=1}^\infty s^{(k)}(\bar{x})\cdot \frac{(x-\bar{x})^k}{k!}$ for any $x \in [\bar{x},\bar{x}+\varepsilon]$. Analyticity of $s$ on $\bar{I}$ implies that $s^{(k)}(x)$ is continuous on $\bar{I}$ for each $k$ and thus $s^{(k)}(\bar{x})=s^{(k)}(\bar{x}^+)$. The expression for the remainder follows from the Taylor theorem.
  
\subsection{Proof of Corollary \ref{cor:extrapolation}}
Since $s$ is analytic on the set $\bar{I}$ defined in the Proof of Theorem \ref{thm:localexpansion}, it is continuously differentiable and we can use Equation \eqref{eq:changeofvariablesbunchingclosure} from the proof of Theorem \ref{cor:AMEbunching}
\begin{equation*}
s'(x)= \theta \cdot \frac{f_{X^*|I^*}(x)}{f_{s(X^*)|I^*}(s(x))},
\end{equation*}
for any $x \in I^*$, where $I^*$ is a set that includes a neighborhood of $\bar{x}$. Differentiating with respect to $x$, we obtain: 
$$s^{(k)}(x)=\theta \cdot \frac{d^{k-1}}{dx^{k-1}}\frac{f_{X^*|I^*}(x)}{f_{s(X^*)|I^*}(s(x))}.$$
Following the same steps as the proof of Theorem \ref{cor:AMEbunching}, we can replace the conditioning on $X_i^* \in I^*$ with conditioning on $X_i^* \le \bar{x}$ and evaluate at $\bar{x}$ to obtain
$$s^{(k)}(\bar{x})=\frac{\theta}{F_X(\bar{x})}\cdot \frac{d^{k-1}}{dx^{k-1}}\frac{f_{X^*}(\bar{x})}{f_{s(X^*)|X=\bar{x}}(s(\bar{x}))}.$$
Working out the derivatives and and evaluating at $\bar{x}$, we have
\begin{align*}
s^{(k)}(x)&=\frac{\theta}{F_X(\bar{x})}\cdot \sum_{\ell = 0}^{k-1} \begin{pmatrix} k-1 \\ \ell \end{pmatrix} f^{(k-1-\ell)}_{X^*}(\bar{x}) \cdot \left.\frac{d^\ell}{dx^{\ell}} \{f_{s(X^*)|X=\bar{x}}(s(x))\}^{-1}\right|_{x=\bar{x}}.
\end{align*}
The derivatives of $1/f_{s(X^*)|X=\bar{x}}(s(x))$ evaluated at $x=\bar{x}$ can be worked out recursively knowing $f_{s(X^*)|X=\bar{x}}(s(\bar{x}))$ and $f^{(\ell)}_{s(X^*)|X=\bar{x}}(s(\bar{x}))$ for each $\ell$. These derivatives are identified since by Eq. \eqref{eq:deconvgeneralpointvappendix}, $f_{s(X^*)|X=\bar{x}}(v)$ is identified for every $v \in s((-\infty, \bar{x}))$, e.g.:
\begin{equation*}
    \frac{d}{dv}f_{s(X^*)|X=\bar{x}}(v)=\frac{1}{2\pi} \cdot \frac{d}{dv}\int \frac{\E[e^{\textbf{i}\xi Y_i}|X_i=\bar{x}]}{\E[e^{\textbf{i}\xi Y_i}|X_i=\bar{x}^+]}e^{-\textbf{i}\xi v} d\xi,
\end{equation*}
which can be differentiated again and again as needed. The power series converges for each $x \in I_{\varepsilon}$ by analyticity of $s(x)$. Finally, $f_{X^*}^{(k)}(\bar{x})$ is identified for each $k$ as $f_{X^*}^{(k)}(\bar{x}) = f_X^{(k)}(\bar{x}^+)$.

\subsection{Proof of Corollary \ref{cor:extrapolationfar}}
Given Equation \eqref{eq:taylorexpression},
$$\text{ATT}(x)=m(x)-\sum_{k=1}^\infty s^{(k)}(\bar{x}^+)\cdot  \frac{(x-\bar{x})^k}{k!}.$$
The Cauchy–Hadamard formula gives that the radius of convergence of the power series appearing on the right-hand side is $R = \left( \limsup_{k \to \infty} \left| \frac{s^{(k)}(c)}{k!} \right|^{1/k} \right)^{-1} $.

\section{Constructing \texorpdfstring{$X_i^*$}{Xi*}} \label{sec:xstarappendix}

In this appendix, we detail the models and results mentioned in Section \ref{sec:xstar}. When not provided in-line, proofs are found in the Supplementary Appendix \ref{sec:xstarproofs}.

\subsection{Models with scalar heterogeneity}\label{sec:economic}
In this section, we consider models in which individuals' choice problems are parameterized by a scalar parameter $\rho_i$, and we obtain a choice rule $X_i = \max \{h(\rho_i),\bar{x}\}$ for a strictly increasing and differentiable function $h$. This is equivalent to Eq. \eqref{eq:xbarmax} with $X_i^*=h(\rho_i)$. Since such settings arise naturally due to physical positivity constraints (the quantity of a good consumed, the time spent on an activity, etc.), we take $\bar{x}=0$ for simplicity. This is without loss of generality because we can always redefine $X_i$ as $X_i - \bar{x}$.

In this class of examples, we suppose individuals have two scalar choice variables $x$ and $r$, with individual $i$'s utility denoted as $V(x,r;\theta_i)$ for a family of utility functions parametrized by $\theta$. We suppose that individuals' choices are made subject to the budget constraint that $p\cdot x+r=W$, where $p$ indicates the relative price of $x$ versus $r$, and $W$ a budget common to all individuals. This, for example, nests the class of Example \ref{ex:cigarettes} in the time use case (e.g., hours spent on TV watching), where $p=1$ and $W=24$: all individuals have 24 hours in a day and can swap time one-for-one between watching TV, $x,$ and spending time on other activities, $r$. Note that the budget constraint can be relaxed to $p\cdot x+r \le W$, provided that utility is strictly increasing in at least one of the two goods for all individuals.
\begin{proposition} \label{prop:transformationexistence}
    Suppose $V:\mathbb{R}^+\times \mathbb{R}^+\rightarrow \mathbb{R}$ is a twice differentiable function of $(x,r),$ with heterogeneity parameterized by a scalar $\theta_i$ such that, with probability one,
    $$X_i = X(\theta_i):= \text{argmax}_x\{V(x,r;\theta_i) \text{ subject to }  p\cdot x+r=W\}.$$
    Define $MRS(x,\theta) := V_{x}(x,W-px;\theta)/V_{r}(x,W-px;\theta),$ and suppose that $MRS(x,\theta)$ is strictly decreasing in $x$ and continuously differentiable on $[0,W] \times \Theta$. Finally, suppose that $P(X_i=W)=0$ (nobody spends their entire budget on $X_i$). Then, with probability one, there exists a strictly increasing and differentiable function $h$ such that $$X_i = \max\{h(\rho_i),0\},$$ 
    where $\rho_i = MRS(0,\theta_i)$.
\end{proposition}

Note that $V(x,r;\theta)$ need be defined only for positive values of $x$ (and not e.g. on all of $\mathbbm{R}^2$). Thus, we do not need to conceive the utility that an individual would receive from a negative amount of cigarettes, or from watching TV for a negative amount of time. The variable $\theta_i$ regulates the relative preference for $x$ over $r,$ in the sense that $V_x(x,r;\theta)/V_r(x,r;\theta)$ is strictly increasing in $\theta$ for all $x,r$. 

The assumptions invoked in Proposition \ref{prop:transformationexistence} are natural if preferences are convex over $(x,r)$, can be parameterized by a scalar $\theta_i$, and utility is smooth in that scalar. While an assumption that the relative price $p$ of $x$ and $r$ is homogeneous across individuals is relatively weak (within e.g. a given market), the restriction that budgets $W$ are homogeneous across individuals $i$ is much stronger. This can be relaxed by assuming that the demand for $x$ does not depend on $W$, which is reasonable if preferences are approximately quasi-linear, $r$ is a good, or $W$ is suitably high. We can proceed without these assumptions by controlling for income in estimation, such that the assumptions of Proposition \ref{prop:transformationexistence} and Theorem \ref{thm:AMEid} need only hold conditional on income. This strategy can be helpful in motivating Proposition \ref{prop:transformationexistence} more generally by also controlling for other proxies of preference heterogeneity, making the scalar heterogeneity assumption less restrictive. Supplementary Appendix \ref{sec:controls} discusses the use of control variables in our approach.

\subsection{Monotonic transformations of \texorpdfstring{$X_i^*$}{Xi*}}
\label{sec:generalization AME}
In this section, we show that we can generalize Theorem \ref{thm:AMEid} to the case where $X^*_i$ is not the primitive source of selection, but instead there exists some (possibly unobserved) index of heterogeneity $\rho_i$ such that $X^*_i = h(\rho_i)$. In this case, we can make Assumptions \ref{as:smono}-\ref{as:deconvregularity} with $\rho_i$ replacing $X_i^*$, while maintaining the constructive estimand of Theorem \ref{thm:AMEid} for $\text{AME}_{\bar{x}}^+$, even if $\rho_i$ is unobserved and the function $h$ is unknown. The index $\rho_i$ will therefore not be unique. For example, in the case of Proposition \ref{prop:transformationexistence} we could make Assumptions \ref{as:smono}-\ref{as:deconvregularity} about $\rho_i:=\theta_i$ or about $\rho_i:=MRS(0,\theta_i)$, rather than on $X_i^*$.

\begin{proposition} \label{prop:transformation}
    Suppose that Assumption \ref{as:smoothcounterfactual} holds, and Assumptions \ref{as:smono}-\ref{as:deconvregularity} all hold for a variable $\rho_i$ such that $X^*_i = h(\rho_i),$ where $h$ is a strictly increasing and differentiable function. Let $\bar{\rho} = h^{-1}(\bar{x}),$ and suppose that $h'(\bar{\rho}) \ne 0.$ Then, Equation \eqref{eq:u'} holds.
\end{proposition}
Equation \eqref{eq:u'} is the expression for $\text{AME}_{\bar{x}}^+$ from Theorem \ref{thm:AMEid}. This expression still holds unchanged. Note that $h$ does not need to be known or identified, and hence $\rho_i$ may be unobserved for all $i$. 
\subsection{Models with unrestricted heterogeneity}\label{sec:unrestrictedheterigeneity}
In this section, we show that $X_i^*$ can be explicitly constructed in choice models with unrestricted heterogeneity -- i.e., if $X_i$ is chosen using the (unrestrictedly heterogeneous) utility $ u_i(x)$. To simplify the exposition, suppose that $\bar{x}=0$, which is without loss of generality, since one may always redefine $X_i$ to be $X_i - \bar{x}$.

In contrast to the model of Section \ref{sec:economic}, this model considers choice over a scalar $x$ only. This may reflect a profiled utility function incorporating any other choice variables that have been maximized over for a given value of $x$. This allows us to remain agnostic about agents' margins of choice and the constraints (e.g. budget) that they may face in addition to $X_i \ge 0$.\footnote{For instance, let $(X_i,R_i)=\text{argmax}_{\substack{(x,r): x \ge 0\\ (x,r) \in \mathcal{S}_i}} V_i(x,r)$ for some set $\mathcal{S}_i$, where $r$ may now be a vector. Then $X_i = \text{argmax}_{x \ge 0} u_i(x),$ where we define $u_i(x):= \max_{r: (x,r) \in \mathcal{S}_i} V_i(x,r)$.}

\begin{proposition} \label{prop:marginalutilitynegative}
    Suppose that, with probability one,  $X_i = \text{argmax}_{x \ge 0} u_i(x),$ and $u_i(\cdot)$ is differentiable at $x=0.$ Then, $\mathbb{P}(u_i'(0) \le 0|X_i=0)=1$.
\end{proposition}
\begin{proof}
The proof is by contradiction: if $u_i'(0)>0$, then $u_i(x)$ cannot be maximized at $x=0.$ This is because $i$ could strictly increase their utility by choosing $x$ to be slightly larger than $0$. It follows that, with probability one, if $X_i=0,$ then $u_i'(0)\leq 0$.
\end{proof}

Define
\begin{equation} \label{eq:defxstar}
    X_i^*:=\begin{cases} X_i & \text{ if } X_i > 0\\ u_i'(0) & \text{ if } X_i = 0. \end{cases}
\end{equation}
Given Proposition \ref{prop:marginalutilitynegative},  $X_i^*$ as defined above satisfies Eq. \eqref{eq:xbarmax} -- i.e., with probability one, $X_i = \max\{0,X_i^*\}$.

\begin{proposition} \label{prop:qc}
    Suppose that, with probability one, $X_i = \text{argmax}_{x \ge 0} u_i(x),$ and $u_i(\cdot)$ is differentiable and quasiconcave, then $\mathbb{P}(u_i'(0) \geq 0|X_i>0)=1.$
\end{proposition}
\begin{proof}
    If $X_i>0$ then $u_i(X_i) > u_i(0)$. If $u_i'(0) < 0,$ there exists an $x \in (0,X_i)$ such that $u_i(x) < u_i'(0)$. Thus the set $\{x': u_i(x') > u_i(x)\}$ includes both $0$ and $X_i$, but does not include $x$. Therefore, it is not convex, which contradicts the quasi-concavity of $u_i$. It follows that,  with probability one, if $X_i>0,$ then $u_i'(0)\geq 0.$
\end{proof}

Given Propositions \ref{prop:marginalutilitynegative} and \ref{prop:qc}, all individuals who strictly prefer $x=0$ to a positive amount are bunched (i.e., $X_i=0$) with a selection variable that is different from the treatment (i.e., $X_i^*<0$). The bunching point may also contain individuals who are exactly indifferent between $x=0$ and a marginally positive amount (i.e. $X_i^*=u_i'(0)=0$). 

We also have that
$$s(x) = \E[Y_i(0)|X_i^*=x]= \begin{cases} \E[Y_i(0)|X_i=x] & \text{ if } x > 0\\ \E[Y_i(0)|u_i'(0)=x] & \text{ if } x = 0 \end{cases} \quad - \quad \E[Y_i(0)|X_i=x^+].$$

The key assumptions on $X_i^*$ and $s(X_i^*)$ in Theorem \ref{thm:AMEid} can now be interpreted in terms of the joint distribution of $Y_i(0)$ and $u_i'(0)$, and how it relates to the joint distribution of $Y_i(0)$ and $X_i$. For instance, part (iii) of Assumption \ref{as:bunchingcontinuousnew1} requires $s(x)$ to be right continuous at $x=0$. This requires that $\lim_{x \downarrow 0} \E[Y_i(0)|X_i=x] = \E[Y_i(0)|u'_i(0)=0].$ In terms of the example of our application, we require that the bunchers that are indifferent between not smoking and smoking a marginal amount (i.e., the ``marginal nonsmokers'') are similar to those that choose to smoke a marginal amount (the ``marginal smokers''), in terms of their mean untreated potential outcome.

Part (ii) of Assumption \ref{as:deconvregularity} says that $Y_i(0)|X_i=x \rightarrow_{d} Y_i(0)|u_i'(0)=0$ as $x \downarrow 0$. Again, this has the interpretation that the marginal nonsmokers are similar to the marginal smokers, now in terms of the entire distribution of their untreated outcomes.

The existence of $f_{s(x^*)}(0)$ in part (ii) of Assumption \ref{as:bunchingcontinuousnew1} requires that
\begin{align*}
    &\lim_{h \uparrow 0} \frac{P(\E[Y_i(0)|u_i'(0)] \le h + y_{0})-P(\E[Y_i(0)|u_i'(0)] \le y_{0})}{h}\\
    & \hspace{2.5in}= \lim_{h \downarrow 0} \frac{P(\E[Y_i(0)|X_i] \le h + y_{0})-P(\E[Y_i(0)|X_i] \le y_{0})}{h}.
\end{align*}
Here the interpretation would be that the ``close-to-marginal nonsmokers'' are similar to the ``close-to-marginal smokers'' in terms of mean untreated outcomes.

Finally, in this context Assumption \ref{as:smono} says that if selection is positive just above $0$, then $\E[Y_i(0)|u'_i(0)= x] < \E[Y_i(0)|u'_i(0)=0]$ for all $x < 0$. In other words, there are no negative values of marginal utility with the property that individuals having that level of $u'_i(0)$ have larger values of $Y_i(0)$ than do the marginal nonsmokers (who have $u'_i(0)=0$). However it does not require that the conditional mean of $Y_i(0)$ be increasing globally in $u'_i(0)$, or even for $Y_i(0)$ and $u'_i(0)$ to be positively correlated overall.

We note that for generality in our setup, we have here defined $X_i^*$ in terms of the marginal utility of $x$. This contrasts with our use of the marginal rate of substitution between $x$ and $r$ in Proposition \ref{prop:transformationexistence}, where the setting was one in which individuals choose $x$ and $r$ subject to a budget constraint. Since marginal utility is not invariant to different cardinalizations of a given individual's preferences, one might worry that applying different arbitrary monotonic transformations to $u_i(\cdot)$ for different individuals might affect the plausibility of Assumptions \ref{as:smono}-\ref{as:deconvregularity} for the construction \eqref{eq:defxstar}. However, note that in the discussion above, one only needs to make assumptions about marginal ($u_i'(0)=0$) and close-to-marginal ($u_i'(0)$ small) bunchers (i.e., nonsmokers). These categories are indeed invariant to smooth monotonic transformations of utility, and thus the quantitative magnitudes of $u_i'(0)$ are not important for defending Assumptions \ref{as:smono}-\ref{as:deconvregularity}.

\subsection{Defining \texorpdfstring{$X_i^*$}{Xi*} by extending the support of treatment} \label{sec:noxstar}

In this section, we show that in a setting with boundary bunching at $\bar{x}=0$, the random variable $X_i^*$ can be defined in the absence of a choice model as an extension from the positive part of the support of $X_i$ into negative values, under suitable conditions. This construction has the important property of satisfying Assumption \ref{as:deconvregularity} automatically.

For any $e \in [0,1]$, $x \ge \bar{x}$ and random variable $A$, define where $Q_A(a) := \inf\{a: F_A(a) \ge e\}$ and $F_A$ is the CDF function of $A$. Given a conditioning variable $X_i$, we can write $A=Q_{A|X}(E)$ with probability one, for a random variable $E \sim \text{Unif}[0,1]$ that satisfies $E \indep X$. See Lemmas 3 and 4 of \citet{goff2024testingidentifyingassumptionsparametric} for a proof of this property, which holds regardless of whether $A$ is discrete or continuously distributed.

Consider in particular the conditional quantile function $g(x,x',e) = Q_{Y(x)|X=x'}(e)$. Then, using that $Y_i=Y_i(X_i)$ we can write, with probability one:
\begin{equation} \label{eq:quantileoutcome}
	Y_i = g(X_i,X_i,E_i),
\end{equation}
where $E_i := F_{Y(X)|X}(Y_i) = F_{Y|X}(Y_i)$. Note that with $Y_i$ continuously distributed conditional on $X_i$, then $E_i|X_i\sim \text{Unif}[0,1],$ and thus $E_i \indep X_i$.\footnote{If $Y_i$ exhibits mass points (e.g. $Y_i$ is discrete), then we can define $E_i|Y_i,X_i \sim \text{Unif}[\lim_{y \uparrow Y_i}F_{Y|X}(y),F_{Y|X}(Y_i)],$ and \eqref{eq:quantileoutcome} still holds with $E_i|X_i \sim \text{Unif}[0,1]$ (Cf. Lemma 4 of \citealt{goff2024testingidentifyingassumptionsparametric}).}  

 For notational simplicity, we consider in this section a setting in which $\bar{x}=0$, and assume only that $X_i$ is defined on the positive side of the real line. Accordingly, consider any $x \ge 0$. Assume that $g$ is differentiable. By the fundamental theorem of calculus, we can write:
\begin{align*}
	g(x,x,e) &= \underbrace{g(0,0,e)+\int_0^{x} g_2(0,v,e)dv}_{:=s(x,e)}+\underbrace{\int_0^x g_1(x,v,e)dv}_{:=m(x,e)},
\end{align*}
where $g_1$ and $g_2$ represent derivative functions of $g$, and the path of integration is from $(x,x')=(0,0)$ to $(0,x)$ and then from $(0,x)$ to $(x,x)$. Note that differentiability of $g$ in $x,x',e$ also implies that $Y_i$ is continuously distributed, conditional on any value $X_i$ (implying part (i) of Assumption \ref{as:deconvregularity}).

We can write
\begin{equation*}
	s(e,x) = g(0,0,e)+\int_0^{x} g_2(0,v,e)dv = Q_{Y(0)|X=x}(e),
\end{equation*}
where the second term above $s(x,e)$ is a pure ``endogeneity term'', capturing how the distribution of the untreated potential outcome $Y_i(0)$ varies across groups with different treatment levels $X_i$. On the other hand $m(x,e):=\int_0^x g_1(x,v,e)dv$ is a pure causal term, summarizing how the distribution of $Y_i(x)$ changes with $x$ with a fixed conditioning group $X_i=x$. Note that $m(0,e)=0$ for all $e$.

By Equation \eqref{eq:quantileoutcome}, we have that, with probability one,
\begin{equation*} 
	Y = m(X,E_i) + s(X,E_i).
\end{equation*}
By totally differentiating $Q_{Y|X=x}(e)=Q_{Y(x)|X=x}(e)$ with respect to $x$, we note that for any $x>0$, it follows that we can identify
\begin{equation} \label{eq:quantilederivative}
	\frac{d}{dx} Q_{Y|X=x}(e) = m'(x,e)+s'(x,e),
\end{equation}
where we define $m'(x,e):=\frac{d}{dx} m(x,e)=g_1(x,x,e),$ and $s'(x,e):=\frac{d}{dx} s(x,e)=g_2(x,x,e)$. As above, $m'$ captures a causal effect, and $s'$ is the selection bias  term.

Let $m'(x):=\int_{0}^1 m'(x,e) de$ and $s'(x):=\int_{0}^1 s'(x,e) de$. Under regularity conditions, the expectation analog of Equation \eqref{eq:quantilederivative} satisfies:
\begin{equation*} 
	\frac{d}{dx} \mathbbm{E}[Y_i|X_i=x] = m'(x)+s'(x),
\end{equation*}
with $m'(x)=\mathbbm{E}[Y_i'(x)|X_i=x]$, where $Y_i'(x)=\frac{d}{dx} Y_i(x)$.

Recall that the parameter $\text{AME}_{\bar{x}}^+$ is well-defined given Assumption \ref{as:continuoussupport} from Section \ref{sec:approach}. However, we can give it an economic interpretation without assuming $X_i^*$ exists ex ante, by supposing that a subset of the bunchers having $X_i=0$ are ``marginal'', indicated by $M_i=1$ (e.g. they satisfy a FOC at $X_i=0$). Then, suppose that all of the marginal bunchers bunch (i.e. $P(X_i=0|M_i=1)=1$), and that the distribution of the slope $Y_i'(0^+)$ of the treatment response function $Y_i(x)$ as $x \downarrow 0$ is the same as for the non-bunchers that have very small values of $X_i$ (i.e. $Y_i'(0^+)|X_i=x \stackrel{d}{\rightarrow} Y_i'(0^+)|M_i=1$, 
where convergence in distribution is as $x \downarrow 0$). Our parameter of interest $\text{AME}_{\bar{x}}^+$ is then the average derivative effect of increasing $x$ from zero among marginal bunchers. This yields part (ii) of Assumption \ref{as:deconvregularity}.

To identify the function $f_{s(X)}(\cdot)$, we leverage the following additional assumptions.
\begin{assumption} \label{as:analyticxstar}
        There exist functions $s,\phi:\mathbb{R}\rightarrow \mathbb{R}$ such that
	\begin{enumerate}[(i)]
		\item For all $x>0$ and $e\in [0,1],$ $(x,e)\mapsto s(x,e) = s(x)+\phi(e)$. 
		\item For $x>0,$ $s(x)$ is an analytic function.
		\item $\phi(\cdot)$ is strictly increasing, and there exists a continuous solution $h$ satisfying the integral equation
		$$\int_{-\infty} ^\infty \phi^{-1}\left(y-s^*(x)\right)\cdot h(x) \cdot dx = F_{Y(0)}(y),$$
        where $s^*(x)$ is an analytic continuation of $s(\cdot)$ to the real line.
	\end{enumerate}
\end{assumption}
\noindent Recall that the function $s$ is defined as $s(x,e)=Q_{Y(0)|X=x}(e)$. To resolve an arbitrary additive normalization in the decomposition between $s$ and $\phi$, we let the $x$-dependent term $s(x)$ capture the conditional expectation of $Y_i(0)$ given $X_i=x$, as compared with its right limit at zero, i.e. $s(x) = \mathbbm{E}[Y_i(0)|X_i=x]-y_{\bar{x}}$ where $y_0:=\mathbbm{E}[Y_i(0)|X_i=x^+]$, recovering the function $s(x).$  Then, since
$$\mathbbm{E}[Y_i(0)|X_i=x]=\int_0^1 Q_{Y(0)|X=x}(e)\cdot de = \int_0^1 s(x,e) de = s(x) + \int_0^1 \phi(e) de,$$ this normalization implies that $\int_0^1 \phi(e) \cdot de = y_0$.

To motivate part (i) of Assumption \ref{as:analyticxstar}, observe the following 
\begin{proposition} \label{prop:additive}
    $Q_{Y(0)|X=x}(e)=s(x)+\phi(e)$ for some functions $s,\phi$ $\iff$ $\{Y_i(0) - \mathbbm{E}[Y_i(0)|X_i]\} \indep X_i$.
\end{proposition}
\noindent Thus part (i) of Assumption \ref{as:analyticxstar} is analagous to part (iii) of Assumption \ref{as:deconvregularity}, but for positive $X_i^*$.\\

The assumption of additive separability (i) cannot be directly verified from the data, because $s(x,e)=Q_{Y(0)|X=x}(e)$ is only identified in the limit as $x \downarrow 0$ and not for multiple values of $x$. For $x>0$, we can instead only identify $Q_{Y(x)|X=x}(e) = g(x,x,e)=m(x,e)+s(x,e)$. However, additive separability between $x$ and $e$ in $Q_{Y(x)|X=x}(e)$ for $x>0$ may be construed as indirect evidence in favor of part (i). For example, if $Y_i(x)-Y_i(0)$ is the same for all $i$ (homogeneous treatment effects), then additive separability of the observable quantile function $Q_{Y|X=x}(e)$ (i.e. $g(x,x,e)$) holds if and only if additive separability of $s(x,e)$ holds. We see this in Figure \ref{fig:condition_kdensities} in our application, where the distribution of $Y|X=x$ only differs substantially across different values of $x$ by a location shift: it is approximately normal with nearly identical variance across $x$. 

Item (ii) of Assumption \ref{as:analyticxstar} echoes Assumption \ref{as:analytic} from the main text, and amounts to imposing a high degree of smoothness on the function $x \mapsto \E[Y_i(0)|X_i=x]$. Item (iii) of Assumption \ref{as:analyticxstar} is high-level, but appears to be mild. The integral equation takes the form of a Fredholm Integral of the First Kind. As an analytic function, $s^*(x)$ is continuously differentiable, and if $s^*(x)$ is furthermore strictly monotonic with inverse $x(\cdot)$, then by a change of variables,
$$\int_{-\infty} ^\infty \phi^{-1}\left(y-t\right)\cdot \frac{h(x(t))}{s^{*'}(x(t))} \cdot dt = F_{Y(0)}(y),$$
which takes the form of a convolution equation, which typically has a solution.\\

Assumption \ref{as:analyticxstar} is useful because it allows us to define a new probability space in which $X_i$ is extended to take negative values, and such that $f_{s(X)}(\cdot)$ is identified through a deconvolution operation on that probability space. In particular, because the function $s: \mathbbm{R}^+ \rightarrow \mathbbm{R}$ is analytic on the positive part of the real line, there exists a unique function $s^*: \mathbbm{R}^+ \rightarrow \mathbbm{R}$ such that $s^*(x)=s(x)$ for all $x \ge 0$, referred to as the \textit{analytic continuation} of $s$. Now define $s^*(x,e) = s^*(x)+\phi(e)$.

The initial probability space has a measure $P$ defined over $(X_i,\{Y_i(x)\}_{x \ge 0})$. We now define a new probability measure $P^*$ over random variables $(X_i^*,Y_i^*(0))$ using $s^*(x,e)$. In particular, (1) we let the conditional distribution of $Y^*_i(0)$ given $X_i^*$ be described by the quantile function $Q^*_{Y^*(0)|X^*=x}(e) = s^*(x,e)$ for all $e \in [0,1]$ and $x \in \mathbbm{R}$; and (2)
 we let the marginal distribution of $X_i^*$ under $P^*$ be described by CDF $F^*_{X^*}(x)= \int_{\infty}^x h(x)$ for $x < 0$ and $H(x) = F_X(x)$ if $x \ge 0$, where $h$ is a solution to the equation in part (iii) of Assumption \ref{as:analyticxstar}.

Under this construction, the following holds with probability one under $P^*$: \begin{equation} \label{eq:starquantile}
	Y_i^*(0)=Q^*_{Y^*(0)|X^*}(E_i^*)=s^*(X_i^*)+\phi(E_i^*),
\end{equation}
where $E_i:=F^*_{Y(0)|X^*}(Y_i(0))$ satisfies $E_i^* \indep X_i^* $ with $E_i^* \sim \text{Unif}[0,1]$ under $P^*$. 

We note that the function $s^*(x,e)$ can be defined constructively from $s(x,e)$ provided that the radius of convergence of its Taylor series about $x=0$ is infinite. In this case, the Taylor series of $s^*(x)$ converges and $s_e(x)$ can be expressed as $s^*(x) = \sum_{k=0}^\infty \frac{x^k}{k!} \cdot s^{(k)}(0)$. However this stronger condition is not needed for what follows.

A key implication is that part (iii) of Assumption \ref{as:deconvregularity} holds with probability one under $P^*$: 
\begin{proposition} \label{prop:indep}
	Under Assumption \ref{as:analyticxstar}, $\{Y^*_i(0)-\mathbbm{E}^*[Y^*_i(0)|X_i^*]\} \indep X_i^*$ under $P^*$.
\end{proposition}
Finally, to unlock the ability to leverage Lemma \ref{lem:u(D*)} using part (iii) of Assumption \ref{as:analyticxstar}, we can use the following result:
\begin{proposition} \label{prop:bunchers}
	Given Assumption \ref{as:analyticxstar}, the distribution of $Y^*_i(0)|X_i^* \le 0$ under $P^*$ is the same as the distribution of $Y_i|X_i=0$ under $P$.
\end{proposition}

\noindent By Proposition \ref{prop:bunchers} and Equation \eqref{eq:starquantile}, we have that
$Y_i|X_i=0 \sim Y_i^*(0) | X^*_i \le 0 \sim s^*(X^*)+\phi(E_i^*) | X^*_i \le 0$. Thus the distribution of $s(X^*)+\phi(E_i^*)$ conditional on $X_i^* \le 0$ is pinned down from the observable distribution of $Y_i|X_i=0$.

Note furthermore that $Y_i|X_i=0^+ \sim s(0^+)+\phi(E_i^*)$, so $\phi(E_i^*) \sim \{Y_i-\mathbbm{E}[Y_i|X_i=0^+]\}|X_i=0^+$. Since $\phi(E_i^*)|X_i^*\le 0 \sim \phi(E_i^*)$, we have that $\phi(E_i^*)|X_i^*\le 0\sim \{Y_i-\mathbbm{E}[Y_i|X_i=0^+]\}|X_i=0^+$. It follows from Proposition \ref{prop:indep} that $\{(Y^*_i(0)-\mathbbm{E}^*[Y^*_i(0)|X_i^*]) \indep X_i^*\}|X_i^* \le 0$, so we can then work out the distribution of $s(X_i^*)$ conditional on $X_i^* \le 0$ via deconvolution as in Section \ref{sec:idfrombunching}.

\newpage
\setcounter{page}{1}
\vspace*{.3cm}

\begin{center}
    {\LARGE \textbf{SUPPLEMENTARY APPENDICES TO}
    \\[2ex]
    \textbf{``Identification of Causal Effects with a Bunching Design''}}\\[1cm]
    {\large \hspace{-.8cm}Carolina Caetano, \hspace{.8cm}Gregorio Caetano, \hspace{1cm}Leonard Goff, \hspace{1.5cm}Eric Nielsen}\\
    {\large \textit{University of Georgia \hspace{.1cm} University of Georgia \hspace{.1cm}University of Calgary \hspace{.1cm}Federal Reserve Board}}\\[.8cm]
    {\large \today}\\
\end{center}

\vspace{.3cm}
\begin{center}
\begin{minipage}{0.8\textwidth}
\small
\begin{center}
    \textbf{Abstract}
\end{center}

Appendix \ref{sec:controls} explains how the method may be adapted to leverage control variables, including when controls are discrete (Appendix \ref{sec:discretecontrols}), continuous (Appendix \ref{sec:continuouscontrols}), and when the vector of controls is large and with mixed continuous and discrete variables (Appendix \ref{sec:discretecontinuouscontrols}). Appendix \ref{ap:emp_supp} contains additional figures referenced in Section \ref{sec:application}. Finally, Appendix \ref{sec:xstarproofs} contains proofs of several results presented in Appendix \ref{sec:xstarappendix} in the paper. 

\end{minipage}
\end{center}

\vspace{.5cm}

\renewcommand{\thesection}{S.\arabic{section}}
\renewcommand{\thesubsection}{S.\arabic{section}.\arabic{subsection}}
\renewcommand{\theequation}{S.\arabic{equation}}
\setcounter{equation}{0}
\renewcommand{\theHequation}{S\arabic{equation}}
\setcounter{section}{0} % ensure starts at S.1
\renewcommand{\theHsection}{S\arabic{section}} 

\phantomsection
\section{Incorporating control variables}
\addcontentsline{toc}{section}{\thesection: Incorporating control variables}
\label{sec:controls}

Suppose that in addition to $X_i$ and $Y_i$, we observe a vector of control variables $Z_i$ that are unaffected by the treatment. Define the conditional $\text{AME}^+_{\bar{x}}$ as 
$$\text{AME}_{\bar{x}}^+(z):=\lim_{x\downarrow \bar{x}}\E\left[\frac{Y_i(x)-Y_i(\bar{x})}{x-\bar{x}}\bigg|X_i=x,Z_i=z\right].$$
Then, Assumptions \ref{as:continuoussupport}-\ref{as:deconvregularity} may be made conditional on $Z_i,$ and the following identification equation holds. 
\begin{align*}
\text{AME}_{\bar{x}}^+(z)&=m'(\bar{x}^+,z)-2\pi \cdot \theta(z)\cdot \left( \int \frac{\E[e^{\textbf{i}\xi Y_i}|X_i=\bar{x},Z_i=z]}{\E[e^{\textbf{i}\xi Y_i}|X_i=\bar{x}^+,Z_i=z]}d\xi\right)^{-1}\cdot\frac{f_{X|Z=z}(\bar{x}^+)}{F_{X|Z=z}(\bar{x})}, \label{eq:u'conditional}
\end{align*}
where  $m'(\bar{x}^+,z)=\lim_{x\downarrow \bar{x}}\frac{d}{dx}\E[Y_i|X_i=x,Z_i=z],$ and $\theta(z)=\text{sgn}(\E[Y_i|X_i=\bar{x}^+,Z_i=z]-\E[Y_i|X_i=\bar{x},Z_i=z])$. 

Note then that we can identify the unconditional $\text{AME}_{\bar{x}}^+,$ given that
$$\text{AME}_{\bar{x}}^+ = \lim_{x\downarrow \bar{x}} \E\left[\E\left[\frac{Y_i(x)-Y_i(\bar{x})}{x-\bar{x}}\bigg|X_i=x,Z_i\right]\right] = \E\left[\text{AME}_{\bar{x}}^+(Z_i)|X_i=\bar{x}^+\right],$$
under suitable conditions to interchange the limit and the expectation.

\phantomsection
\subsection{Estimation with discrete controls}
\label{sec:discretecontrols}
Estimation with controls depends on the nature of $Z_i.$ If $Z_i$ has a finite support, i.e. $Z_i\in \{z_1,\dots,z_L\}$ with $\mathbb{P}(Z_i=z_l)>0,$ for all $l=1,\dots,L$, then the exact procedures described for the unconditional case may be performed separately for each $z_l$.  That is, for all $z_l,$ calculate
\begin{equation*}
    \hat{p}_{l,\bar{x}}=\hat{\mathbb{P}}(Z_i=z_l|X_i=\bar{x})=\hat{F}_X(\bar{x})^{-1}\cdot \frac{1}{n}\sum_{i=1}^n \bm{1}(X_i=\bar{x},Z_i=z_l), \text{ for }l=1,\dots,L.
\end{equation*}
Then, for all $z_l$ such that $\hat{p}_{l,\bar{x}}>0$, restrict the sample to observations such that $Z_i=z_l,$ and estimate $\text{AME}_{\bar{x}}^+(z_l)$ just as described in the unconditional case using the new, restricted, data. 

The average marginal treatment effect estimator in this case is
\begin{equation*}
    \text{AME}_{\bar{x}}^+=\sum_{l=1}^L \hat{p}_{l,\bar{x}} \cdot \widehat{\text{AME}}_{\bar{x}}^+(z_l),
\end{equation*}
where $\widehat{\text{AME}}_{\bar{x}}^+(z_l)$ is the estimator described in Section \ref{sec:estimation} applied to the subsample with $Z_i=z_l$. Note that it is not possible to estimate $\text{AME}_{\bar{x}}^+(z_l)$ when $\hat{p}_{l,\bar{x}}=0,$ but it is also not necessary to do so, since those treatment effects have weight equal to zero in the estimator formula.

\phantomsection
\subsection{Estimation with continuous controls}
\label{sec:continuouscontrols}
When $Z_i$ is continuously distributed, one may apply a smoothing technique to the estimators described above, so as to use information coming from values of the control around $Z_i$ to perform the estimation. A simple strategy  to estimate $\text{AME}_{\bar{x}}^+(Z_i)$ is as follows: let $Z_i=(Z_{1i},\dots,Z_{Mi})',$ and for bandwidths $\kappa_1,\dots,\kappa_M,$ and kernel functions  $K_1,\dots,K_M$, restrict the sample to observations such that $-\kappa_1<Z_{1j}<\kappa_1, \dots,  -\kappa_M<Z_{Mj}<\kappa_M$. Index the resulting dataset by $t$,  suppose it has $n_T$ observations, and define
\begin{equation*}
    K_\kappa(Z_t-Z_i):=\frac{1}{\kappa_1\cdot \cdot \cdot  \kappa_M}K_1\left(\frac{Z_{1t}-Z_{1i}}{\kappa_1}\right) \cdot \cdot \cdot  K_M\left(\frac{Z_{Mt}-Z_{Mi}}{\kappa_M}\right).
\end{equation*}
Then, for each value $Z_i$ such that the restricted sample has bunching, i.e.
\begin{equation*}
   \hat{p}_{i,\bar{x}}= \frac{1}{n_T}\sum_{t=1}^{n_T}\bm{1}(X_t=\bar{x})>0,
\end{equation*}
perform the methods described
for unconditional estimation, only weighting each observation by $k_{\kappa}(Z_t-Z_i)=K_\kappa(Z_t-Z_i)/\sum_{t=1}^TK_\kappa(Z_t-Z_i).$\footnote{Thus, $\hat{F}_{X|Z=Z_i}(\bar{x})=\frac{1}{n_T}\sum_{t=1}^{n_T}\bm{1}(X_i=\bar{x})k_{\kappa}(Z_t-Z_i)$, and $\hat{E}[Y_i|X_i=\bar{x},Z_i]=\hat{F}_{X|Z=Z_i}(\bar{x})^{-1}\cdot \frac{1}{n_T}\sum_{t=1}^{n_T}Y_i\bm{1}(X_i=\bar{x})k_{\kappa}(Z_t-Z_i).$ The densities $\hat{f}_{X|Z=Z_i}(\bar{x}^+)$ and $\hat{f}_{Y|X=\bar{x},Z_i}(Y_i)$ are implemented in the same way using the restricted sample, substituting $i$ by $t$ and $n$ by $n_T$ in the formulas, and multiplying terms inside sums indexed by $t$ by $k_{\kappa}(Z_t-Z_i)$. Finally, $\hat{\E}[Y_i|X_i=\bar{x}^+,Z_i]$ and $\hat{m}'(\bar{x}^+,Z_i)$ are respectively the intercept and slope coefficients of a local linear regression of $Y_t$ onto $X_t$ at zero, using only observations such that $X_t>\bar{x}$ and weights $k_{\kappa}(Z_t-Z_i)$; and $\hat{\E}[\hat{f}_{Y|X=\bar{x},Z=Z_i}(Y_i)|X_i=\bar{x},Z_i]$ is the intercept of the same procedure, only with $\hat{f}_{Y|X=\bar{x},Z=Z_i}(Y_i)$ instead of $Y_i$. }

Then,
\begin{equation*}
    \text{AME}_{\bar{x}}^+=\sum_{l=1}^L \hat{p}_{i,\bar{x}} \cdot \widehat{\text{AME}}_{\bar{x}}^+(Z_i).
\end{equation*}
As in the previous section, it is not necessary to compute $\widehat{\text{AME}}_{\bar{x}}^+(Z_i)$ when $\hat{p}_{i,\bar{x}}=0.$

\phantomsection
\subsection{Estimation with mixed or large dimensional controls}
\label{sec:discretecontinuouscontrols}
In practice, most control lists include a mixture of discrete and continuous variables, and may include a large number of terms. In such cases, smoothing is either impractical or impossible. We have had success with a discretization technique which implements clustering methods, which are popular in machine learning and have been recently adopted in economics.\footnote{See, e.g. \citet{Bonhomme_Manresa_ECMA,Bonhomme_Manresa_Lamadon,cheng2019clustering,cytrynbaum2020blocked, ccn_metrics, ccn_empirical, CCNS}.}

Let $\{\hat{\mathcal{C}}_1,\dots,\hat{\mathcal{C}}_C\}$ be a finite partition of the observations into groups, which we call clusters, and let  $\hat{C}_{i}=(\bm{1}(Z_i\in \hat{\mathcal{C}}_1),\dots,\bm{1}(Z_i\in \hat{\mathcal{C}}_{C}))'$ be the cluster indicators. We propose substituting $Z_i$ with $\hat{C}_{i},$ which has finite support. This then transforms the estimation procedure into a discrete controls case, which can be implemented exactly as described in the Supplementary Appendix Section \ref{sec:discretecontrols}. 

Explicitly, for each cluster $
\mathcal{C}_c,$ calculate
\begin{equation*}
    \hat{p}_{c,\bar{x}}=\hat{F}_X(\bar{x})^{-1}\cdot \frac{1}{n}\sum_{i=1}^n \bm{1}(X_i=\bar{x},Z_i\in \mathcal{C}_c), \text{ for }c=1,\dots,C.
\end{equation*}
Then, for those clusters with $\hat{p}_{c,\bar{x}}>0,$
estimate $\widehat{\text{AME}}_{\bar{x}}^+(\hat{\mathcal{C}}_c)$ separately using a new dataset composed only of observations within cluster $\mathcal{C}_c$ (i.e.  $i$ such that  $Z_i\in \mathcal{C}_c$). For this, follow the exact procedures described in the unconditional case. The average marginal treatment effect estimator is, then,
\begin{equation*}
    \widehat{\text{AME}}_{\bar{x}}^+=\sum_{c=1}^C \hat{p}_{c,\bar{x}} \cdot \widehat{\text{AME}}_{\bar{x}}^+(\hat{\mathcal{C}}_c).
\end{equation*}
As in the previous sections, it is not necessary to estimate $\text{AME}_{\bar{x}}^+(\hat{\mathcal{C}}_c)$  when $\hat{p}_{c,\bar{x}}=0.$

In general, if $\text{AME}_{\bar{x}}^+(z)$ is continuous in $z$, the ability of this estimator to approximate $\text{AME}_{\bar{x}}^+(z)$ depends on how much information about $Z_i$ is given by the cluster indicator vector $\hat{C}_{C}.$ Thus, it is desirable to choose a clustering method that minimizes the within-cluster variation in the values of $Z_i.$ All unsupervised clustering methods in the statistical learning literature could in principle be used (e.g. k-means, k-medoids, self-organizing maps, and spectral -- see \cite{hastie2009elements}). If feasible, we recommend using hierarchical clustering for its well-known stability.\footnote{Hierarchical clustering requires the choice of a linkage method and a dissimilarity measure. We recommend using Ward's linkage and the Gower measure for mixed continuous and discrete controls.}

The clustering strategy requires the choice of the number of clusters, which modulates the bias-variance trade-off in the estimation of $\text{AME}_{\bar{x}}^+(z)$. The more clusters are used, the more similar are the $Z_i$ within each cluster, and thus the smaller the bias and the larger the variance. Although there must exist an optimal number of clusters, there are as yet no established methods to aid with this decision. 

Nevertheless, note that we are not directly interested in $\text{AME}_{\bar{x}}^+(z)$ but rather in $\widehat{\text{AME}}_{\bar{x}}^+,$ which aggregates the information over all clusters. The trade-off is, in theory, much less important for $\widehat{\text{AME}}_{\bar{x}}^+,$ and thus one should err on the side of having a larger number of clusters, with an eye for instability which could be created by pathological clusters (e.g. clusters with bunching but with too few observations near the bunching point, or clusters where every observation is bunched).

\phantomsection
\section{Supplemental Empirical Plots}
\label{ap:emp_supp}

\begin{figure}[H]
\protect\caption{Birth Weight QQ Plots Conditional on Maternal Smoking\label{fig:qq_plots}}
\vspace{-.2in}
\begin{center}
\includegraphics[scale=0.58]{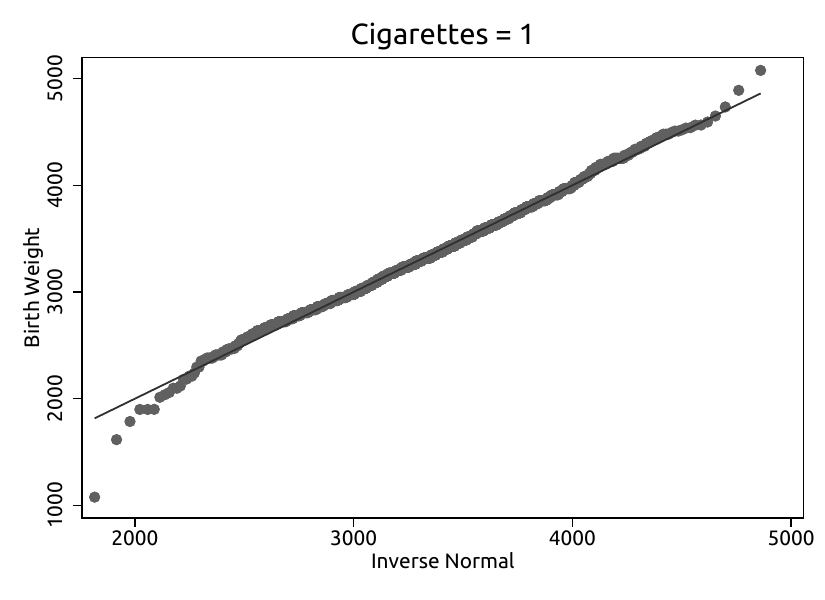}
\includegraphics[scale=0.58]{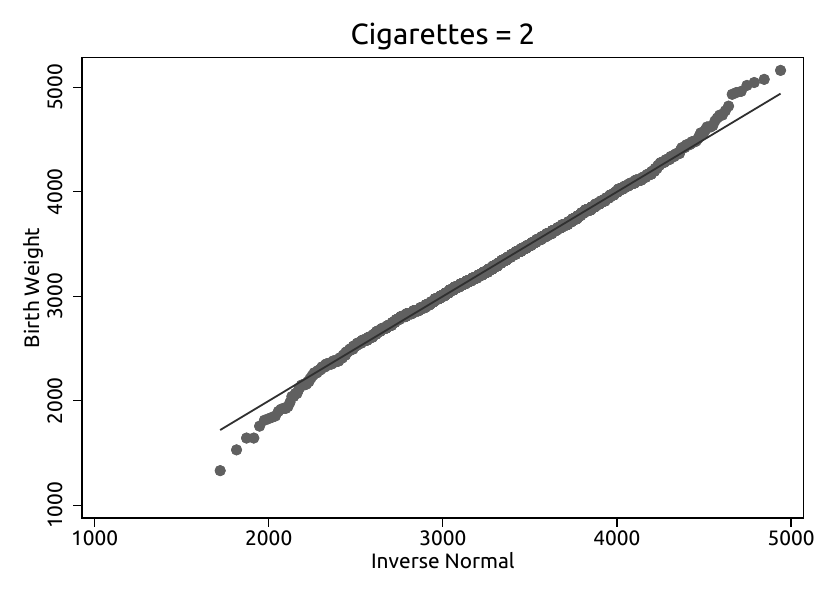}
\includegraphics[scale=0.58]{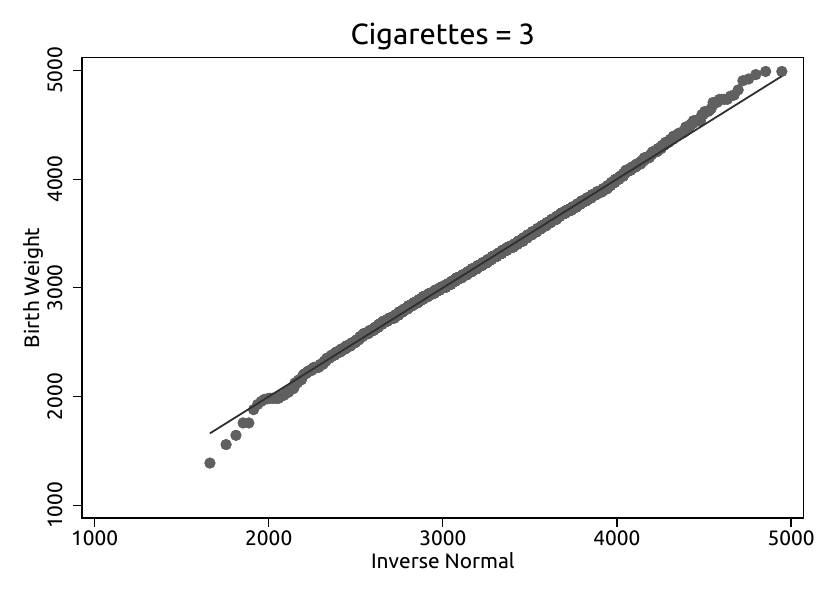}
\includegraphics[scale=0.58]{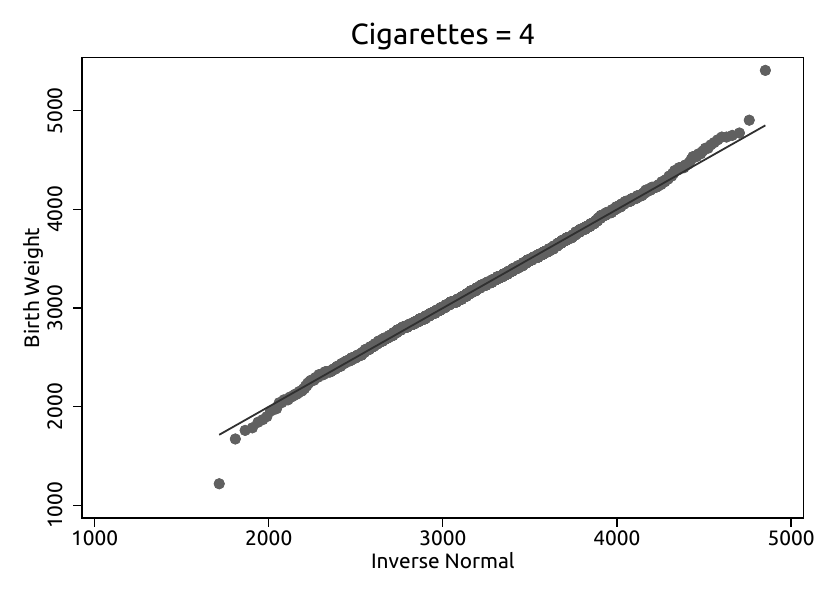}
\includegraphics[scale=0.58]{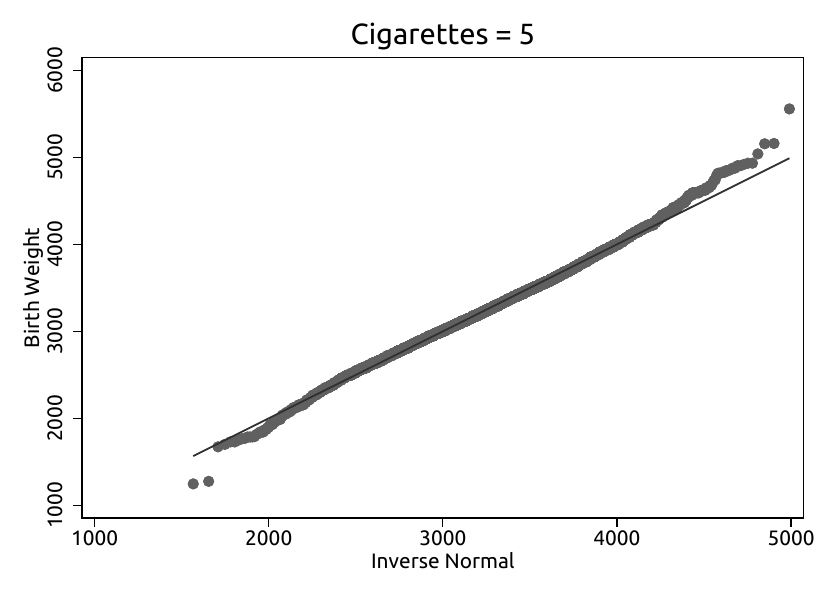}
\end{center}
\vspace{-0.2in}
\singlespace \footnotesize{Note: Each panel plots the estimated QQ plot for birth weight conditional on a different value of maternal cigarettes smoked per day. The closer the plot is to the 45-degree line, the closer to normal is the conditional birth weight distribution. Data taken from \cite{almond2005costs}.}
\end{figure}

\phantomsection
\section{Proofs for Appendix \texorpdfstring{\ref{sec:xstarappendix}}{B}} 
\label{sec:xstarproofs}

\subsection{Proof of Proposition \texorpdfstring{\ref{prop:transformationexistence}}{B.1}}
    We can write the optimization problem as
    $$X(\theta):= \text{argmax}_x\{V(x,W-px;\theta) \text{ subject to } x \ge 0, r\ge 0\}.$$
    Given the convexity of $V$, we have by the first order condition that $MRS(X(\theta),\theta) = 1$ if and only if there is an interior maximum $X(\theta) \in (0,W)$. In this case, differentiating with respect to $\theta$,
    $$ MRS_x(X(\theta),\theta)\cdot X'(\theta) = -MRS_{\theta}(X(\theta),\theta),$$
    where $X'(\theta)$ exists by the implicit function theorem. By assumption $MRS_{\theta}(x,\theta):=\frac{\partial}{\partial \theta} MRS(x,\theta) > 0$ for all $x \in (0,\bar{W}), \theta \in \Theta$. Since $MRS_x(x,\theta):=\frac{\partial}{\partial x} MRS(x,\theta) < 0$ for all $x \in (0,W), \theta \in \Theta$, we must then have that $X'(\theta)>0$ given that $X(\theta) \in (0,W)$: optimal choices are strictly increasing in $\theta$ for any $\theta$ such that $0< X(\theta) < W$. Let $\bar{\theta} = \lim_{x \downarrow 0} X^{-1}(x)$. If $\theta = \bar{\theta}$, then $X(\theta)=0$ by continuity and if $\theta < \bar{\theta}$, we must have $X(\theta(\rho_i)) \notin (0,W)$.
    
    Meanwhile, since $MRS_{\theta}(x,\theta)>0$, $MRS(x,\theta)$ is strictly increasing in $\theta$ for each $x \ge 0$, the function $\rho(\theta):=MRS(0,\theta)$ is also strictly increasing in $\theta$. Let $\theta(\rho)$ be the inverse function of $\rho(\theta)$, which is also strictly increasing and differentiable in $\rho$ (given that these properties hold for $\rho(\cdot)$). 
   Given the assumption that $P(X_i=W)=0$, we have, combining cases, that with probability one:
    $$X_i = \begin{cases}
        0 &\text{ if } \theta(\rho_i) \le \bar{\theta}\\
        X(\theta(\rho_i)) & \text{ if } \theta(\rho_i) > \bar{\theta}
    \end{cases} \;\;= \;\max\{X(\theta(\rho_i)), 0\}.$$
    The result of the Proposition then holds with the strictly increasing and differentiable function $h(\rho) = X(\theta(\rho))$.
       
\subsection{Proof of Proposition \ref{prop:transformation}}
    Rewriting Equation \eqref{eq:xbarmax} as $X_i = \max\{h(\rho_i),h(\bar{\rho})\}$ where $\bar{\rho} = h^{-1}(\bar{x})$, we have, given Theorem \ref{thm:AMEid}, that    \begin{equation}\label{eq:u'transformation}
    \widetilde{\text{AME}}_{\bar{\rho}}^+=\lim_{\rho\downarrow \bar{\rho}}\frac{d}{d\rho}\E[Y_i|\rho_i=\rho]-\frac{\text{sgn}\left(\E[Y_i|\rho_i=\bar{\rho}^+]-\E[Y_i|\rho_i=\bar{\rho}]\right) \cdot f_\rho(\bar{\rho}^+)}{F_\rho(\bar{\rho})\cdot \frac{1}{2\pi}\int \frac{\E[e^{\textbf{i}\xi Y_i}|\rho_i=\bar{\rho}]}{\E[e^{\textbf{i}\xi Y_i}|\rho_i=\bar{\rho}^+]}d\xi},
    \end{equation}
    where $\widetilde{\text{AME}}_{\bar{\rho}}^+:=\lim_{\rho\downarrow \bar{\rho}}\E[(\tilde{Y}_i(\rho)-\tilde{Y}_i(\bar{\rho}))/(\rho-\bar{\rho})|\rho_i=\rho]$, and $\tilde{Y}_i(\rho)$ indicate potential outcomes with respect to $\rho$: $\tilde{Y}_i(\rho) = Y_i(h(\rho))$ for any $\rho \ge \bar{\rho}$. Equivalently, $Y_i(x) = \tilde{Y}_i(h^{-1}(x))$ for any $x \ge \bar{x}$.

    First, notice that under the maintained assumptions
    \begin{align*}
    \widetilde{\text{AME}}_{\bar{\rho}}^+ &= \E[\tilde{Y}'(\bar{\rho})|\rho_i = \bar{\rho}] = \E\left[\left.\frac{d}{d\rho}\tilde{Y}(\bar{\rho})\right|\rho_i = \bar{\rho}\right] =  \E\left[\left.\left.\frac{d}{d\rho}Y_i(h(\rho))\right|_{\rho=\bar{\rho}}\right|X_i^* = \bar{x}\right]\\%\bar{\rho}] = \E\left[\left.\left.\frac{dx}{d\rho}\frac{d}{dx}\tilde{Y}(h^{-1}(x))\right|_{x=h(\bar{\rho})}\right|\rho_i = \bar{\rho}\right]\\
    &=h'(\bar{\rho}) \cdot \E[Y_i'(\bar{x})|X^*_i = \bar{x}]=h'(\bar{\rho}) \cdot \text{AME}_{\bar{x}}^+,
    \end{align*}
using the chain rule and defining $X_i^* = h(\rho_i)$. Meanwhile, all of the terms on the right-hand side of \eqref{eq:u'transformation} remain invariant under the substitution of $\rho_i$ by $X_i^*$ and $\bar{\rho}$ by $\bar{x}$, except that, similarly,
$$\lim_{\rho\downarrow \bar{\rho}}\frac{d}{d\rho}\E[Y_i|\rho_i=\rho]=\lim_{\rho\downarrow \bar{\rho}}\left.\frac{dx}{d\rho}\frac{d}{dx}\E[Y_i|X_i=x]\right|_{x=\bar{\rho}} = h'(\bar{\rho}) \cdot \lim_{\rho\downarrow \bar{\rho}}\frac{d}{dx}\E[Y_i|X_i=x], $$
while also $f_\rho(\bar{\rho}^+) = h'(\bar{\rho}) \cdot f_X(\bar{x}^+)$. Equation \eqref{eq:u'transformation} thus implies $\eqref{eq:u'}$, provided that $h'(\bar{\rho}) \ne 0$.

\subsection{Proof of Proposition \ref{prop:additive}}
In one direction, observe that if $Q_{Y(0)|X=x}(e) = s(x)+\phi(e)$, then $Y_i(0) - s(X_i) = Q_{Y(0)|X_i}(E_i)-s(X_i)=\phi(E_i)+s(X_i)-s(X_i) = \phi(E_i)$, and $E_i \indep X_i$. To see the other direction, define $s(x):=\mathbbm{E}[Y_i(0)|X_i=x]$ and $\mathcal{E}_i := Y_i(0)-s(X_i)$. Then, by definition, we can write $Q_{Y(0)|X_i}(E_i)=Y_i(0)=s(X_i)+\mathcal{E}_i = s(X_i)+Q_{\mathcal{E}_i|X_i}(\mathcal{E}_i)$. The condition $(Y_i(0)-\mathbbm{E}[Y_i(0)|X_i]) \indep X_i$ implies that $\mathcal{E}_i \indep X_i$, so we can replace $Q_{\mathcal{E}_i|X_i}(\mathcal{E}_i)$ with the unconditional $Q_{\mathcal{E}_i}(\mathcal{E}_i)$. Note that for any $e \in [0,1]$, $Q_{Y(0)|X=x}(e) = s(x) + Q_{\mathcal{E}_i}(e),$ since $Q_{\mathcal{E}_i}(\cdot)$ is strictly increasing. Now define $\phi(e) = Q_{\mathcal{E}_i}(e)$.

\subsection{Proof of Proposition \ref{prop:indep}}
Given \eqref{eq:starquantile}, we have that, with probability one under $P^*,$ $$Y^*_i(0)-\mathbbm{E}^*[Y^*_i(0)|X_i^*] = \{s^*(X_i^*)+\phi(E_i^*)\} - s^*(X_i^*)=\phi(E_i^*).$$
Then, since $E_i^* \sim \text{Unif}[0,1]$ under $P^*,$ and $Y^*_i(0)-\mathbbm{E}^*[Y^*_i(0)|X_i^*]$ is a measurable function of $E_i^*,$ it follows that $\{Y^*_i(0)-\mathbbm{E}^*[Y^*_i(0)|X_i^*]\} \indep X_i^*$ under $P^*$.

\subsection{Proof of Proposition \ref{prop:bunchers}}
	\begin{align*}
	P^*(Y^*_i(0)\le y|X_i^* \le 0) &= \frac{1}{P^*(X_i^* \le 0)}\cdot \int_{-\infty}^0 dx \cdot h(x)\cdot P^*(Y^*_i(0)\le y|X_i^* =x)\\
	&= \frac{1}{P^*(X_i^* \le 0)}\cdot \int_{-\infty}^0 dx \cdot h(x)\cdot P^*(s^*(x)+\phi(E_i^*)\le y|X_i^* =x)\\
	&= \frac{1}{P^*(X_i^* \le 0)}\cdot \int_{-\infty}^0 dx \cdot h(x)\cdot P^*(E_i^*\le \phi^{-1}(y-s^*(x)))\\
	&= \frac{1}{P^*(X_i^* \le 0)}\cdot \int_{-\infty}^0 dx \cdot h(x)\cdot P^*(E_i^*\le \phi^{-1}(y-s^*(x)))\\
	&=\frac{1}{P(X_i = 0)}\cdot \int_{-\infty}^0 dx \cdot h(x)\cdot  \phi^{-1}(y-s^*(x)),
	%P(Y_i \le y|X_i=0)
\end{align*}
using that $E_i^* \indep X_i^*,$ and that $E_i^* \sim \text{Unif}[0,1]$.

By the same steps, and using that $h(x)=f_X(0)$ and $Q_{Y(0)|X=x}(e)=Q^*_{Y^*(0)|X^*=x}(e)$ for all $x>0$,
\begin{align*}
	P(Y_i(0)\le y|X_i > 0) &= \frac{1}{P(X_i > 0)}\cdot \int_0^{\infty} dx \cdot f_X(x)\cdot P(Y_i(0)\le y|X_i =x)\\
	&= \frac{1}{1-P(X_i = 0)}\cdot\int_0^{\infty} dx \cdot h(x)\cdot P^*(s^*(x)+\phi(E_i^*)\le y|X_i^* =x)\\
	&=\frac{1}{1-P(X_i = 0)}\cdot \int_0^{\infty} dx \cdot h(x)\cdot  \phi^{-1}(y-s^*(x)).
	%P(Y_i \le y|X_i=0)
\end{align*}
Using both of these results, we have, by part (iii) of Assumption~\ref{as:analyticxstar}, that
\[
\begin{aligned}
F_{Y(0)}(y) 
&= \int_{-\infty}^0 \phi^{-1}\left(y - s^*(x)\right) \cdot h(x)\, dx 
+ \int_{0}^\infty \phi^{-1}\left(y - s^*(x)\right) \cdot h(x)\, dx \\
&= P(X_i = 0) \cdot P^*(Y^*_i(0) \le y \mid X_i^* \le 0) 
+ (1 - P(X_i = 0)) \cdot P(Y_i(0) \le y \mid X_i > 0).
\end{aligned}
\]
Meanwhile, by the law of iterated expectations, we also have that
$$F_{Y(0)}(y) = P(X_i = 0) \cdot P(Y_i(0)\le y|X_i = 0)+ (1-P(X_i = 0))\cdot P(Y_i(0)\le y|X_i > 0),$$
and therefore $P(Y_i\le y|X_i = 0) = P(Y_i(0)\le y|X_i = 0) = P^*(Y^*_i(0)\le y|X_i^* \le 0)$ for all $y$.

\end{document}